%% file: main.tex
\documentclass{article}
\usepackage[margin=1in]{geometry}

\input{preamble}

\begin{document} 

\title{Integral Online Algorithms for  Set Cover  and \\    Load Balancing  with Convex Objectives}

	\author{Thomas Kesselheim\thanks{
	(thomas.kesselheim@uni-bonn.de)
 Institute of Computer Science,	    University of Bonn. 
    }
	\and Marco Molinaro\thanks{
         (mmolinaro@microsoft.com)
         Microsoft Research and PUC-Rio.   
Supported in part by the Coordenação de Aperfeiçoamento de Pessoal de Nível Superior - Brasil (CAPES) - Finance Code 001, and by Bolsa de Produtividade em Pesquisa $\#3$02121/2025-0 from CNPq.
    }
    \and Kalen Patton\thanks{
        (kpatton33@gatech.edu)
        School of Computer Science,
        Georgia Tech.
        Supported in part by NSF awards CCF-2327010 and CCF-2440113, and by the Georgia Tech ARC-ACO Fellowship
        }
	\and Sahil Singla\thanks{
        (ssingla@gatech.edu)
        School of Computer Science,
        Georgia Tech.
        Supported in part by NSF awards CCF-2327010 and CCF-2440113.
        }
}

\maketitle

\begin{abstract}
\medskip

\noindent Online Set Cover and Load Balancing are central problems in online optimization, and there is a long line of work focusing on developing  algorithms for these problems with convex objectives. Although we know optimal online algorithms with  $\ell_p$-norm objectives,  recent developments for general norms and convex objectives that rely on the online primal-dual framework apply only to \emph{fractional} settings due to large integrality gaps.

\medskip
\noindent  Our work focuses on directly designing integral online algorithms for Set Cover and Load Balancing with convex objectives,  bypassing the convex-relaxation and the primal-dual technique.  
Some of the main implications of our approach are:

\smallskip
    \begin{enumerate}
        \item For Online Set Cover, we can extend  the results of  \cite{AzarBCCCG0KNNP16} for convex objectives and of \cite{KMS-STOC24} for symmetric norms from fractional to integral settings.

        \item  Our results for convex objectives and symmetric norms even apply to the Online Generalized Scheduling Problem, which generalizes both Set Cover and Load Balancing.  Previous works could only handle the offline version of this problem with norm objectives \cite{DLR-SODA23}.

        \item Our approach easily extends to settings involving disjoint-composition of norms. This allows us to recover or improve the norm-composition results of \cite{NS-ICALP17,KMS-STOC24} and extend our results to a large class of norms beyond the symmetric setting.
        
    \end{enumerate}

    \medskip
    
  \noindent  Our approach involves first reducing these online problems to online packing problems, and to then design good approximation algorithms for the latter.
  To solve these packing problem, we use two key ideas. First, we decouple the global packing problem into a series of local packing problems on different machines. Second, we choose random activation thresholds for machines such that  conditional on a machine being activated the expected number of jobs it covers is high compared to its cost.
    This approach may be of independent interest and could find applications to other online problems. 
\end{abstract}

\medskip

\newpage
\section{Introduction}
    \input{intro}

\section{Warm-up: A New Proof for Online Set Cover}\label{sec:warmup}
    \input{warm-up}

\section{Reducing Online Generalized Scheduling to a Packing Problem}
\label{sec:reduction}
    \input{reduction}

\section{Convex Aggregate Function with Monotone Gradients} \label{sec:p-bound}
    \input{p-bounded}

\section{Norm Aggregate Function} \label{sec:norm}
    \input{norms}

\vspace{1cm}
\appendix

{\LARGE \bf Appendix}
    \input{appendix}

\begin{small}
\bibliographystyle{alpha}
\bibliography{bib,online-lp-short-coped}
\end{small}

\end{document}

%% file: preamble.tex
\usepackage[utf8]{inputenc}
\usepackage{amsmath, amsfonts, amssymb, amsthm, dsfont}
\usepackage{babel}
\usepackage{parskip}
\usepackage{enumerate}
\usepackage[normalem]{ulem} 
\usepackage{graphicx}
\usepackage{mathtools}
\usepackage{xspace}
\usepackage{comment}
\usepackage{bm}
\usepackage{soul}

\usepackage[usenames,dvipsnames]{xcolor}
\usepackage[backref=page,linktocpage=true,breaklinks,colorlinks,citecolor=PineGreen,linkcolor=Violet]{hyperref}
\usepackage{cleveref}

\usepackage{algorithm}
\usepackage[noend]{algpseudocode}

\DeclarePairedDelimiter\ceil{\lceil}{\rceil}
\DeclarePairedDelimiter\floor{\lfloor}{\rfloor}

\newcommand{\cA}{\mathcal{A}}
\newcommand{\cF}{\mathcal{F}}
\newcommand{\cG}{\mathcal{G}}
\newcommand{\R}{\mathbb{R}}
\newcommand{\Rp}{\mathbb{R}_{\geq 0}}

\newcommand{\E}{\mathbb{E}}
\newcommand{\Var}{\textup{Var}}

\newcommand{\ind}{\mathds{1}}

\newcommand{\Opt}{\textnormal{OPT}}
\newcommand{\OptBSP}{\textnormal{OPT}_{\textnormal{BSP}}}
\newcommand{\Alg}{\textnormal{Alg}}

\newcommand{\ones}{\bm{1}}
\newcommand{\ba}{\mathbf{a}}
\renewcommand{\d}{\mathrm{d}}
\newcommand{\cI}{\mathcal{I}}

\newtheorem{theorem}{Theorem}[section]
\newtheorem{lemma}[theorem]{Lemma}

\newtheorem{corollary}[theorem]{Corollary}

\newtheorem{claim}[theorem]{Claim}
\newtheorem{cor}[theorem]{Corollary}
\newtheorem{definition}[theorem]{Definition}

\newcommand{\vertm}[1]{{\left\vert\kern-0.25ex\left\vert\kern-0.25ex\left\vert #1     \right\vert\kern-0.25ex\right\vert\kern-0.25ex\right\vert}}

\newcommand{\topk}{\ensuremath{\textrm{Top-k}}\xspace}

\newcommand{\overbar}[1]{\mkern 1.5mu\overline{\mkern-1.5mu#1\mkern-1.5mu}\mkern 1.5mu}

\newcounter{note}[section]

\newcommand{\knote}[1]{\refstepcounter{note}$\ll${\bf Kalen~\thenote:}
  {\sf \color{blue}  #1}$\gg$\marginpar{\tiny\bf KP~\thenote}}

\newcommand{\mnote}[1]{\refstepcounter{note}$\ll${\bf Marco~\thenote:}
  {\sf \color{blue}  #1}$\gg$\marginpar{\tiny\bf MM~\thenote}}

\newcommand{\red}[1]{{\color{red} #1}}
\newcommand{\blue}[1]{{\color{blue} #1}}

\newcommand{\OPT}{\textup{OPT}}

%% file: intro.tex
Online Set Cover and Online Load Balancing are central problems in the theory of online optimization, capturing a broad range of settings in which requests arrive one by one, and each must be immediately satisfied while minimizing an objective function. In classic formulations, this objective function has a simple structure, such as the sum of costs in the case of Set Cover \cite{AlonAA-STOC03} and the maximum load on a machine in the case of Load Balancing \cite{AspnesAFPW-JACM97}. We know tight (poly)logarithmic approximation algorithms for both these problems. However,  applications of these problems often require performance guarantees for more general objectives. As a result, there is a long line of work in developing algorithms for these problems that can handle a broader class of convex objectives. In particular, both these classic  problems have been studied with objectives that are $\ell_p$-norms  \cite{awerbuch,AzarBCCCG0KNNP16}, Top-$k$\footnote{Top-$k$ norm of a non-negative vector is the sum of its largest $k$ coordinates.} and 
 symmetric\footnote{A norm $\|\cdot\| : \R^m \rightarrow \R$ is \emph{symmetric} if for every vector $x$, $\|x\|$ is the same for every permutation of the coordinates of $x$.} norms \cite{CS-STOC19,KMS-SODA23,KMS-STOC24}, sums of norms \cite{NS-ICALP17,KMS-STOC24}, and even general convex functions \cite{AzarBCCCG0KNNP16}.

For such convex objectives, many state-of-the-art results are derived via the online primal-dual method, which involves maintaining feasible solutions to primal and dual convex-relaxations as each request arrives.  This general technique, which was originally motivated from Online Set Cover, has been crucial for several other online problems including scheduling \cite{BN-JS13}, network design \cite{AAABN-TALG06}, routing \cite{BN-FOCS06}, auction design \cite{BJN-ESA07,huangKim}, paging \cite{BBN-JACM12}, and resource allocation \cite{BuchbinderMOR,MR-ICALP12,AD-SODA15}; see the book of Buchbinder and Naor \cite{BuchbinderNaor-Book09} for an overview. In 2016, Azar et.\,al.\,\cite{AzarBCCCG0KNNP16} (where 3 independent papers were merged)  introduced a general framework for applying the online primal-dual method to online covering and packing problems, unifying many of the previously distinct settings.

Although the online primal-dual technique has proven to be a powerful tool, it often only yields a \emph{fractional} solution to the optimization problem, which must then be rounded online. For simple objectives, such as $\ell_1$ and $\ell_\infty$ norms, this is not necessarily an issue due to the existence of simple and effective online rounding schemes. Yet as we move towards more complex convex objectives, our convex-relaxation may have a large integrality gap without an easy fix. For example, consider the toy example of an unweighted load balancing problem with a single job and $m$ machines. If our goal is to minimize the $\ell_2$-norm of machine loads, then any integral solution must incur cost $1$, since the job must be entirely scheduled on a machine. However, the LP-relaxation of the problem has optimum $1/\sqrt{m}$, since the job can put fractional load $1/m$ on each machine. This gives us a $\sqrt{m}$ integrality gap for the natural LP-relaxation. While this can be fixed for the $\ell_2$-norm by modifying the relaxation's objective, it is unclear how to make such fixes for general convex functions. Indeed, the results of several prior works on Online Set Cover and Load Balancing for general convex objectives are limited to fractional settings, and obtaining an integral algorithm for the same problem setting has been open \cite{AzarBCCCG0KNNP16,NS-ICALP17,KMS-STOC24}. 
This limitation motivates the question: 
\begin{quote}
    \emph{Can we design integral algorithms for Online Set cover and Load Balancing with convex objectives without relying on a convex-relaxation and incurring a loss due to the integrality gap?} 
\end{quote}

Our main result is a direct technique to design  integral online algorithms for such problems with convex objectives, bypassing the convex-relaxation and the primal-dual technique. Our techniques apply even to the ``Online Generalized Unrelated Machine Scheduling'' problem, which captures both Set Cover and Load Balancing problems as special cases, and can also capture other problems such as non-metric Facility Location. Only recently,  Deng, Li, and Rabani \cite{DLR-SODA23} designed an $O(\log n)$-approximation algorithm  for the \emph{offline} version of this problem  with symmetric norm objectives. Our results imply the first polylog approximation  for the online setting of this problem, and also imply near-optimal online  algorithms with more general convex objectives.

\subsection{Model and Results}

We are interested in  Online Generalized Scheduling with convex objectives.   The offline version of this problem with norm objectives was introduced in \cite{DLR-SODA23} as a generalization of makespan minimization, Set Cover, and minimum-norm Load-Balancing.  

\begin{definition}[Online Generalized Scheduling]
    In this problem, $n$ jobs arrive one at a time and must be scheduled on $m$ machines. When job $j \in [n]$ is scheduled on machine $i \in [m]$, it incurs a load $p_{ij} \in \Rp$. Each machine $i$ has a monotone norm $\|\cdot\|_i: \R^{n} \rightarrow \R$ to compute its total load, and we are also given a monotone convex function $f:\R^m \rightarrow \R$ that \emph{aggregates}  total  loads of all the machines. The goal is to find an allocation $x$ to  
\begin{equation}\tag{Gen-Sched} \label{eq:covGenSched}
    \min 
    f\big(\Lambda(x) \big) \quad \text{where } \quad \textstyle \Lambda_i(x) := \|(x_{i1}, \ldots, x_{in})\|_i \enspace .
\end{equation}
\end{definition}
(Later, in \Cref{def:genSchedMultChoices}, we will present a generalization of this problem where jobs can be scheduled in multiple ways on a machine, but here we restrict attention to this important special case for simplicity.)

To build some intuition for \ref{eq:covGenSched}, let us see some classical problems that it captures. It captures online makespan minimization \cite{AspnesAFPW-JACM97} when we have $n$ jobs and $m$ machines with $f$ being the $\ell_\infty$ norm and each $\|\cdot\|_i$ being the $\ell_1$ norm. It captures Online Set Cover \cite{AlonAA-STOC03} when we have a machine corresponding to  each of the $m$ sets and the online elements correspond to the $n$ jobs: we set  $p_{ij}=1$ if the element for job $j$ belongs to the set for machine $i$ and set  $p_{ij}=\infty$ otherwise; 
we take each norm $\|\cdot\|_i$ to be the $\ell_\infty$ norm and we take  $f$ to be the weighted $\ell_1$ norm with norm weights corresponding to the weights of the $m$ sets. To capture non-metric Facility Location \cite{AAABN-TALG06},   the $m$ machines correspond to the  facility locations and the $n$ online clients correspond to the $n$ jobs. Now $p_{ij}$ is the distance between client $j$ and facility $i$, and we take $f$ to be the  $\ell_1$ norm and the norm $\|\cdot\|_i$ to be the sum of the $\ell_1$-norm on assignment costs and the $\ell_\infty$-norm times cost of opening the $i$-th facility. By taking $f$ to be any general convex cost function,  
\ref{eq:covGenSched} captures convex cost generalizations of all these problems.

We obtain new near-optimal online integral algorithms for Online Generalized Scheduling for large classes of convex cost functions, 
where previously either only fractional algorithms were known or  nothing non-trivial was known. 
In particular, we obtain results for general convex and norm  aggregate functions.

\paragraph{General Convex Objectives.}   For arbitrary monotone convex functions, it is easy to show that no bounded competitive ratios are possible. Roughly, this is because a general convex function could increase arbitrarily; e.g., it might blow to infinity as soon as the load exceeds the optimum. To avoid this, \cite{AzarBCCCG0KNNP16} introduced the class of \emph{$p$-bounded} functions. Roughly, it means that the cost function $f$ has monotone gradients and does not grow faster than a degree $p$ polynomial in any direction. Besides many types of low-degree polynomials, it can also capture $\ell_p$ norms since $\|x\|_p^p$  is a $p$-bounded function. 

The main result of \cite{AzarBCCCG0KNNP16} was the use of online primal-dual technique to obtain 
an $O(p\cdot \log m)^p$-competitive algorithm for the \emph{fractional} Online Coverage Problem $\{\min f(x) \mid Ax \geq b~,~x \geq 0\}$, which in particular captures Online Set Cover. In the special case of $\|x\|_p^p$, they were able to give an integral $O(p \cdot \log m \cdot \log n)^p$-competitive algorithm for Online Set Cover.\footnote{Note that the extra dependency on $n$ is necessary for the integral problem; e.g., there is already an $\Omega(\log m \cdot \log n)$ hardness for the classic Online Set Cover problem, under standard complexity hardness assumptions~\cite{korman2004use}.}
However, since the natural convex relaxation has a large integrality gap, they had to resort to a different relaxation specific to the $\ell_p$-norm. 
A key open question left by this work was whether one could design an integral algorithm for Online Set Cover with  general $p$-bounded convex functions. We resolve this question, and further extend the result to the more general setting of Online Generalized Scheduling.

 \begin{theorem} \label{thm:OGSpboundedSym}
     For Online Generalized Scheduling with a $p$-bounded convex cost function $f$, there is an $O(p^2 \cdot (\log m)^2 \cdot \log n \cdot \log\!\log n)^p$-competitive algorithm when the  norms  computing the machine loads are symmetric.
 \end{theorem}

An immediate application of this theorem is to Online Set Cover with $p$-bounded convex functions, where only an online \emph{fractional} algorithm was known (as discussed above). In fact, in the case of Online Set Cover, the guarantees in \Cref{thm:OGSpboundedSym} can be further improved. We can save at least a factor of $\log m$ in \Cref{thm:OGSpboundedSym}, but even tighter bounds are possible for specific applications. For example, we can examine the $\ell_p$-norm Set Cover problem from \cite{AzarBCCCG0KNNP16}, where the objective is to minimize an $\ell_p$ norm over multiple linear cost functions over the sets. This objective is the $p$th root of a $p$-bounded convex function, and here our methods are able to get an $O(p \log m \cdot \log n)$-competitive ratio.

 \begin{theorem}  \label{thm:setCoverpBoundedNew}
     For Online Set Cover with a $p$-bounded convex cost function $f$, there is an $O(p^2 \cdot \log m \cdot \log n \cdot \log \log n)^p$-competitive algorithm.
       For Online Set Cover with an $\ell_p$-norm over multiple linear cost functions, we can obtain $O(p \log m \cdot \log n)$-competitive algorithm. 
 \end{theorem}

 Notably, this latter result improves upon the $O(\frac{p^3}{\log p} \log m \cdot \log n)$-competitive ratio of \cite{AzarBCCCG0KNNP16}, which illustrates how our techniques benefit from not needing to round a fractional solution.

Our results and techniques for $p$-bounded convex functions are actually much more general than \Cref{thm:OGSpboundedSym}. In particular, the norms computing the machine loads could be more general than symmetric. We show that, as long as a certain online packing problem (\Cref{defn:normPack}) has a constant-competitive algorithm for each machine norm $\|\cdot\|_i$, we get our competitive bounds for Online Generalized Scheduling. This allows us to capture several other applications such as the following Non-metric Facility Location problem with $p$-bounded convex  objectives.

In Online Non-metric Facility Location (a.k.a.\,load balancing with startup costs), the setting is similar to Online Generalized Scheduling, with machines as facilities and jobs as clients, except that there is both a one-time cost $c_i$ to open facility $i$ and an additive cost $p_{ij}$ to assign a client $j$ to facility $i$. We seek to minimize an aggregate cost function $f$ over the vector of costs associated with each facility. This problem has been greatly studied offline for norm and submodular objectives \cite{DLR-SODA23, AbbasiABBGST-ICALP24,ST-TALG10}, but we obtain the first online results for this problem with general convex aggregate functions.
\begin{theorem} \label{thm:nonmetricpBounded}
    For Online Non-metric Facility Location with a $p$-bounded convex cost function $f$, there is an $O(p^2 \cdot (\log m)^2 \cdot \log n \cdot \log\!\log n)^p$-competitive algorithm.
\end{theorem}

Next, we discuss our results for the setting where the convex cost function $f$ is a norm.

\paragraph{Norm Objectives and Compositions.}
Although norms are convex functions, most of them are not $p$-bounded since they do not even have monotone gradients; thus most previous results and techniques (as well as our Theorems \ref{thm:OGSpboundedSym}, \ref{thm:setCoverpBoundedNew}, and \ref{thm:nonmetricpBounded} above) do not directly apply when the cost function is a norm. One exception are $\ell_p$-norms, mentioned above, since $\|x\|_p^p = \sum_i x_i^p$ is $p$-bounded. Since prior works relied on $p$-boundedness, they do not extend beyond $\ell_p$-norms~\cite{AzarBCCCG0KNNP16}. Recent progress has been made for sums of $\ell_p$-norms \cite{NS-ICALP17} and for general symmetric  norms \cite{KMS-STOC24}; however, these are again obtained via online primal-dual techniques for the fractional version of the problem, and it is unclear how to extend them to integral settings. For the more general Online Generalized Scheduling, as mentioned above, the situation is even worse, since algorithms are only known for the \emph{offline} case \cite{DLR-SODA23}.

Our techniques are able to bypass these difficulties and obtain  general norm results for the online and integral settings. In particular, we obtain the following result for Online Generalized Scheduling. 

    \begin{theorem} \label{thm:schedNorm}
     For Online Generalized Scheduling where the aggregation function $f$ and the individual machine norms are symmetric norms, there is an $O(\log^2 m \cdot \log n)$-competitive algorithm. Moreover, if the aggregation function is either an $\ell_p$ or Top-$k$ norm, then we obtain the improved competitive ratio of $O(\log m \cdot \log n)$.
    \end{theorem}

    The first result in the theorem extends the $O(\log n)$-approximation for the offline case from~\cite{DLR-SODA23} (notice that, as in Set Cover, the online guarantees have an unavoidable additional $\log m$ factor compared to the offline results). 
    The second guarantee in this theorem is tight under standard complexity hardness assumptions \cite{korman2004use} since  it generalizes weighted Online Set Cover.

    Another strength of our techniques is that they handle much better \emph{compositions of norms}, a setting beyond symmetric norms where prior techniques have struggled. For instance, consider  the simpler setting of Online Load Balancing (i.e., the total load of machine $i$ is just the sum of $p_{ij}$'s of the jobs $j$ assigned to it) where the aggregation function is  a sum of $L \geq 1$ symmetric norms over disjoint subsets of coordinates. In this setting, \cite{KMS-SODA23} showed an $O(L \log^2 m)$-competitive algorithm.
    We are now able to obtain a much milder logarithmic dependence on $L$, and generalize to a much broader class of nested norms.
    
    More specifically, consider any nested norm $\|\cdot\| : \R^m \rightarrow \R$ of the form $\|x\| = \|(\|x_{S_1}\|'_{1},\ldots,\|x_{S_L}\|'_{L})\|'$, for disjoint subsets $S_1,\ldots, S_L$ of coordinates ($x_S$ denotes the vector $x$ restricted to the coordinates in the set $S$). Roughly speaking, our results are modular and we can employ them recursively, such that the final competitive ratio is the product of the competitive ratio for the top-level norm $\|\cdot\|'$ times that for the worst (\emph{not the sum}) of the second-level norms $\|\cdot\|'_{\ell}$ (actually the $\log n$ terms do not even multiply). As a sample of our composition results, we have the following (see \Cref{sec:normComp} for a more general composition result):

 \begin{theorem}  \label{thm:genSchedComposition}    
     Consider Generalized Online Scheduling where the the aggregation function $f$ is a nested norm $\|\cdot\|$, as defined above. If all norms involved are symmetric (i.e., the norms $\|\cdot\|'$ and $\|\cdot\|'_{\ell}$ in the nested norm $\|\cdot\|$, and the individual machine norms $\|\cdot\|_i$), then there is an $O(\log^2 m \cdot \log^2 L \cdot \log n)$-competitive algorithm for this problem.
 \end{theorem}

     Such norm compositions have also been studied in the context of Online Set Cover, where in particular it has was used to obtain improved guarantees for Online Buy-at-bulk Network Design~\cite{NS-ICALP17}. This last reference gave an $O(\log m)$-competitive algorithm when the cost function is a sum of $\ell_p$ norms, and \cite{KMS-STOC24} extended it to an $O(\log^3 m)$-competitive algorithm for sums (and other compositions) of symmetric norms; however, again these results are only for the \emph{fractional} version of the problem. \Cref{thm:genSchedComposition} implies integral versions of these results, with matching or improved dependence on $\log m$. Using an argument from~\cite{NS-ICALP17}, here we do not need to assume disjointness of the sets, see \Cref{app:proofsIntro} for details. 

 \begin{cor} \label{cor:setCoverComp}
    Consider Online Set Cover with the cost function given by a norm $\|x\| = \sum_{\ell \le L} \|x_{S_L}\|'_{\ell}$ where the sets $S_\ell \subseteq [m]$'s may have intersections. If each of the norms $\|\cdot\|'_{\ell}$ is an $\ell_p$ norm, then there is an $O(\log m \cdot \log L \cdot \log n)$-competitive algorithm for the problem. If each of the norms $\|\cdot\|'_{\ell}$ is symmetric, then there is an $O(\log^2 m \cdot \log L \cdot \log n)$-competitive algorithm for the problem.
 \end{cor}

We note that something similar to \Cref{thm:schedNorm} can also be obtained by combining our new \Cref{thm:OGSpboundedSym} for $p$-bounded functions with 
the recent work of \cite{KMS-STOC24}, which shows how to approximate any symmetric norm by a $p$-bounded function. However, for symmetric norms this approach loses extra polylog factors, and more fundamentally, for norm compositions it yields weak bounds with a linear dependency on the number $L$ of norms.

\subsection{Our Techniques}

Our approach to Online Generalized Scheduling is to guess the optimum and then repeatedly solve the online packing problem of maximizing the number of scheduled/accepted jobs subject to the aggregate load being at most the optimum. 

\medskip \textbf{Key Ideas.} To solve the online packing problem, our first key idea is to  assign individual budgets to each of the machines (actually we make multiple copies of each machine, with multiple possible budgets).
This allows us to decouple the machines and  run separate  ``admission control'' algorithms on each machine, based only on their local budget. To do admission control effectively, ideally, we would like to pay the cost of a machine only if it can schedule many jobs. However, since the inputs are adversarial, the past may not predict the future accurately. 

To overcome this, the second key idea we use is to determine random thresholds to ``activate'' machines:  a machine rejects all elements until the number of elements it could have accepted exceeds the random threshold, and only then we activate and pay the cost of the machine. By choosing the threshold distribution carefully, we balance out the cost of activating the machines with the number of jobs we can schedule on them. 
We note similar random activation strategies have appeared before in the context of related online constrained maximization problems; see \Cref{sec:furtherRelated} for further discussion.

To illustrate our techniques in more depth, we first show how to derive an algorithm for  Online (Weighted) Set Cover. It is $O(\log m \cdot \log n)$-competitive, so is optimal. The technical details are given in \Cref{sec:warmup}.

\medskip
\textbf{Set Cover.} In this case, our algorithm repeatedly solves a packing problem, whose online version is Budgeted Maximum Coverage. Here, instead of minimizing the cost of chosen sets, we are only allowed to choose sets so that their total cost stays within a given budget. We also can accept an element only if it is covered immediately after it arrives. The goal is to maximize the number of covered elements. Notice that this packing problem is not the dual of Online Set Cover, but is instead the problem obtained by exchanging the objective and constraints.

We can show that an $\alpha$-competitive algorithm for this packing problem gives an $O(\frac{\log n}{\alpha})$-competitive algorithm for Online Set Cover. This latter algorithm is obtained by repeatedly applying the packing algorithm on all elements that have not been packed so far. While this result may not sound surprising, getting the technical details correct is non-trivial because the packing algorithm necessarily needs to be randomized and needs to know the optimal offline value, where the offline optimal value is a random variable when we solve the packing problem repeatedly.

Our core contribution is designing an algorithm for (generalized versions of) this online packing problem. One would want to activate the sets for which the cost per covered element is low. However, when activating a set, we do not know how many elements will arrive in the future that will be covered by this set. Recall that elements only count towards our objective if they are covered at the time of arrival. Therefore, when activating a set, there is always the risk of no future elements arriving that can be covered by this set.

To mitigate this risk, we draw one randomized threshold per set. We activate a set when the number of elements that have been rejected, but could have been covered by this set, exceeds the threshold. The distribution of the threshold is chosen carefully so that conditional on a set being activated the expected number of elements it covers is high enough compared to its cost. 

\medskip
\textbf{General Convex Objectives.} Our algorithm can be extended to Online Set Cover with convex objectives and Online Generalized Scheduling. To apply our technique to Online Set Cover with convex objectives, the difficulty is that there is no more notion of the ``cost'' of a set. We therefore adapt the thresholds over time, always considering the marginal increase of the cost function at this point. This introduces additional stochastic dependencies because the threshold then depends on what other sets have already been activated.

For the extension to Online Generalized Scheduling, we need to apply one other modification. While previously, activating a set caused a fixed cost (as reflected by an $\ell_\infty$-norm), the cost of a machine will vary depending on how many and which elements it accepted. We generalize the previously mentioned approach by assigning each machine an individual budget, which also has to be activated before it is used. This essentially decouples the problem on individual machines, because we can then apply any algorithm for the acceptance decision to a single machine with the only constraint that the cost needs to stay within the assigned budget. 

This latter online algorithm for the single machine can also be another instantiation of our overall algorithm, in the case where each machine must internally solve a scheduling problem. Such a setting is modeled by the Online Generalized Scheduling problem with cost function given by a nested norm, so in this way, we are able to get a composition theorem for the norms that our techniques can handle.

\subsection{Outline}
We give a mostly self-contained application of our technique to Online Set Cover in \Cref{sec:warmup}. The main steps follow the ones of our result for Online Generalized Scheduling. In \Cref{sec:reduction}, we define the packing problem (\ref{eq:packGenSched}) that we reduce to,  and give the details of the reduction. In \Cref{sec:p-bound}, we solve this packing problem when the cost function is $p$-bounded; the core focus is solving the packing problem with machine-budgets alluded to above (\ref{eq:budgetGenSched}), and then we use it to solve the original packing problem \ref{eq:packGenSched}. Finally, in \Cref{sec:norm} we perform similar steps but for the case of norm cost functions, including how to ``compose'' our algorithms in the presence of composed norms.


\subsection{Further Related Work} \label{sec:furtherRelated}

There is an extensive work on both Online Set Cover and Online Load-Balancing (see the surveys \cite{azar2005line,BuchbinderNaor-Book09}) as well as approximations for problems with $\ell_p$-norms objective~(e.g., \cite{chandra,cody,awerbuch,uniformMachines,simultaneousLoadBal,anupamSLB,MolStochLB,BARTAL201927}). We discuss below TCS works that focus on designing approximation algorithms beyond $\ell_p$-norm objectives.

\medskip \textbf{Offline and Stochastic optimization.} 
One of the earliest work on load balancing with symmetric convex objectives focused on obtaining simultaneous approximation guarantees \cite{GM-Algo06}. The last decade has seen several work on offline optimization problems involving general norm or convex objectives, such as 
 clustering \cite{AbbasiBBCGKMSS-FOCS23,CS-STOC19,kMedian2}, load balancing \cite{CS-STOC19,ChakrabartyS19b,loadBalEasy},  facility location \cite{GMS-EC23}, and nearest neighbor search \cite{ANNNorms}. {Related to (offline) Generalized Scheduling, assignment problems with non-linear per-machine cost functions have also been studied, such as the Generalized Machine Activation problem of~\cite{LiKhullerMachineAct} (see references therein; see also~\cite{packingCovering13} for the online version of a special case of this problem, namely that of Unrelated Machine Scheduling with Startup Costs).} There is also a lot of interest in designing approximation algorithms for stochastic discrete optimization problems with norm objectives, such as stochastic load balancing \cite{stochLoadBalNorms,IbrahimpurS-SOSA22}, stochastic probing \cite{PRS-APPROX23,KMS-STOC24}, and  stochastic spanning tree \cite{ibrahimpur-thesis2022}.

\medskip \textbf{Online optimization.} As discussed earlier, the initial work  on online problems with norm or $p$-bounded objectives  focused on online set cover  \cite{AzarBCCCG0KNNP16,NS-ICALP17}.  Since then, several other online problems have been studied   such as  online facility location \cite{PRS-APPROX23}, 
online paging \cite{MenacheS15,minMaxPaging,GKP-STOC25}, online packing \cite{KMS-STOC24}, online vector scheduling \cite{KMS-SODA23}, and bandit with knapsacks \cite{KMS-SODA23,KMS-STOC24}.

 \medskip \textbf{Random activation strategies.} We remark that  ``random activation'' strategies similar to the ones designed in this paper have appeared before in the context of related online constrained maximization problems, such as the online max $k$-coverage problem \cite{awerbuch1996making} and its extension to the online weighted rank maximization problem \cite{buchbinder2012approximation}. In these works, and in our present work, random activations serve to balance the risk of activating a machine too soon or too late when faced with adversarial arrivals. 

However, one notable difference in our setting  from these prior works is that the budgeted problem (``admission control'') within each machine is non-trivial. That is, even after a machine has been activated, it is not obvious how to decide which jobs it should accept, as accepted jobs must still satisfy a norm-budget constraint. Moreover, the job admission decisions of one machine can affect the activation of other machines. Perhaps surprisingly, our results show that a random activation strategy is still compatible with the use of black-box admission control algorithms on the machines, despite these complex dynamics.


%% file: warm-up.tex
Before proving our main theorems, we will explain our main ideas in the simpler setting of \emph{Online Set Cover}.

\begin{definition}[Online Set Cover (OSC)]
    We have $n$ elements and $m$ sets $(S_i)_{i \in {m}}$ with positive set costs $(c_i)_{i \in [m]}$. Elements $j \in [n]$ arrive online and reveal their membership in each set $S_i$. The algorithm must maintain a family of \emph{active} sets $A$ that cover each arriving element, i.e., it constructs a sequence of active sets $\varnothing = A_0 \subseteq \dots \subseteq A_n \subseteq [m]$, where $A_j$ denotes the active sets after arrival $j$ and ensure that the coverage constraint $j \in \cup_{i \in A_j} S_i$ is satisfied. The goal of the algorithm is to minimize $\sum_{i \in A_n} c_i$.
\end{definition}

We can be write online set cover as an instance of \ref{eq:covGenSched}:  sets $S_i$ correspond to machines and each machine incurs a load $c_i$ if any element in $S_i$ is assigned to it, which corresponds to buying the set $i$. Formally, denote
\[
p_{ij} =\begin{cases}
        c_i & j \in S_i\\
        +\infty & j \not \in S_i
    \end{cases}, \qquad \forall i \in [m],~ \forall j \in [n] \enspace .
  \]
Now OSC corresponds to \ref{eq:covGenSched} with the following parameters:
\begin{align*}
    f(\Lambda) = \|\Lambda\|_1 = \sum_{i \in [m]} \Lambda_i, \quad \text{ where }
    \Lambda_i := \|(x_{ij} \cdot p_{ij})_j\|_\infty \quad \forall i \in [m].
\end{align*}

\subsection{Reducing to a Packing Problem}
\label{sec:warmup:reduction}

Instead of directly solving the set cover problem, we reduce it to (several instances of) the following packing problem. In this problem, instead of minimizing cost subject to a coverage constraint, we now seek to maximize coverage subject to a budget constraint.

\begin{definition}[Online Budgeted Coverage Maximization (OBCM)\footnote{We remark that this problem closely resembles the online max $k$-coverage problem of \cite{awerbuch1996making}, which corresponds to the special case of $B = k$ and $c_i = 1$ for each $i$. As such, the algorithm in this warm-up setting closely resembles that of \cite{awerbuch1996making}. However, recall that we ultimately seek to extend this algorithm to more general settings in Sections \ref{sec:p-bound} and \ref{sec:norm}, and this warm-up algorithm serves as a base that we will build on later.}]
We have $n$ elements, $m$ sets $(S_i)_{i \in m}$ with positive set costs $(c_i)_{i \in [m]}$, a budget $B$. Elements $j \in [n]$  arrive online and reveal their membership in each set $S_i$. 
The algorithm must maintain a family of active sets $\varnothing = A_0 \subseteq \dots \subseteq A_n$ online that satisfies the budget constraint $\sum_{i \in A_j} c_i \leq B$ for all $j$. The goal is  to maximize $|\{j : j \in \bigcup_{i \in A_j} S_i\}|$, i.e., the number of elements that are covered when they arrive.
\end{definition}

Essentially, OBCM is the packing problem that results from swapping the constraints and objective of OSC, as opposed to the dual packing problem (as in LP duality). Observe that the inputs of OSC and OBCM are the same, except the budget $B$. Consider any instance $I_{OSC}$ of Online Set Cover. Now consider the corresponding instance of OBCM with $B = \Opt(I_{OSC})$. If an online algorithm for OBCM was able to obtain the optimal value of $n$ for this instance of OBCM, then the same algorithm would also exactly solve OSC. 
In a similar vein, the following lemma argues that even with just an approximate algorithm for OBCM, we can get an approximate algorithm for OSC.

\begin{lemma}
\label{lem:warm-up-reduction}
    Suppose for some $\alpha \in (0,1]$  there is an $\alpha$-competitive online algorithm for OBCM that is given the value of $\Opt_{\textnormal{OBCM}}$ as input. Then, there exists an $O((1/\alpha)\cdot \log n)$-competitive  algorithm for OSC.
\end{lemma}
\begin{proof}[Proof sketch]
    The key idea is to use multiply copies of the OBCM algorithm to iteratively assign all $n$ arrivals, and then bound the number of copies required. For simplicity, let us assume for now that the OBCM algorithm is deterministic. Then our construction will go as follows:

    First, by standard guess-and-double tricks, we can lose a constant in our competitive ratio to assume that we know the optimal value of $\Opt_{OSC}$. Using this value for $B = \OPT_{OSC}$, we will initialize $N = \frac{\ceil{\log n}}{\alpha}$ algorithmic agents, each running a copy of the OBCM algorithm. For reasons which will become apparent soon, we let the $k$th agent use $(1 - \alpha)^{k-1}n$ as its estimate on $\Opt_{\textnormal{OBCM}}$. 
    
    When an element $j$ arrives, it is initially sent to the first agent. If the agent chooses not to cover $j$, then $j$ is sent to the next agent in the sequence. This process then continues with each agent, until either some agent has activated a set covering $j$, or all agents have chosen not to cover $j$. In our OSC instance, we will activate a set if it is activated by any agent.

    Notice that because of our guarantees, the first agent will cover at least $\alpha n$ elements, leaving at most $(1 - \alpha)n$ to arrive for the second agent. Likewise, the second agent will leave at most $(1 - \alpha)^2n$ for the third agent, and so on. Ultimately, the last agent will leave at most $(1 - \alpha)^N n < 1$ uncovered elements, hence all elements are covered by some set purchased by some agent. Since each agent opens sets of total cost at most $B = \Opt_{OSC}$, we incur total cost at most $NB = O(\log n/\alpha) \cdot \Opt_{OSC}$.
\end{proof}
Of course, we are brushing over some important details in this sketch, such as the case where the OBCM algorithm is randomized and only covers an $\alpha$-fraction of elements in expectation. We show how to handle these details in the general reduction given in \Cref{sec:reduction}.

\subsection{Solving the Packing Problem}
\label{sec:warmup:packing}
Next, we design an online algorithm for OBCM. 

\begin{theorem}\label{thm:warm-up-packing}
Given the value of $\Opt_{\textnormal{OBCM}}$, there is a $\frac{1}{O(\log m)}$-competitive algorithm for OBCM, .
\end{theorem}
Combining this theorem with \Cref{lem:warm-up-reduction} immediately implies the following corollary.

\begin{corollary}
    There is an $O(\log n \cdot \log m)$-competitive algorithm for Online Set Cover.
\end{corollary}

We first describe our algorithm to prove \Cref{thm:warm-up-packing}.

\medskip \textbf{Algorithm.} 
We will design an algorithm that guesses a threshold $\tau_i$ for each set $S_i$, where $\tau_i$ determines how long to wait before activating set $S_i$. In particular, we choose a threshold $\tau_i = \bar \tau_i \cdot \frac{c_i \Opt_{\textnormal{OBCM}}}{2B}$, where $\bar \tau_i \in \{0, \frac{1}{\ceil{2 \log m}}, \frac{2}{\ceil{2 \log m}}, \dots, 1 \}$ drawn such that
$\Pr\left[\bar \tau_i = 1 - \frac{k}{\ceil{2 \log m}}\right] = 2^{-k}$ for  $k \in \{1, \ldots, \ceil{2 \log m}\}$ (and $\bar \tau_i = 0$ with remaining probability). 
Crucially, our distribution over $\tau_i$ is chosen such that, conditional on $S_i$ activating, we expect to activate $S_i$ early enough to still cover many remaining elements.

When each element $j$ arrives, if $j$ is already covered (i.e. $j \in \bigcup_{i \in A_j} S_i$), then we don't open any new sets, but we consider the first activated $S_i$ for which $j \in A_i$ to be the set ``responsible'' for covering $j$ (i.e. by adding $j$ to $C_i$).
Otherwise, 
if $j$ is not already covered by an active set, then we \emph{offer} $j$ to each inactive set containing $j$ (in lexicographical order) by adding $j$ to $O_i$. 
When a set $i$ is offered $j$, it checks if it has received at least $\tau_i$ total offers. If so, and we have budget remaining, then $S_i$ activates. Otherwise it passes, and $j$ is offered to the next set in the sequence. As soon as some set $i$ activates, we stop making offers, so at most one new set is activated. We then repeat this process until all $j \in [n]$ arrive.

We define our algorithm formally in \Cref{alg:setCoverPacking}.

\begin{algorithm}[H]
\caption{Packing Algorithm for OBCM} \label{alg:setCoverPacking}
\begin{enumerate} \setlength{\itemindent}{-1em}
\item Start with active sets $A = \varnothing$ and  $\forall i \in [m]$  with covered elements $C_i = \varnothing$ and offered elements $O_i = \varnothing$.
\medskip

\item For each set $i$,  independently draw a random threshold 
        $\tau_i := \bar{\tau}_i \cdot \frac{c_i}{2B}\cdot  \Opt_{\textnormal{OBCM}}$, where   multiplier $\bar{\tau}_i$ is distributed as $(1 - \frac{k}{\ceil{2\log m}})^+$ with probability $2^{-k-1}$ for $k \ge 0$.

\medskip

\item For each element $j \in [n]$:
\begin{enumerate} \setlength{\itemindent}{-1em}
    
 \item[]       If $j$ is covered, $j \in \bigcup_{i \in A} S_i$, then  add $j$ to $C_i$, where $i \in A$ is the earliest activated set  containing $j$.
 \medskip
    
 \item[]        Else, offer $j$  one by one to inactive sets containing it (i.e., $i \in [m] \setminus A$ with  $S_i \ni j$):
 \begin{itemize}
     \item[]  If the number of offered elements to $S_i$ exceed $\tau$ (i.e., $|O_i| \geq \tau_i$) and we have not already violated budget  (i.e., $\sum_{\ell \in A} c_\ell \leq B$), then activate set $S_i$ by adding 
      $i$ to $A$. Moreover,  since $S_i$ is responsible for covering $j$, add $j$ to $C_i$. Now go to the next element without  offering $j$ to any other sets.
\end{itemize}
\end{enumerate}
\end{enumerate}
\medskip
 \end{algorithm}   

\begin{proof}[Proof of \Cref{thm:warm-up-packing}]
    We assume that each set has cost at most the budget $B$, otherwise we can just ignore it. Also, we can assume that we are allowed to exceed the budget by a single set activation, losing only a factor $2$ in our approximation. To do this, we simulate the budget-violating algorithm, and with probability $1/2$ we choose either to only activate the last set (which exceeds the budget), or to use all sets activated before the budget is violated.  
    
To show that this algorithm gets our desired competitive ratio, we consider two cases depending on the expected total cost of the active sets.

\textbf{Case 1  ($\E[\sum_{i \in A} c_i] < \frac{B}{2}$):}
By Markov's inequality,  in this case we have    $\sum_{i \in A} c_i < B$ with probability at least $\frac{1}{2}$. 
In this event, the following claim tells us that $|\bigcup_{i \in A} C_i| \geq \frac{\Opt_{\textnormal{OBCM}}}{2}$, so we get $\E|\bigcup_{i \in A} C_i| \geq \frac{\Opt_{\textnormal{OBCM}}}{4}$.

\begin{claim}\label{claim:warm-up-1}
    If $\sum_{i \in A} c_i < B$, then $|\bigcup_{i \in A} C_i| \geq \frac{\Opt_{\textnormal{OBCM}}}{2}$.
\end{claim}

\begin{proof}
    Let $A^*$ be the sets activated in the optimal solution. Consider an element $j$ which is covered by $A^*$, but not by $A$. Then for some $i \in A^* \setminus A$, we have $j \in S_i$. Moreover, we have $j \in O_i$, which is to say that $j$ was offered to $i$ in the algorithm, but not taken.
    
    Since  $\sum_{i \in A} c_i < B$ at the end of the algorithm, it must be the case that each $i$ has $|O_i| < \tau_i \leq \frac{c_i \Opt_{\textnormal{OBCM}}}{2B}$, because when $|O_i|$ reaches $\tau_i$, set $i$ activates and we stop adding elements to $O_i$. With these two observations, we obtain the bound
    \begin{align*}
        \Opt_{\textnormal{OBCM}} -\Big|\bigcup_{i \in A} C_i\Big| ~\leq~  \Big| \bigcup_{i \in A^*} S_i \setminus \Big(\bigcup_{i \in A} C_i\Big) \Big|
        ~\leq~ \Big| \bigcup_{i \in A^*} O_i \Big|
        ~\leq~ \sum_{i \in A^*} \frac{c_i \Opt_{\textnormal{OBCM}}}{2B}
        ~\leq~ \frac{\Opt_{\textnormal{OBCM}}}{2} \enspace.
    \end{align*}
    Hence, we conclude $|\bigcup_{i \in A} C_i| \geq \frac{\Opt_{\textnormal{OBCM}}}{2}$.
\end{proof}

Now consider the remaining case.

\textbf{Case 2 ($\E[\sum_{i \in A} c_i ]\geq \frac{B}{2}$):}
We prove a lower bound on the expected coverage of any set.
\begin{claim}\label{claim:warm-up-2}
    For each $i \in [m]$, we have $\E|C_i| \geq \Pr(i \in A) \cdot \frac{c_i \Opt_{\textnormal{OBCM}}}{8B\log m} - O(\frac{\Opt_{\textnormal{OBCM}}}{m^2})$.
\end{claim}
Before proving the claim, we complete the proof of the theorem. Notice that   \Cref{claim:warm-up-2} implies
\begin{align*}
\E|\bigcup_{i \in A} S_i| ~=~ \E \sum_{i \in [m]} |C_i| ~&\geq \sum_{i \in [m]} \left(\Pr(i \in A) \cdot \frac{c_i \Opt_{\textnormal{OBCM}}}{8B\log m} - O\Big(\frac{\Opt_{\textnormal{OBCM}}}{m^2}\Big)\right) \\
&\geq \frac{\Opt_{\textnormal{OBCM}}}{8B \log m} \cdot \frac{B}{2} - O\Big(\frac{\Opt_{\textnormal{OBCM}}}{m}\Big) ~\geq~ \frac{\Opt_{\textnormal{OBCM}}}{17\log m}\enspace ,
\end{align*}
where the second inequality uses that we are in Case 2 where  $\E\sum_{i \in A} c_i = \sum_{i \in A}\Pr(i \in A) \cdot c_i\geq \frac{B}{2}$.

Therefore, in either case our algorithm is $\frac{1}{O(\log m)}$-competitive,  which proves \Cref{thm:warm-up-packing}.
\end{proof}

Finally, we prove the missing \Cref{claim:warm-up-2}, where the idea is to analyze the randomness of threshold $\tau_i$ conditioned on the thresholds of every other set. 

\begin{proof}[Proof of \Cref{claim:warm-up-2}]
    We fix the choice of thresholds $\tau_{-i}$ for all sets except $i$, and consider how different values of $\tau_i$ change the behavior of the algorithm. First, notice that $C_i$ is monotonically decreasing in $\tau_i$. Intuitively, this is because the lower the value of $\tau_i$, the sooner set $S_i$ activates, and hence the large the coverage $C_i$ is on the remaining elements (we will prove this rigorously later in   \Cref{claim:mono-new2}). Hence, there is some $\tau_{\max} = \tau_{\max}(\tau_{-i})$ in the support of $\tau_i$ such that $i$ is activated (i.e., $|C_i| > 0$) if and only if $\tau_i \leq \tau_{\max}$. 

    The key observation is that $|C_i| \geq \tau_{\max} - \tau_i$, meaning the the number of elements covered by set $i$ is at least $\tau_{\max} - \tau_i$. To see this, suppose $\tau_i < \tau_{\max}$, and let $O_i$ and $O_i^{\max}$ denote the elements offered to $S_i$ when using threshold $\tau_i$ and $\tau_{\max}$, respectively. Observe the following: 
    \begin{itemize}
        \item $|O_i^{\max}| = \tau_{\max}$ and $|O_i| = \tau_i$, since we assume that $i$ activates in both cases.
        \item $O_i \subseteq O_i^{\max}$, since the threshold on $i$ does not affect $O_i$ until $i$ activates, after which $O_i$ is frozen.
        \item $O^{\max}_i \setminus O_i \subseteq C_i$, since any $j \in O^{\max}_i \setminus O_i$ is in $S_i$, but is not covered by a set activated before $S_i$.
    \end{itemize}
    From these observations, we conclude that if $\tau_i < \tau_{\max}$, we have
    \[
    |C_i| \geq |O^{\max}_i \setminus O_i| \geq \tau_{\max} - \tau_i .
    \]

    So, crucially, by the definition of the distribution on $\bar{\tau}_i$, whenever $\tau_i$ is \emph{strictly} less than $\tau_{\max}$, then we have
    \[
    |C_i| \geq \tau_{\max} - \tau_i \geq \frac{1}{\ceil{2\log m}} \cdot \frac{c_i \Opt_{\textnormal{OBCM}}}{2B} .
    \]

    Additionally, if we have $\tau_{\max} > 0$, then because of the distribution from which $\tau_i$ is sampled, we have $\Pr[\tau_i < \tau_{\max} \mid \tau_i \leq \tau_{\max}] = \frac{1}{2}$. This implies
    \begin{align*}
        \E [|C_i|\mid \tau_{-i}] &= \Pr\Big[\tau_i \leq \tau_{\max} \mid \tau_{-i}\Big] \cdot \Pr\Big[\tau_i < \tau_{\max} \mid \tau_i \leq \tau_{\max}\Big] \cdot \E\Big[|C_i| \mid \tau_i < \tau_{\max},~\tau_{-i}\Big]\\
        &\geq \Pr\Big[\tau_i \leq \tau_{\max} \mid \tau_{-i}\Big] \cdot \frac{1}{2} \cdot \frac{1}{\ceil{2\log m}} \cdot \frac{c_i \Opt_{\textnormal{OBCM}}}{2B} - \Pr\Big[\tau_i \leq \tau_{\max} \leq 0 \mid \tau_{-i}\Big]\cdot \Opt_{\textnormal{OBCM}}\\
        &\geq \Pr\Big[i \in A \mid \tau_{-i}\Big] \cdot \frac{c_i \Opt_{\textnormal{OBCM}}}{8 B \ceil{\log m}} - \frac{\Opt_{\textnormal{OBCM}}}{m^2}.
    \end{align*}
    Taking expectation over $\tau_{-i}$  completes the proof of the claim.
\end{proof}


%% file: reduction.tex
As the online set cover example from the previous section illustrated, a key idea of our approach for solving online generalized scheduling is reducing it to a packing problem, where instead of scheduling all jobs while minimizing cost, we now seek to maximize the number of jobs scheduled subject to a cost constraint. We now first introduce this packing problem in \Cref{sec:problemDef}, and then in the remainder of the section prove the desired reduction, mirroring \Cref{sec:warmup:reduction}.

The core of our final argument will be to solve this packing problem, generalizing the ideas in \Cref{sec:warmup:packing}. We provide the respective packing algorithms and analyses in \Cref{sec:p-bound} for $p$-bounded functions and in \Cref{sec:norm} for norms.

\subsection{Problem Definitions} \label{sec:problemDef}
Now we consider our main problem of interest, \emph{Online Generalized Scheduling}, described in the introduction. Actually, we now define a more general version of the problem, where each job can be processed in multiple ways on the same machine, which is required for some of our applications.

\begin{definition}[Online Generalized Scheduling] \label{def:genSchedMultChoices}
    Let $f:\Rp^m \rightarrow \Rp$ be a monotone convex \emph{aggregate} function, and let $\|\cdot\|_i: \R^{nr} \rightarrow \Rp$ for $i \in [m]$ be monotone norms.     
    In this problem, \ref{eq:covGenSched}$_{f,\{\|\cdot\|_i\}_i}$, $n$ jobs arrive one at a time and must be scheduled on $m$ machines. Each job has at most $r$ ways of being scheduled on a given machine, where scheduling job $j \in [n]$ on machine $i \in [m]$ in way $k \in [r]$ incurs a load $p_{ijk} \in \Rp$.
    The goal is to find online an allocation $x$ of all jobs ($x_{ijk} \in \{0,1\}$ indicates whether $j$ was scheduled on $i$ in way $k$) to  
\begin{equation}\tag{Gen-Sched} \label{eq:covGenSched}
    \min 
    f(\Lambda(x)) \quad \text{where } \quad \textstyle \Lambda_i(x) := \|(x_{ijk} p_{ijk})_{jk}\|_i. 
\end{equation}
When $r = 1$, we write simply $x_{ij} := x_{ij1}$ and $p_{ij} := p_{ij1}$.
\end{definition}

We will reduce \ref{eq:covGenSched}  to a ``dual'' packing problem, where we seek to maximize the number of jobs scheduled subject to a cost constraint.  The natural associated packing problem is \emph{Online Schedule-Packing}, which corresponds to the OBCM problem in the online set cover example. Since we will only be approximately solving the packing problems, we define \emph{partial allocation/scheduling} to mean an allocation that does not need to schedule all the jobs.

\begin{definition}[Online Schedule-Packing] 
The input to this problem, \ref{eq:packGenSched}$_{f,\{\|\cdot\|_i\}_i}$, is the same as \ref{eq:covGenSched}, except that we are also given a \emph{budget} $B \in \Rp$. When job $j \in [n]$ arrives, the algorithm irrevocably decides whether to schedule it or not, and scheduling it on machine  $i \in [m]$ in way $k \in [r]$ incurs a load $p_{ijk} \in \R_{\ge 0} $.  The goal is to find online a \emph{partial allocation} $x$ maximizing the number of scheduled jobs  subject to the aggregate budget constraint:
\begin{equation} \label{eq:packGenSched}
  \max~  \sum_{i,j,k} x_{ijk}  \quad \text{ s.t. }  \quad  f(\Lambda(x)) \le B, \quad \text{ where } ~ \Lambda_i(x) = \|(x_{ijk} p_{ijk})_{j,k}\|_i. \tag{Sched-Pack} 
\end{equation}

\end{definition}

Finally, to solve \ref{eq:packGenSched}, we also need to know how to allocate jobs within a machine once it is active. This was trivial for OBCM, where  $\|\cdot\|_i$ are $\ell_\infty$ norms, since we allocate  every job to an active machine when feasible. However, for other norms, it is non-trivial to pack jobs on an active machine $i$ while maintaining the constraint  $\|(x_{ijk} p_{ijk})_{j,k}\|_i \le b_i$. Hence, we define Norm Packing problem to be the special case of  \ref{eq:packGenSched} with a single-machine. 

\begin{definition}[Norm Packing] \label{defn:normPack}
This problem, \ref{eq:normPack}$_{\|\cdot\|}$, is the special case of  \ref{eq:packGenSched} with a single-machine  and the aggregate function $f$ being the identity function,  namely:
\begin{equation}
  \max~  \sum_{j,k} x_{jk} ~\quad \textrm{ s.t. } \quad  \|(x_{jk} p_{jk})_{j,k}\| \leq B. \tag{Norm-Pack} \label{eq:normPack}
\end{equation}
\end{definition} 

Our main result of this section is  that there exists a good algorithm \ref{eq:covGenSched} if there exist good algorithms for the corresponding \ref{eq:packGenSched}  and  \ref{eq:normPack}$\|\cdot\|_i$ problems. This reduction will hold as long as the aggregate function satisfies a natural property, which we call $p$-subadditivity.

\subsection{ $p$-Subadditivity and $(\alpha, c)$-Solvability}

Before proving the reduction, we need to introduce an important property associated with the aggregate function $f$, which is of  \emph{$p$-subadditivity}.

\begin{definition}[$p$-subadditive]
    A monotone convex function $f: \Rp^m \to \Rp$ is \emph{$p$-subadditive} for $p \geq 1$ if $f^{1/p}$ is subadditive. That is, for all $x, y \in \R^m_+$, we have $f(x + y)^{1/p} \leq f(x)^{1/p} + f(y)^{1/p}$.
\end{definition}

This class of functions clearly captures captures all monotone norms for $p=1$. Now we will see that it also captures the class of all \emph{$p$-bounded} functions, as defined in \cite{AzarBCCCG0KNNP16}. Roughly, it means that the function does not grow faster than a degree $p$ polynomial in any direction.

\begin{definition}[$p$-bounded]\label{def:p-bounded}
    A monotone, differentiable, and convex function $f: \Rp^m \to \Rp$ is \emph{$p$-bounded} for $p \geq 1$ if $\nabla f$ is monotone increasing and $\frac{\langle\nabla f(y), y \rangle}{f(y)} \leq p$ for all $y \in \Rp^m$.    
\end{definition}

Besides many types of low-degree polynomials, it can also capture $\ell_p$ since $\|x\|_p^p$  is a $p$-bounded function. 

\begin{lemma}\label{lem:p-subadditive}
    Every $p$-bounded convex function $f : \Rp^m \rightarrow \R$ is $p$-subadditive.
\end{lemma}
\begin{proof}
     We use \cite[Lemma 4.a]{AzarBCCCG0KNNP16}, which states that for any $p$-bounded 
     $f:\Rp^m \rightarrow \R$ and any $\delta > 1$, we have $f(\delta x) \leq \delta^p f(x)$. Using this fact and convexity, we have for all $\lambda \in [0, 1]$ that
    \[ 
        f(x + y) ~=~ f\left(\lambda(\tfrac{x}{\lambda}) + (1-\lambda)(\tfrac{y}{1-\lambda})\right)
        ~\leq~ \lambda f(\tfrac{x}{\lambda}) + (1 - \lambda) f(\tfrac{y}{1-\lambda})
        ~\leq~ \lambda^{1-p}f(x) + (1-\lambda)^{1-p} f(y).
    \]
    Solving for the optimal $\lambda$, we find that $\lambda = \frac{f(x)^{1/p}}{f(x)^{1/p} + f(y)^{1/p}}$ gives the tightest inequality. Plugging in this value of $\lambda$, the RHS simplifies to $(f(x)^{1/p} + f(y)^{1/p})^p$, which gives the desired result.
\end{proof}

Finally, since in  \Cref{sec:p-bound} and \Cref{sec:norm} we will only be able to (approximately) solve \ref{eq:packGenSched} with a  violation in the budget constraint, we introduce the notion  ``solvability'' to capture   bi-criteria approximations.

\begin{definition}[$(\alpha, c)$-Solvable Online Packing]\label{def:online-packable}
    An instance  $I$ of \ref{eq:packGenSched} (which could be \ref{eq:normPack}) is called 
    \emph{$(\alpha, c)$-solvable} for $\alpha \in [0,1]$ and $c \geq 1$ if there is an online algorithm for $I$ that violates the budget $B$ by at most a factor $c$ (i.e., incurring load at most $cB$), and when given a lower bound $M \leq \Opt(I)$ up front, it obtains expected objective value at least $\alpha M$.
\end{definition}

For our reduction, we require as a condition that \ref{eq:normPack}$_{\|\cdot\|_i}$ is $(\alpha, c)$-solvable for the norms $\|\cdot\|_i$. Thus, as a necessary base case, we provide in \Cref{lem:single-machine} some simple and common norms satisfy this condition are given. See \Cref{app:single-machine} for a proof of these cases.

\begin{lemma}\label{lem:single-machine}
\ref{eq:normPack}$_{\|\cdot\|}$  is $(\alpha, c)$-solvable for the following norms:
    \begin{enumerate}
        \item For $\|\cdot\| = \|\cdot\|_\infty$, \ref{eq:normPack}$_{\|\cdot\|}$ is $(1, 1)$-solvable.
        \item When $\|\cdot\|$ is symmetric, \ref{eq:normPack}$_{\|\cdot\|}$ is $(\frac{1}{3}, 1)$-solvable.
        \item For $\|x\| = c\|x\|_\infty + \sum_{j,k} w_{jk} |x_{jk}|$ (e.g., a machine with activation and assignment costs), \ref{eq:normPack}$_{\|\cdot\|}$ is $(\frac{1}{3}, 2)$-solvable. Also, the algorithm is oblivious to the future values of $w_{jk}$.
    \end{enumerate}
\end{lemma}

In the remainder of the section, we reduce  \ref{eq:covGenSched} with a $p$-subadditive aggregate function to $(\alpha, c)$-solvability of a corresponding \ref{eq:packGenSched}  problem.  In \Cref{sec:p-bound} and \Cref{sec:norm}, we show how to solve the latter when the norm problems \ref{eq:normPack}$_{\|\cdot\|_i}$ are $(\alpha,c)$-solvable.

\subsection{Reduction from \ref{eq:covGenSched} to \ref{eq:packGenSched}}

We can now state our scheduling-to-packing reduction. 

\begin{theorem}\label{thm:packCovReduction}
        Suppose \ref{eq:packGenSched}$_{f, \{\|\cdot\|_i\}_i}$ is $(\alpha, c)$-solvable     
        for some $p$-subadditive aggregate function $f$ and norms $\|\cdot\|_i$. Then, there exists an $O\big((\frac{1}{\alpha} \cdot p \log n \cdot \log\!\log n)^p \cdot c\big)$-competitive algorithm for \ref{eq:covGenSched}$_{f, \{\|\cdot\|_i\}_i}$. Moreover, when $p=1$, there exists an improved $O(\frac{1}{\alpha}\cdot\log n \cdot c )$-competitive algorithm for \ref{eq:covGenSched}$_{f, \{\|\cdot\|_i\}_i}$. 
\end{theorem}

We now prove  this theorem for general $p$, deferring the improved bound for $p=1$ to \Cref{app:normReduction}. 

Consider an instance $I_{OGS}$ of \ref{eq:covGenSched} that we wish to approximate using the assumed algorithm for \ref{eq:packGenSched}. We first note that, by incurring an additional loss, we can assume that we have an estimate of the optimum $\Opt_{OGS}$ of $I_{OGS}$.

\begin{claim}\label{lem:p-bound-doubling}
    Suppose that there is an algorithm for \ref{eq:covGenSched} with $p$-subadditive aggregate function that is $\beta$-competitive given an estimate $\widehat{\Opt}_{OGS}$ such that $\Opt_{OGS} \leq \widehat{\Opt}_{OGS} \leq 2^p \cdot \Opt_{OGS}$. Then there is also a $\beta 2^p$-competitive algorithm without such an estimate.
\end{claim}

The proof is deferred to \Cref{app:doubling}. It uses a standard guess-and-double approach, together with the $p$-subadditivity of the aggregate function (\Cref{lem:p-subadditive}) to combine the solutions obtained in each doubling phase.

Thus, assume we have an estimate $\widehat{\OPT}_{OGS}$ from \Cref{lem:p-bound-doubling}. Let $\cA_{SP}$ denote the algorithm for \ref{eq:packGenSched} assumed in the statement of theorem due to $(\alpha, c)$-solvability. We first show that we can this algorithm to find a schedule for $I_{GS}$ with cost at most  $O\big((\frac{1}{\alpha}\cdot \log n \log\!\log n)^p \cdot c\big)\, \Opt_{OGS}$ with probability $\frac{1}{2}$; later we will ``iterate'' this procedure to boost the success probability to $1$.

\begin{lemma}\label{lemma:packCovReduction1}
    Using the algorithm $\cA_{SP}$, there is a randomized algorithm such that for any subinstance (i.e., subset of the jobs) $I' \subseteq I_{OGS}$, it creates a partial assignment $X$ of the jobs in $I'$ satisfying:

    \begin{enumerate}
        \item The cost of the assignment satisfies $f(\Lambda(X)) \le  O\big( (\frac{1}{\alpha}\cdot \log n \log\!\log n)^p \cdot c\big)\,\Opt_{OGS}$, with probability 1. 

        \item With probability at least $\frac{1}{2}$, $X$ assigns all jobs of $I'$.
    \end{enumerate}
\end{lemma}

We discuss the main ideas of the proof of this lemma here, and defer the details to \Cref{app:cov-to-pack-proof}.

\begin{proof}[Proof Sketch of \Cref{lemma:packCovReduction1}]
    Fix an Online Generalize Scheduling subinstance $I' \subseteq I_{OGS}$. In order to try to schedule all the jobs in $I'$, we will run $N \cdot (\log n + 1)$, where $N := \frac{10\ln (2\log n) + 2}{\alpha}$, copies (henceforth referred to as ``agents'') of the packing algorithm $\cA_{SP}$. We group consecutive agents into $(1+\log n)$ groups, each with $N$ agents. Each agent in the $k$th group uses $M_k := \frac{n}{2^k}$ as the required lower bound on the optimum of the instance comprising the jobs it receives. Each of these agents runs on its own instance; all of these instances have the same aggregate function $f$ and  norms $\|\cdot\|_i$ as in the original instance $I_{OGS}$, but now have an aggregate load budget of $B = \widehat{\OPT}_{OGS}$. 

    The algorithm for trying to schedule the jobs from the instance $I'$ is then the following: when a new job from $I'$ arrives, we first send it to (the instance of) the first agent; if this job is not scheduled by the agent, we send it to the next agent, and so on. Notice that because of how the agents are grouped, a job is offered first to all agents of group 1, then to all agents of group 2, etc. This  gives a (partial) scheduling of the jobs in $I'$. 

    The claim is that with probability $\frac{1}{2}$ this process schedules all jobs from $I'$, and the final schedule always has cost $O(\frac{c \log n \log\!\log n}{\alpha})^p \cdot \widehat{\OPT}_{OGS}$. The high-level intuition is: assuming for now that the $\cA_{SP}$ is an $\alpha$-approximation for the packing problem \ref{eq:packGenSched} without requiring a lower bound on the optimum (so we do not need to worry about the correctness of the estimates $M_k$), each of the agents in a group is expected to assign an $\alpha$-fraction of the jobs it receives, and thus (further assuming independence between agents, for now) a group leaves unassigned $\lesssim (1-\alpha)^{10\log\!\log n/\alpha} \ll \frac{1}{2}$ of the jobs it receives. Because of the extra slack in the exponent, by a martingale tail bound, this upper bound holds with probability at least $1 - \frac{1}{2(\log n + 1)}$. Taking a union bound over the $\log n + 1$ groups, with probability at least $\frac{1}{2}$ this halving of the unscheduled jobs holds for all groups and thus jobs are actually scheduled. Since the assignment of each agent fits within a blown up load budget $c \, \widehat{\OPT}_{OGS}$, using the $p$-subadditiviy of both the aggregate functions (\Cref{lem:p-subadditive}) and additivity of the  norms, the union of the assignment of these agents has cost $O(\frac{\log n \log\!\log n}{\alpha})^pc \cdot \widehat{\OPT}_{OGS}$, essentially giving the bound claimed in the theorem. A main complication for formalizing this idea is the possible randomness of the agents, so the instance seen by the $k$th agent is random and depends on the actions of agents $1,2,\ldots, k-1$, introducing correlations in this process. 
\end{proof}

    Using this lemma we now prove \Cref{thm:packCovReduction}.

    \begin{proof}[Proof of \Cref{thm:packCovReduction}]
    In order to schedule all the jobs of the instance $I_{OGS}$ with probability 1 (instead of $\frac{1}{2}$ as in the previous lemma) we use a similar repetition idea as in the latter. More precisely, let $\cA_{partial}$ denote the algorithm from \Cref{lemma:packCovReduction1}. We run in sequence infinitely many agents, each using an instantiation of the algorithm $\cA_{partial}$, where an incoming job from $I_{OGS}$ is sent to the first agent, and if it is not scheduled by it we send the job to the next agent, and so on. 
    
    We note that \emph{all jobs} of the instance $I_{OGS}$ are scheduled by this procedure with probability 1, since there are as many agents as needed for that (this can be made formal using \eqref{eq:tailTau} below). Thus, to prove \Cref{thm:packCovReduction}, it suffices to control the total load of the schedule produced. At a high-level, the guarantee in Item 2 of \Cref{lemma:packCovReduction1} indicates that we actually only need $2$ agents in expectation to schedule all the jobs, and the guarantee in Item 1 helps controlling the total cost of the combined schedule of these agents. 

    To make this formal, let $X^i$ denote the schedule made by the $i$th agent, and $X$ denote the schedule made by our whole procedure, namely $X = \sum_i X^i$. Also, define $\tau$ to be the index of the last agent that receives some job to be scheduled (so $X = \sum_{i \le \tau} X^i$). \Cref{lemma:packCovReduction1} guarantees that the load of the schedule of each agent satisfies  $f(\Lambda(X^i)) \le  O(\frac{\log n \log\!\log n}{\alpha})^p c \cdot \Opt_{OGS}$. Thus, using $p$-subadditivity of the aggregate function $f$ (\Cref{lem:p-subadditive}) and subadditivity of $\Lambda$ (by subadditivity of norms), the cost of the total schedule satisfies 
    \begin{align}
        f(\Lambda(X)) = f\bigg(\Lambda\bigg(\sum_{i \le \tau} X^i\bigg)\bigg) \le \bigg( \sum_{i \le \tau} f(\lambda(X^i))^{1/p}\bigg)^p \le \tau^p \cdot O\bigg(\frac{\log n \log\!\log n}{\alpha}\bigg)^p c \cdot \Opt_{OGS}. \label{eq:numCopies}
    \end{align}
    We now need to bound the number of used agents $\tau$. For that, we use a martingale concentration bound proved in \cite{cicalese20a}, which we slightly extend in \Cref{app:conc}.

    Let $Y_i$ be the indicator that the $i$th agent assigned all jobs it received; thus, $\tau$ is at the first index $i$ where $Y_i = 1$. Let $\cF_{i-1}$ be the $\sigma$-algebra generated by the agents $1,2,\ldots,i-1$. Notice that $\cF_{i-1}$ fixes the subinstance of $I_{OGS}$ that is given as input to the $i$th agent, and therefore we can apply   \Cref{lemma:packCovReduction1} conditionally to obtain that $\E [Y_i \mid \cF_{i-1}] \ge \frac{1}{2}$. Thus, the event ``$\tau > v$'' is equivalent to the event ``$Y_1 + \ldots + Y_v \le 0 \textrm{ and } \sum_{i \le v} \E [Y_i \mid \cF_{i-1}] \ge \frac{v}{2}$'', so applying martingale concentration from \Cref{thm:mart}   gives
    \begin{align}
    \Pr(\tau > v) \,=\, \Pr\bigg(Y_1 + \ldots + Y_v \le 0 \textrm{ and } \sum_{i \le v} \E [Y_i \mid \cF_{i-1}] \ge \frac{v}{2} \bigg) \,\le\, \exp\bigg({-\frac{3}{14} \frac{v}{2}}\bigg)  \label{eq:tailTau}
    \end{align}
    for $v \ge 2$. Using this, we can bound $\E [\tau^p]$ by integrating tails:
    \begin{align*}
        \E [\tau^p] = \int_0^\infty v^{p-1} \cdot \Pr(\tau \ge v) \,\d v \,\le\, \int_0^2 v^{p-1} \,\d v + \int_2^\infty v^{p-1} \cdot e^{-\frac{3v}{28}}  \,\d v \,\le\, \frac{2^p}{p} + \frac{p!}{(3/28)^{p+1}},
    \end{align*}
    where the last inequality follows because the second integral is at most $\frac{1}{(3/28)}$ times the $p$-th moment of the exponential distribution with rate parameter $\frac{3}{28}$, which is ${p!}\cdot{(\frac{3}{28})^{-p}}$. 

    Since this last upper bound is $O(p)^p$, plugging it into \eqref{eq:numCopies} we obtain that the expected load of the schedule $X$ is upper bounded $\E f(\lambda(X)) \le O(\frac{1}{\alpha}\cdot p c \log n \log\!\log n)^p \cdot \Opt_{OGS}$. This concludes the proof of \Cref{thm:packCovReduction}. 
    \end{proof}


%% file: p-bounded.tex
The last section reduced  generalization scheduling (\ref{eq:covGenSched}) to  schedule packing (\ref{eq:packGenSched}), where instead of scheduling all jobs while minimizing cost, we seek to maximize the number of jobs scheduled subject to a cost constraint. Now we design bi-criteria approximation algorithms for \ref{eq:packGenSched}.

Recall that an instance of \ref{eq:packGenSched} is 
    \emph{$(\alpha, c)$-solvable}  if  given a lower bound $M$ on $\OPT_{SP}$ we can schedule $\alpha M$ jobs while violating the budget $B$ by a factor of $c$. 
    The following is the main result of this section.

\begin{theorem}\label{thm:pboundPackSched}
Consider an instance of  \ref{eq:packGenSched}$_{f, \{\|\cdot\|_i\}_i}$ where the aggregate function $f$ is $p$-bounded and  \ref{eq:normPack}$_{\|\cdot\|_i}$ is $(\alpha, c)$-solvable for each norm $\|\cdot\|_i$. Then, \ref{eq:packGenSched}$_{f, \{\|\cdot\|_i\}_i}$ is $(\Omega(\frac{\alpha}{\log^2 m}), (3pc)^p)$-solvable. 
\end{theorem}

Combining this theorem with \Cref{thm:packCovReduction}, we obtain the following result for \ref{eq:covGenSched}.

\begin{corollary} \label{cor:genSched}
    When the conditions of \Cref{thm:pboundPackSched} hold, we can get an $O(cp^2 \log n (\log\!\log n) \log^2 m/\alpha)^p$-competitive algorithm for \ref{eq:covGenSched}$_{f, \{\|\cdot\|_i\}_i}$.
\end{corollary}

Together with the $(\alpha,c)$-solvability results of \ref{eq:normPack} for different norms from \Cref{lem:single-machine}, this corollary directly implies \Cref{thm:OGSpboundedSym} stated in the introduction (see \Cref{app:proofsIntro} for full details).

To prove \Cref{thm:pboundPackSched}, we will  reduce the problem to a variant where each machine has a preset budget on its total load.  \Cref{sec:budgetSchedPack}  formally defines this budgeted problem,  \Cref{sec:pbound-budget-sched-alg}  gives an algorithm for this budgeted problem, and finally, \Cref{sec:reductionP} proves the reduction from \ref{eq:packGenSched} to this budget problem.

\subsection{Online Budgeted Schedule-Packing} \label{sec:budgetSchedPack}

As we seek to solve \ref{eq:packGenSched} using the techniques  for OBCM  in the warm-up \Cref{sec:warmup}, we notice the following key difference between the problems: in OBCM, the $\ell_{\infty}$ norm associated to each set/machine $S_i$ means that we only need to decide whether to select this set or not, and we do not need to closely track exactly which elements/jobs $j$ are assigned to $S_i$ as the load/cost was always $c_i$. In contrast, the presence of general monotone norms $\| \cdot\|_i$ in \ref{eq:packGenSched} requires keeping track of the loads of the individual machines. 

To circumvent this, we decouple the decision of machine selection with the assignment decision into the machines. More precisely, now each machine $i$ has a fixed ``activation'' budget $b_i$: if this machine is ever used then it ``costs'' the full  $b_i$ (mirroring the Set Cover case), which is then aggregated over all the machines by the {aggregate load function $f$}. This leads to the following \emph{Online Budgeted Schedule-Packing} problem. 

\begin{definition}[Online Budgeted Schedule-Packing] 
The input to this problem, \ref{eq:budgetGenSched}$_{f,\{\|\cdot\|_i\}_i}$, is the same as \ref{eq:packGenSched}, except that each machine $i \in [n]$ also has an \emph{activation budget} $b_i \in \Rp$.  The goal is to find online a partial allocation $x$  together with machine activations $y$ ($y_i \in \{0,1\}$ indicates machine $i$ was activated) maximizing the number of  scheduled jobs subject to the  budget of each machine as well as the {aggregate load budget}, namely: 
\begin{equation}
  \max~  \sum_{i,j,k} x_{ijk} \text{ s.t. }  f(y_1 b_1,\ldots,y_m b_m) \le B ~\text{ and }~ \|(x_{ijk} p_{ijk})_{j,k}\|_i \le y_i b_i ~\forall i \in [m] . \quad \tag{Budgeted-Sched-Pack} \label{eq:budgetGenSched}
\end{equation}
\end{definition}
Notice that the activation requirement on the machines means that if $y_i = 0$, then we must have $x_{ijk} = 0$ for all jobs $j$ and  ways $k$.

In order to prove \Cref{thm:pboundPackSched}, we will first show that this budgeted version of the problem has a bi-criteria approximation as long as the single-machine scheduling problem is solvable.

\begin{theorem} \label{thm:pboundBudgetSched}
    Consider the online  \ref{eq:budgetGenSched}$_{f,\{\|\cdot\|_i\}_i}$ problem where the aggregate function $f$ is $p$-bounded and  $f(0) = 0$. Suppose for each  norm $\|\cdot\|_i$ the problem \ref{eq:normPack}$_{\|\cdot\|_i}$ is $(\alpha, c)$-solvable.
    Then, for any $s \geq 1$, there is an online algorithm for \ref{eq:budgetGenSched}$_{f,\{\|\cdot\|_i\}_i}$ with the following guarantee: Given the value of $\OptBSP$ up front, the algorithm obtains expected objective value at least $\Omega\big(\frac{\alpha}{(1 + 1/s)^p \log^2 m}\big) \cdot \OptBSP$ while violating the budget $B$ by at most a factor $s^p$ (i.e. incurring aggregate load at most $s^pB$) and violating each machine's load budget by at most a factor $c$ (i.e. so that $\|(x_{ijk} p_{ijk})_{j,k}\|_i \leq cy_i b_i$).
\end{theorem}

Notice that in \Cref{thm:pboundBudgetSched}, we assume we are given the value of $\OptBSP$ exactly. This is a convenient assumption to adopt for proving the theorem, but in actuality, we will apply \Cref{thm:pboundBudgetSched} by only giving the algorithm a lower bound $M \leq \OptBSP$. Nevertheless, our algorithm must still obtain the approximation factor promised by \Cref{thm:pboundBudgetSched} in terms of $M$, since there is some point in the online arrivals at which the hindsight optimum crosses the value $M$.


\subsection{An Algorithm for \ref{eq:budgetGenSched}}\label{sec:pbound-budget-sched-alg}

In this section, we will now prove \Cref{thm:pboundBudgetSched}. Here, we will abuse notation and simply write $f(y)$ for $f((b_i y_i)_i)$, as the $b_i$ values are fixed and the aggregate load is only determined by $y$.

The algorithm we use will be similar to that in \Cref{sec:warmup}, in that we will offer jobs as they arrive to machines, and once a machine has received enough offers, it will activate. However, we will use $\Delta_i f$ instead of $c_i$ to make our activation decisions, where $\Delta_i f(y) := f(y + e_i) - f(y)$, since $\Delta_i f$ represents the marginal cost of activating machine $i$. 

During the description of algorithm, we will define the following sets, again similar to the warm-up: 
\begin{itemize}
    \item $O_i$ is the set of items offered to machine $i$ before it gets activated.
    \item $O^{post}_i$ the set of items offered after it is activated.
    \item $T_i = (O_i \cup O_i^{post})$ be the total set of items offered to machine $i$
    \item $\Alg_i$ is the set selected by the online algorithm of machine $i$ (after being active), and $\Alg := \bigcup_i \Alg_i$.
    \item For $S \subseteq [n]$, $\OPT_i(S)$ is the set of assigned jobs in an optimal solution to the \ref{eq:normPack}$_{\|\cdot\|_i}$ problem with budget $b_i$ and only jobs in $S$.
\end{itemize}

Also, we assume that each machine can be activated within the budget $B$, namely $f(e_i) \le B$, else we can just ignore it. We also can assume that we are allowed to exceed the (potentially expanded) budget $s^pB$ by a single machine activation, losing only a factor $2$ in our approximation. To do this, we simulate the budget-violating algorithm, and with probability $1/2$ we choose either to only allocate to the last machine (which exceeds the budget), or to allocate to all machines activated before the budget is violated. Finally, in the sequel by ``\ref{eq:normPack}$_{\|\cdot\|_i}$ algorithm'' we mean one that achieves the assumed $(\alpha, c)$-solvable for this problem. 

Given this, our algorithm is the following:

\begin{algorithm}
\caption{Algorithm for \ref{eq:budgetGenSched} with $p$-bounded aggregate function $f$}\label{alg:p-bounded}

First, for each machine, draw independently a threshold $\bar{\tau}_i$ distributed as $(1 - \frac{k}{3\log m})^+$ with probability $2^{-k-1}$ (for $k \ge 0$). At any given point in the algorithm, we will keep the threshold $\tau_i := \bar{\tau}_i \cdot \frac{\OptBSP \cdot \Delta_i f(y)}{10 \cdot (s+1)^pB}$ for machine $i$. Then for each arriving job $j$:
\begin{enumerate}
    \item Offer it to each active machine $i$ in order of activation time, i.e. add job $j$ to $O^{post}_i$. If the internal \ref{eq:normPack}$_{\|\cdot\|_i}$ algorithm of this machine schedules this job, set the corresponding $x_{ijk} = 1$ and add $j$ to $\Alg_i$. Else offer $j$ to the next active machine, etc.
    \medskip

    \item If the item is not selected by any active machine, then offer it to all inactive machines (e.g., in lexicographical order). In other words, add $j$ to the $O_i$'s of each inactive machine $i$ one at a time, until some machine (activates and) schedules $j$ or until all machines have been offered $j$.
    \medskip

    \item If an inactive machine $i$ is offered $j$ and crosses the threshold with its offered items, i.e., $\OPT(O_i, b_i) \ge \tau_i$, AND $f(y) < s^pB$, then do the following: 
    \begin{enumerate}
        \item Define $\ba_i := \Delta_i f(y)$ as the current marginal.
        \item Activate this machine (i.e., add $i$ to $A$ and update $y_i = 1$).
        \item Draw a guess $M_i$ 
        uniformly from $\{\frac{\OptBSP}{2}, \frac{\OptBSP}{4}, \frac{\OptBSP}{8}, \dots, \frac{\OptBSP}{2^{\floor{2\log m}}}\}$. 
        \item Start an internal \ref{eq:normPack}$_{\|\cdot\|_i}$ algorithm for machine $i$ with guess $M_i$. This algorithm will receive only jobs offered to $i$, i.e. jobs in $O_i^{post}$.
        \item Offer $j$ to machine $i$. If $i$ schedules $j$, set $x_{ijk} = 1$, add $j$ to $\Alg_i$, and skip to next job $j+1$. Else, continue Step 2, i.e., offer to the next inactive machine.
    \end{enumerate}
    \medskip

    \item Stop after all $n$ items have arrived.
\end{enumerate}
 \end{algorithm}


\begin{proof}[Proof overview for \Cref{thm:pboundBudgetSched}] 

It is clear that \Cref{alg:p-bounded} satisfies the budget violation requirements for \Cref{thm:pboundBudgetSched} (with our assumptions), so it only remains to show that it gets the required approximation factor of $\OptBSP$. Similar to \Cref{sec:warmup}, we seek to show argue two possible cases: either the algorithm activates many machines and reaches the aggregate load budget, or the algorithm activates few machines, meaning the marginal value of the remaining machines is small, and in both cases we obtain an approximation of $\OptBSP$.

To do this, we first show an upper bound on $\OptBSP$ that can be obtained in terms of the active set $A$ by considering both of these cases. 

\begin{lemma}\label{lem:opt-upper-bound}
    For any realization of the algorithm above, we have
    \[
    \frac{9}{10}\OptBSP \leq |\Alg| + \sum_{i \in A} \max \left\{\frac{\OptBSP \cdot\ba_i}{s^pB},~ |\Opt_i(T_i)| \right\} \enspace .
    \]
\end{lemma}
In the above statement, the first term of the maximum dominates in the even that the fraction of the budget $B$ we use is large, reflecting Case~2 in \Cref{sec:warmup}. On the other hand, if little of the budget is used, as in Case~1 of \Cref{sec:warmup}, then we expect the hindsight optimum on activated machines to be large, which is given by the second term in the max. For this general setting, it is convenient to combine these cases into a single bound on $\OptBSP$, as opposed to considering them separately as we did previously.

Next, we will show that our algorithm obtains an approximation of this upper bound. Specifically, that $\E |\Alg_i|$ can by lower bounded in terms of the contribution of $i$ to \Cref{lem:opt-upper-bound}.

\begin{lemma}\label{lem:unifiedLemma}
    For some constant $C > 0$ we have for every machine $i \in [m]$
    \[
    \E\, |\Alg_i| \geq \frac{\alpha}{C\log^2 m} \cdot \E \left[ \ind(i \in A) \max\left\{\frac{ \OptBSP \cdot \ba_i}{B \cdot (s+1)^{p}},~ |\Opt_i(T_i)|\right\}\right] - \frac{2\OptBSP}{m^2}.
    \] 
\end{lemma}
Finally, combining \Cref{lem:unifiedLemma} with \Cref{lem:opt-upper-bound}, we obtain the competitive bound necessary for \Cref{thm:pboundBudgetSched}:
    \begin{align*}
        \E |\Alg| 
        &\geq \frac{\alpha}{C \log^2 m}\E\left[  \sum_{i \in A} \max\left\{\frac{\OptBSP \cdot\ba_i}{B \cdot (s+1)^p},~ |\Opt_i(T_i)|\right\}\right] - \frac{2\OptBSP}{m}\\
        &\geq \frac{\alpha}{C (1+ 1/s)^p\log^2 m}\E\left[  \sum_{i \in A} \max\left\{\frac{\OptBSP \cdot\ba_i}{s^p B},~ |\Opt_i(T_i)|\right\}\right] - \frac{2\OptBSP}{m}\\
        &\geq \frac{\alpha}{C (1 + 1/s)^p\log^2 m} \cdot \bigg(\frac{9}{10}\OptBSP - \E|\Alg|\bigg) - \frac{2\OptBSP}{m}.
    \end{align*}
    Now collecting the terms $\E |\Alg|$ gives the desired bound:
    \begin{align*}
        2\E |\Alg| \geq \frac{9\alpha (1 - o(1))}{10C(1 + 1/s)^p \log^2 m} \cdot \OptBSP \qquad \Longrightarrow \qquad  \E|\Alg| \geq \Omega\left(\frac{\alpha}{(1 + 1/s)^p \log^2 m} \right) \cdot \OptBSP.
    \end{align*}

This completes the proof of \Cref{thm:pboundBudgetSched}, assuming \Cref{lem:opt-upper-bound} and \Cref{lem:unifiedLemma}, which we prove next. 
\end{proof}


    \subsubsection{Bounding $\OptBSP$ using Active Sets: Proof of \Cref{lem:opt-upper-bound}}
To start, we prove our upper bound on $\OptBSP$.
\begin{proof}[Proof of \Cref{lem:opt-upper-bound}]
    First, suppose that the algorithm ends with $f(y) \geq s^pB$. Then the sum of the marginals activation satisfies $\sum_{i \in A} \ba_i = f(y) - f(0) \ge s^p B$. This implies  $\sum_{i \in A} \frac{|\OptBSP| \cdot \ba_i}{s^p B} \ge |\OptBSP|$, and the lemma holds. 
    

    So, from now consider the case that the algorithm ends with $f(y) < s^p B$.
    
    Let $U_i$ be the set of elements which are assigned in $\OptBSP$ to machine $i$, but are unassigned to any machine in $\Alg$ (i.e., not picked by any of the single-machine online algorithms; these are exactly the items we don't get value from). Then $\OptBSP - |\Alg| \leq \sum_{i \in A^*} |U_i|$, where $A^*$ is the set of machines activated in the optimum assignment.

    Additionally, notice that $|U_i| = |\Opt_i(U_i)| \leq |\Opt_i(T_i)|$, since all of $U_i$ can pack into machine $i$ and $U_i \subseteq T_i$ since items in $U_i$ are offered to every machine. Moreover, since we assumed that the algorithm terminates before crossing threshold $s^pB$, it must be the case that each inactive machine $i \in [m] \setminus A$ has not been activated because it never crossed its threshold, i.e., $|\Opt_i(T_i)| < \tau_i$ at all points in time. Note that as $\Delta_i f(y^t)$ is non-decreasing and $\bar{\tau}_i \leq 1$, this implies that $|\Opt_i(T_i)| < \frac{|\OptBSP| \cdot \Delta_i f(\chi_A)}{10 \cdot (s+1)^p B}$. Therefore, we have
    \begin{align}
        |\OptBSP| - |\Alg| ~\leq~ \sum_{i \in A^*} |U_i| 
        ~& \leq \sum_{i \in A^*} |\Opt_i(T_i)|\notag\\
        &\leq \sum_{i \in A} |\Opt_i(T_i)| + \sum_{i \in A^* \setminus A} \frac{|\OptBSP| \cdot \Delta_i f(\chi_A)}{10 \cdot (s+1)^p B}\notag\\
        &\le \sum_{i \in A} |\Opt_i(T_i)| + \frac{|\OptBSP|}{10 \cdot (s+1)^p B} \cdot \sum_{i \in A^*}  \Delta_i f(\chi_A). \label{eq:sumDeltas}
    \end{align}
    Using the gradient monotonicity of $f$ (which implies that the marginal function $y \mapsto \Delta_i f(y)$ is an increasing function), we expand the last summation: consider the elements of $A^*$ in some order $i^*_1, i^*_2, \ldots$, and define the prefix $A^*_j = \{i^*_1,\ldots, i^*_j\}$; then
    \begin{align*}
    \sum_{i \in A^*}  \Delta_i f(\chi_A) &= \sum_j \bigg(f(\chi_A + e_{i^*_j}) - f(\chi_A)\bigg)\\
    &\stackrel{\star}{\le} \sum_j \bigg(f(\chi_A + \chi_{A^*_{j-1}} + e_{i^*_j}) - f(\chi_A + \chi_{A^*_{j-1}})\bigg)\\
    &= \sum_j \bigg(f(\chi_A + \chi_{A^*_j}) - f(\chi_A + \chi_{A^*_{j-1}})\bigg)\\
    &\le f(\chi_A + \chi_{A^*}) ~\le~  \Big(f(\chi_A)^{1/p} + f(\chi_{A^*})^{1/p}\Big)^p ~\le~ (s+1)^p B,
    \end{align*}
    where $\star$ uses monotonicity of gradient and the second-last inequality uses $p$-subadditivity of $p$-bounded functions
    (\Cref{lem:p-subadditive}). Since $A^*$ is feasible for the budget $B$ and we are in the case where $A$ does not cross the budget, we get $\sum_{i \in A^*}  \Delta_i f(\chi_A) \le (s+1)^{p} B$. Employing this on \eqref{eq:sumDeltas} and reorganizing, we get
    $\frac{9}{10} |\OptBSP| \le \sum_{i \in A} |\OPT_i(T_i)|. $
    Together with the previous case, this implies the lemma. 
\end{proof}

    \subsubsection{Bounding the Contribution of a Single Machine: Proof of \Cref{lem:unifiedLemma}}

    Now, we come to the core part of our argument. We will prove \Cref{lem:unifiedLemma} by lower bounding the contribution of a single machine to the algorithm's solution. 

\begin{proof}[Proof of \Cref{lem:unifiedLemma}]
    
    To prove the lemma, we will fix all randomness regarding other machines, that is for every $i' \neq i$, fix the choice of $\bar{\tau}_{i'}$, $M_{i'}$, and any internal randomness in the algorithm for machine $i'$. We summarize all these random outcomes in a scenario $\omega$. Our plan will be to lower-bound $\E[\Alg_i \mid \omega]$ for any fixed choice of such a scenario $\omega$.

    %

    Notice that the set $T_i$ of elements offered to machine $i$ is a function of only $\omega$ and $\bar{\tau}_i$ (i.e., it is independent of the choice of $M_i$ and the internal algorithm on machine $i$). This is because $T_i$ is determined only by the time at which machine $i$ is activated and the elements taken by higher priority machines. Our goal will be to better understand the dependence on $\bar{\tau}_i$. To this end, the following claim is helpful, showing monotonicity.
    
    \begin{claim} \label{claim:mono-new2}
    Fix any scenario $\omega$. The set $T_i = T_i(\bar{\tau}_i, \omega)$ is monotonically decreasing in $\bar{\tau}_i$ for any $\omega$ (i.e., the lower the threshold $\bar{\tau}_i$, the more items are offered to machine $i$). Moreover, $\ba_i(\bar{\tau}_i, \omega)$ is monotonically increasing in $\bar{\tau}_i$ (i.e., the higher the threshold $\bar{\tau}_i$, the longer it takes for col $i$ to get activated, so the higher it is its marginal cost at the time of activation, as the gradient of $f$ is non-decreasing).
    \end{claim}
    \begin{proof}[Proof of Claim \ref{claim:mono-new2}]
    Let $\bar{\tau}_i < \bar{\tau}_i'$. Consider the possible executions of the algorithm using thresholds $\bar{\tau}_i$ and $\bar{\tau}_i'$ for machine $i$, given $\omega$. Let $T_i$ and $T_i'$ respectively be the total items offered to machine $i$ in each run. We seek to show that $T_i \supseteq T_i'$.

    If in both executions, machine $i$ activates at the same time, then $T_i = T_i'$ and $\ba_i(\bar{\tau}_i, \omega) = \ba_i(\bar{\tau}_i', \omega)$. Thus, we may assume that at some arrival $t$, machine $i$ activates in only one of these executions. Since $\bar{\tau}_i < \bar{\tau}_i'$, it must be the case that $i$ activates at time $t$ with threshold $\bar{\tau}_i$, but not with threshold $\bar{\tau}_i'$. From this, we get the second part of our claim, since $\ba_i(\bar{\tau}_i, \omega) \leq \ba_i(\bar{\tau}_i', \omega)$ follows from the fact that $\ba_i$ only increases monotonically in the execution with $\bar{\tau}_i'$ after time $t$. 
    
    Let $A_{prev}$ be the set of active machines prior to time $t$, which is identical for both executions. Let $\Alg_{prev}$ and $\Alg_{prev}'$ denote the items assigned to machines in $A_{prev}$ by the end of each respective execution. Notice that from time $t$ onward, all items are offered to the machines of $A_{prev}$ \emph{before any other machine}. As a consequence, we have the following:
    \begin{itemize}
        \item $T_i = [n] \setminus \Alg_{prev}$, since any item not taken by $A_{prev}$ is then offered to $i$ in the execution with $\bar{\tau}_i$.
        \item $T_i' \subseteq [n] \setminus \Alg_{prev}'$, since any item taken by $A_{prev}$ is NOT offered to $i$ in the execution with $\bar{\tau}_i'$.
        \item $\Alg_{prev} = \Alg_{prev}'$ since the assignments to $A_{prev}$ are made independently of any other machine.
    \end{itemize}
    Therefore, we conclude $T_i \supseteq T_i'$, proving \Cref{claim:mono-new2}.
\end{proof}

    Notice that if $\Pr[i \in A \mid \omega] > 0$, then \Cref{claim:mono-new2} guarantees that there is some $\bar{\tau}_{max} = \bar{\tau}_{max}(\omega)$ such that $i \in A$ if and only if $\bar{\tau}_i \leq \bar{\tau}_{max}$. In addition, we seek to show that when $i$ activates, there is a constant chance that $\bar{\tau}_i$ is \textit{strictly} less than $\bar{\tau}_{max}$.
    
    \begin{claim}
    \label{lem:threshold-probabilities}
    In every scenario $\omega$, we have
    \[
    \Pr\bigg[\bar{\tau}_i \leq \bar{\tau}_{max} - \frac{1}{3 \log m} ~\bigg|~ \omega\bigg] \geq \frac{1}{2} \Pr[i \in A \mid \omega] - \frac{1}{2 m^3}.
    \]
    \end{claim}
    \begin{proof}[Proof of \Cref{lem:threshold-probabilities}]
    The intuition is that we want to avoid an event where, as soon as $i$ activates, no more jobs arrive which can be scheduled on $i$. However, if we know that $i$ activates with $\bar{\tau}_i < \bar{\tau}_{max}$ strictly, then there must be at least enough jobs arriving in the future that, had we used $\bar{\tau}_{max}$, we would still have activated $i$. Crucially, the distribution from which we sampled $\bar{\tau}_i$ guarantees that this happens with probability roughly $\frac{1}{2}$, and that the gap $\bar{\tau}_{max} - \bar{\tau}_{i}$ is large enough to effectly lower bound $|\Opt_i(O_i^{post})|$ in this event.
    
    Recall that $\bar{\tau}_i = (1 - \frac{k}{3 \log m})^+$ with probability $2^{-k-1}$ for $k \geq 0$. Therefore, let $\bar{\tau}_{max} = (1 - \frac{k'}{3 \log m})$ for some value of $k'$. Note that if $k' \geq 3 \log m-1$, then $\Pr[i \in A] = 2 \cdot 2^{-k'-1} \leq \frac{2}{m^3}$. So, the claim holds.
    Otherwise,  $\Pr[\bar{\tau}_i \leq \bar{\tau}_{max} - \frac{1}{3 \log m}] = \Pr[\bar{\tau}_i \leq (1 - \frac{k'+1}{3 \log m})] = \frac{1}{2} \Pr[\bar{\tau}_i \leq (1 - \frac{k'}{3 \log m})] = \frac{1}{2} \Pr[\bar{\tau}_i \leq \bar{\tau}_{max}] = \frac{1}{2} \Pr[i \in A]$.
\end{proof}

Using \Cref{lem:threshold-probabilities}, we can now obtain the first half of the lower bound on $\E|\Alg_i|$ in \Cref{lem:unifiedLemma}, which we get by examining the value the algorithm obtains when $\bar \tau_i < \bar \tau_{max}$.
\begin{claim}\label{lem:alg_i-bound1}
\begin{align}
    \E[\Alg_i]  \ge \E \left[\ind(i \in A) \cdot  \frac{\alpha \cdot \ba_i \cdot \OptBSP}{240 B \log^2 m \cdot (s+1)^{p}} \right]  - \frac{2\OptBSP}{m^2}.\label{eq:alg_i-bound1}
\end{align}
\end{claim}
    
    \begin{proof}
    We fix a scenario $\omega$. So, the only remaining randomness is stemming in $\bar{\tau}_i$, $M_i$, and the internal randomness of the algorithm of machine $i$.

    Let us first discuss the effect of different outcomes $\bar{\tau}_i \leq \bar{\tau}_{max}$. Note that we have\footnote{where we omit the scenario $(\bar{\tau}_i, \omega)$, e.g., $T_i = T_i(\bar{\tau}_i, \omega)$} $T_i \supseteq T_i(\bar{\tau}_{max}, \omega)$ by \Cref{claim:mono-new2}, so we have 
    \[
    |\Opt_i(O_i)| + |\Opt_i(O_i^{post})| \geq |\Opt_i(T_i)| \geq |\Opt_i(T_i(\bar{\tau}_{max}, \omega))|.
    \]
    Furthermore,
    \[
    |\Opt(T_i(\bar{\tau}_{max}, \omega))| \ge \bar{\tau}_{max} \cdot \frac{\ba_i(\bar{\tau}_{max}, \omega) \cdot \OptBSP}{10B \cdot (s+1)^{p}}
    \]
    because machine $i$ gets activated in scenario $(\bar{\tau}_{max}, \omega)$ and
    \[
    |\Opt(O_i)| \leq \bar{\tau}_i \cdot \frac{\ba_i(\bar{\tau}_i, \omega) \cdot \OptBSP}{10B \cdot (s+1)^{p}};
    \]
    because after $\Opt(O_i)$ has surpassed this quantity, no further elements will be added to $O_i$.
    
    %
    Since by Claim \ref{claim:mono-new2}, we have $\ba_i(\bar{\tau}_{max}, \omega) \ge \ba_i(\bar{\tau}_i, \omega)$, we get
    \begin{align}\label{eq:opt-opost-lowerb}
    |\Opt(O_i^{post})| \geq |\Opt(T_i(\bar{\tau}_{max}, \omega))| - |\Opt(O_i)| \geq (\bar{\tau}_{max} - \bar{\tau}_i) \cdot \frac{\ba_i(\bar{\tau}_{max}, \omega) \cdot \OptBSP}{10 B \cdot (s+1)^{p}}. 
    \end{align}

    Next, we take the expectation over $\bar{\tau}_i$ but keep the scenario $\omega$ fixed. Note that this does not affect $\ba_i(\bar{\tau}_{max}, \omega)$, since $\omega$ fixes $\bar{\tau}_{max}$. So, \Cref{lem:threshold-probabilities} implies
    \begin{align*}
    \E[|\Opt(O^{post}_i)| ~\mid~ \omega] & \geq \Pr\bigg[\bar{\tau}_i \leq \bar{\tau}_{max} - \frac{1}{3 \log m} ~\bigg|~ \omega\bigg] \cdot  \frac{\ba_i(\bar{\tau}_{max}, \omega) \cdot \OptBSP}{30 B \log m \cdot (s+1)^{p}} \\
    & \geq \frac{1}{2} \bigg(\Pr(i \in A \mid \omega) - \frac{1}{2m^3}\bigg) \cdot \frac{\ba_i(\bar{\tau}_{max}, \omega) \cdot \OptBSP}{30 B \log m \cdot (s+1)^{p}}.
    \end{align*}

    Recall that $\Alg_i$ obtains an $\alpha/2$ approximation of $\Opt(O_i^{post})$ as long as the guess $M_i$ is ``good,'' meaning $\frac{1}{2}|\Opt(O_i^{post})| \leq M_i \leq |\Opt(O_i^{post})|$. Moreover, this occurs with probability at least $\frac{1}{2\log m}$ unless $|\Opt(O_i^{post})| < \frac{\OptBSP}{m^2}$. Therefore, we get 
    %
    \begin{align*}
        \E[|\Alg_i| ~\mid~ \omega] &\ge \E\bigg[\frac{\alpha}{2} \cdot \frac{1}{2 \log m} \cdot \E[\Opt(O^{post}_i)] ~\Big|~ \omega\bigg] - \frac{\OptBSP}{m^2}\\
        & \ge \bigg(\Pr(i \in A \mid \omega) - \frac{1}{2m^3}\bigg)\cdot \frac{\alpha \cdot \ba_i(\bar{\tau}_{max}, \omega) \cdot \OptBSP}{240 B \log^2 m \cdot (s+1)^{p}} - \frac{\OptBSP}{m^2}.
    \end{align*}

    As argued before, we always have $\ba_i = \Delta_i f = f(\cdot + e_i) - f(\cdot) \le (s+1)^p B$, because $f(y + e_i) \le (f(y)^{1/p} + f(e_i)^{1/p})^p \le (s + 1)^p B$, the last inequality because of the stopping rule of the algorithm and because we assumed $f(e_i) \le B$ for all machines $i$. Then taking expectations w.r.t $\omega$ on the previous displayed inequality and using $\E_\omega[\Pr(i \in A \mid \omega) \cdot \ba_i(\bar{\tau}_{max},\omega)] \geq \E_{\omega, \bar\tau_i}[\ind({i \in A}) \cdot \ba_i]$, completes the proof of \Cref{lem:alg_i-bound1}.
\end{proof}

    Next, we prove the second half of \Cref{lem:unifiedLemma}, for which we seek to lower bound $\E |\Alg_i|$ in terms of $\Opt_i(T_i)$.

\begin{claim}\label{lem:alg_i-bound2}
    For every machine $i \in [m|$, we have
    \begin{align}
    \E|\Alg_i| \ge \E\bigg[\ind(i \in A) \cdot \frac{\alpha}{4 \log m} \cdot \left(|\Opt_i(T_i)| - \frac{\ba_i \cdot \OptBSP}{10B \cdot (s+1)^{p}}\right)^+\bigg] - \frac{\OptBSP}{m^2}.\label{eq:alg_i-bound2}
    \end{align}
\end{claim}

\begin{proof}
    For any outcome of the random draws, we have
    \[
    |\Opt_i(O^{post}_i)| \geq |\Opt_i(T_i)| - |\Opt_i(O_i)| \geq \left(|\Opt_i(T_i)| - \frac{\ba_i(\bar{\tau}_i, \omega) \cdot \OptBSP}{10B \cdot (s+1)^{p}}\right)^+.
    \]
    Notice that both the fact whether $i \in A$ and the set $T_i$ (and hence $\OPT_i(T_i)$) are independent of $M_i$. Therefore, the approximation guarantee of the algorithm, along with the probability of a ``good'' guess $M_i$, gives us
    \begin{align*}
    \E[|\Alg_i| ~\mid~ \omega, \bar{\tau}_i,  ] &\ge \ind(i \in A) \cdot \frac{\alpha |\OPT_i(O_i^{post})|}{4 \log m} - \frac{\OptBSP}{m^2} \\
    & \ge \ind(i \in A) \cdot \frac{\alpha}{4 \log m} \cdot \left(|\Opt_i(T_i)| - \frac{\ba_i(\bar{\tau}_i, \omega) \cdot \OptBSP}{10B \cdot (s+1)^{p}}\right)^+ - \frac{\OptBSP}{m^2}.
    \end{align*}
    Thus, taking expectations completes the proof of \Cref{lem:alg_i-bound2}.
\end{proof}

\paragraph{Putting the Pieces Together.}
    Now, we can combine  \Cref{lem:alg_i-bound1} and \Cref{lem:alg_i-bound2} to obtain \Cref{lem:unifiedLemma}.
    Averaging the bounds \eqref{eq:alg_i-bound1} and \eqref{eq:alg_i-bound2}, we now get 
    \begin{align*}
        \E|\Alg_i| \geq& \frac{1}{2}\E\left[\ind(i \in A) \cdot \left( \frac{\alpha \OptBSP\cdot \ba_i}{240 B \log^2 m \cdot (s+1)^{p}} + \frac{\alpha}{4 \log m} \cdot \left(|\Opt_i(T_i)| - \frac{\ba_i \cdot \OptBSP}{10B \cdot (s+1)^{p}}\right)^+ \right)\right]- \frac{2\OptBSP}{m^2}\\
        \geq& \frac{\alpha}{480\log^2 m}\E\left[\ind(i \in A) \cdot \left( \frac{\OptBSP\cdot \ba_i}{B \cdot (s+1)^{p}} + \left(|\Opt_i(T_i)| - \frac{\ba_i \cdot \OptBSP}{B \cdot (s+1)^{p}}\right)^+ \right)\right]- \frac{2\OptBSP}{m^2}\\
        \geq& \frac{\alpha}{480\log^2 m}\E\left[\ind(i \in A) \cdot  \max\left\{\frac{\OptBSP \cdot \ba_i}{B \cdot (s+1)^{p}}, ~|\Opt_i(T_i)|\right\}  \right]- \frac{2\OptBSP}{m^2} ,
    \end{align*}
    which proves \Cref{lem:unifiedLemma} with  constant $C = 480$.
\end{proof}

\subsection{Reducing \ref{eq:packGenSched} to \ref{eq:budgetGenSched}} \label{sec:reductionP}

Now with our algorithmic guarantees for \ref{eq:budgetGenSched} given by \Cref{thm:pboundBudgetSched}, we will complete the proof of \Cref{thm:pboundPackSched}, which  reduces \ref{eq:packGenSched} to \ref{eq:budgetGenSched}.

    Recall that the difference between the problems \ref{eq:packGenSched} and \ref{eq:budgetGenSched} is that in the latter each machine has a fixed activation budget (and incurred load) $b_i$, while in the former the machines may incur any amount of load, depending on the jobs scheduled. The idea of the reduction is to ``discretize'' the possible loads of the machines in \ref{eq:packGenSched} and create an instance of \ref{eq:budgetGenSched} with one copy of the machine with each of these budgets.

\begin{proof}[Proof of \Cref{thm:pboundPackSched}]Suppose we are given an instance $I_{SP}$ of \ref{eq:packGenSched} satisfying the conditions of \Cref{thm:pboundPackSched} (with load budget $B$, $p$-bounded aggregate function $f$, and  norms $\|\cdot\|_i$) and a lower bound $M \leq \Opt(I_{SP})$). We create an instance $I_{BSP}$ of  \ref{eq:budgetGenSched} as follows: for each machine $i \in [m]$ in $I_{SP}$, we make $\ceil{\log m}$ copies $(i,1), (i,2), \dots, (i, \ceil{\log m})$ in $I_{BSP}$ with the same internal norm $\|\cdot\|_{i,\ell} := \|\cdot\|_i$ and processing times $p_{(i,\ell) jk} := p_{ijk}$, but the $\ell$th copy has load budget $b_{i,\ell} = {b_i^{\max}}\cdot {2^{-\ell}}$, where
    \[b_i^{\max} = \sup\{b \geq 0 : f(b \cdot e_i) \leq B\};
    \]
    in other words, $b_i^{\max}$ is the largest cost that machine $i$ can incur without exceeding the aggregate load budget on its own. The aggregate function $f'$ of the new instance $I_{BSG}$ is defined by employing the original aggregate function $f$ over the sum of the copies of the machines
    \begin{align*}
    f'((z_{i,\ell})_{i,\ell}) = f\Big(\sum_\ell z_{1,\ell} \,,\, \sum_\ell z_{2,\ell} \,,\, \ldots \,,\, \sum_\ell z_{m,\ell}\Big).
    \end{align*}
    Finally, the new instance has aggregate load budget $3^p B$.

    We first claim that $\Opt(I_{BSP}) \geq \Opt(I_{SP})$. To see this, let $x^*$ denote the optimal scheduling for $I_{SP}$ (so $x^*_{ijk}$ indicates the scheduling of job $h$ onto machine $i$ in way $k$ in this solution). We then construct an scheduling for $I_{BSP}$ by making, for each machine $i$, the scheduling $(x^*_{ijk})_{j,k}$ on the copy $(i, \ell_i)$ of machine $i$ with the smallest budget $b_{i,\ell_i}$ satisfying $\|(x_{ijk}^* p_{ijk})_{j,k}\|_i \leq b_{i,\ell_i}$. This scheduling clearly has objective value $\Opt(I_{SP})$ (i.e., schedules as many elements), so we just need to show that it is a feasible schedule in $I_{BSP}$.

    Since in the constructed schedule, for each machine $i$ only the $\ell_i$th copy is used, its total load is given by 
    \begin{align}
        f'(0,\ldots, b_{1,\ell_1}, 0, \ldots, b_{2,\ell_2}, 0, \ldots, b_{m,\ell_m}) = f(b_{1,\ell_1}, b_{2, \ell_2}, \ldots, b_{m,\ell_m}). \label{eq:redBudget1}
    \end{align}
    Moreover, by the choice of each copy $\ell_i$ used, we have either $b_{i, \ell_i} \leq 2 \|(x_{ijk}^*p_{ijk})_{j,k}\|_i$ or $b_{i, \ell_i}$ is the minimum load budget copy of machine $i$, namely $\ell_i = \ceil{\log m}$ and so $b_{i, \ell_i} = \frac{b^{max}_i}{m}$. Therefore, the vector $(b_{1,\ell_1}, \ldots, b_{m,\ell_m})$ is coordinate-wise dominated by the vector $2 \Lambda(x^*) + \frac{1}{m} (b^{max}_1,\ldots,b^{max}_m)$ (recall that the $i$th coordinate of $\Lambda(x^*)$ is precisely $\|(x_{ijk}^*p_{ijk})_{j,k}\|_i$). Then using the $p$-subadditivity of $f$ (\Cref{lem:p-subadditive}) we have:
    \begin{align*}
        f(b_{1,\ell_1}, b_{2, \ell_2}, \ldots, b_{m,\ell_m}) \le \Big(f(\Lambda(x^*))^{1/p} + f(\Lambda(x^*))^{1/p} + f(\tfrac{1}{m} (b^{max}_1,\ldots,b^{max}_m))^{1/p} \Big)^p \le 3^p B,
    \end{align*}
    where the last inequality follows because each of the 3 terms in the middle sum is at most $B^{1/p}$: $f(\Lambda(x^*)) \le B$ because of the feasibility of $x^*$, and using convexity of $f$ and then the definition of $b^{max}_i$ we have $f(\tfrac{1}{m} (b^{max}_1,\ldots,b^{max}_m)) \le \frac{1}{m} \sum_{i=1}^m f(b^{max}_i \cdot e_i) \le B$. Therefore, the constructed schedule is indeed feasible for the instance $I_{BSP}$, which has aggregate load budget $3^p B$, and gives the claim $\Opt(I_{BSP}) \geq \Opt(I_{SP})$.

    Then we can use the algorithm for \ref{eq:budgetGenSched} given by \Cref{thm:pboundBudgetSched} with $s = p$ (call it $\cA_{BSP}$) to solve the instance $I_{BSP}$ as follows. It is easy to verify that since $f$ is $p$-bounded, the aggregate function $f'$ of the instance $I_{BSP}$ also has this property.
    Thus, we can employ the algorithm $\cA_{BSP}$ on the instance $I_{BSP}$, providing it $\ceil{M}$ as a guess on the optimum $\Opt(I_{BSP})$. Since this is a lower bound $\ceil{M} \leq \Opt(I_{SP}) \leq \Opt(I_{BSP})$, our algorithm obtains the guarantees of \Cref{thm:pboundBudgetSched} on some prefix of the arrivals, where the hindsight optimum is $\ceil{M}$. Hence, this algorithm yields a solution $\bar{x}, \bar y$ (where $\bar{x}_{(i,\ell) j k}$ indicated whether job $j$ is assigned to machine $(i,\ell)$ in way $k$, and $\bar y_{i, \ell}$ indicated whether machine $(i, \ell)$ is active) that schedules in expectation at least $\Omega\big(\frac{\alpha}{(1 + 1/p)^p \log^2 m}\big) \cdot M = \Omega\big(\frac{\alpha}{\log^2 m}\big) \cdot M$ jobs, incurring at most $(3p)^pB$ total load. By combining the schedules of all copies of each machine $i$, we can convert this solution $\bar{x}$ to a solution $\tilde{x}$ for the original instance $I_{SP}$ that schedules, i.e., $\tilde{x}_{ijk} := \sum_\ell \bar{x}_{(i,\ell)jk}$. We just need to verify that the solution $\tilde{x}$ blows up the aggregate budget of $I_{SP}$ by at most a factor of $(3pc)^p$, namely, that $f(\Lambda(\tilde{x})) \le (3pc)^pB$:
    \begin{align*}
    f(\Lambda(\tilde{x})) = f\Big((\|(\tilde{x}_{ijk} p_{ijk})_{j,k}\|_i)_i\Big) &\le f\Big(\sum_k \|(\bar{x}_{(1,\ell)jk}\cdot p_{1jk})_{j,k}\|_1\,,\, \ldots\,,\, \sum_k \|(\bar{x}_{(m,\ell)jk}\cdot p_{mjk})_{j,k}\|_m\Big) \\
    &= f'\Big( (\|(\bar{x}_{(i,\ell)jk} \cdot p_{ijk})_{j,k}\|_i)_{i,\ell}  \Big) \\
    &\le f'\Big( (c\bar y_{i,\ell} b_{i,\ell})_{i,\ell}  \Big) ~\le~  f'\Big( (\bar y_{i,\ell} b_{i,\ell})_{i,\ell}  \Big) \cdot c^p ~\le~ (3pc)^p B,
    \end{align*}
    where the first inequality follows from the subadditivity of norms (i.e., $\|(\tilde{x}_{ijk} p_{ijk})_{j,k}\|_i = \|\sum_\ell (\bar{x}_{(i,\ell)jk} \cdot p_{ijk})_{j,k}\|_i \le \sum_\ell \|(\bar{x}_{(i,\ell)jk} \cdot p_{ijk})_{j,k}\|_i$), the second equality holds by definition of the function $f'$, and the last inequality by the approximate feasibility of $\bar{x}$ for the instance $I_{BSP}$. Thus, $\tilde{x}$ is the desired solution, which concludes the proof of the theorem.
\end{proof}

\medskip


%% file: norms.tex
We now design competitive algorithms for the \ref{eq:covGenSched} problem when the aggregate function $f$ is a norm. Recall again that the norm setting is not directly captured by the $p$-bounded setting from the previous section because, except for $\ell_1$, norms do not have monotone gradients. To differentiate the aggregate function norm from the norms on the machines, in this section we will call them \emph{outer norm} and \emph{inner norms}, respectively.

Since every norm is $1$-subadditive,  our reduction from \Cref{sec:reduction}  implies that \ref{eq:covGenSched} can be reduced to \ref{eq:packGenSched}: more precisely, if \ref{eq:packGenSched} is $(\alpha,c)$-solvable (i.e., given a lower bound $M$ on $\OPT_{SP}$ we can schedule $\alpha M$ jobs while violating the budget $B$ by a factor of $c$), we can obtain a $O(\frac{1}{\alpha} \cdot \log n \cdot c)$-competitive algorithm for \ref{eq:covGenSched}. This section  focuses on solvability of  \ref{eq:packGenSched}  for different classes of norms. 

\subsection{\ref{eq:packGenSched} with Norms and their Compositions} \label{sec:normComp}

Recall that \ref{eq:packGenSched}, with outer norm $\|\cdot\|$, is the problem whose offline version is given by 
\[ \textstyle \max\{\sum_{i,j,k} x_{ijk}  \text{ s.t. } \|\Lambda(x)\| \le B~,~ \sum_{i,k} x_{ijk} \le 1~,~ x \in \{0,1\}^{mnr}\},\] 
where $\Lambda_i(x) = \|(x_{ijk} p_{ijk})_{j,k}\|_i$ is the load on  $i$-th machine. 

The following is the first main result of this section.

\begin{theorem}\label{thm:genBudgetNorm}
    Consider the online problem \ref{eq:packGenSched}$_{\|\cdot\|, \{\|\cdot\|_i\}_i}$, with outer norm $\|\cdot\|$, where for each inner norm $\|\cdot\|_i$ the problem \ref{eq:normPack}$_{\|\cdot\|_i}$ is $(\alpha, c)$-solvable. 
    Then \ref{eq:packGenSched}$_{\|\cdot\|, \{\|\cdot\|_i\}_i}$ is $(\alpha', O(c))$-solvable, where $\alpha'$ is given for different norms $\|\cdot\|$ below:
    \begin{enumerate}
        \item For (weighted) symmetric norms, we obtain $\alpha' =\Omega(\alpha / \log^2 m)$. 
        
        \item For the special case of the symmetric norm being $\ell_p$ or Top-$k$ norms, we obtain $\alpha' = \Omega(\alpha / \log m)$. 
    \end{enumerate}
\end{theorem}

Together with the \ref{eq:covGenSched}-to-\ref{eq:packGenSched} reduction from \Cref{thm:packCovReduction}, this result directly implies \Cref{thm:schedNorm} stated in the introduction.

To prove \Cref{thm:genBudgetNorm}, we follow the same two-step approach as in \Cref{sec:p-bound}: first step is to reduce \ref{eq:packGenSched} to \ref{eq:budgetGenSched} and the second step is to design a bi-criteria approximation algorithm for  \ref{eq:budgetGenSched}. The proof of the reduction in the first step is quite similar to \Cref{sec:p-bound}, but the precise details change since for norms we can obtain much better guarantees in this reduction.  However, the proof of the second step is differs a lot since we don't have gradient monotonicity. 

To overcome this, in \Cref{sec:BSPl1} we first design a good algorithm for \ref{eq:budgetGenSched} in the special case of a weighted $\ell_1$ outer norms $\|z\| = \sum_i w_i \cdot |z_i|$, and then in the later sections we give reductions from 
$\ell_p$, Top-$k$, and symmetric norms to a weighted $\ell_1$-norm to prove \Cref{thm:genBudgetNorm}. 
 This reduction may seem surprising because in general it is not possible to approximate general norms everywhere with a weighted $\ell_1$-norm within polylog factors; e.g., even $\ell_2$ in $m$ dimensions has a $\sqrt{m}$ multiplicative factor gap to all weighted $\ell_1$ norms. Our proof crucially exploits the structure of \ref{eq:budgetGenSched}, in particular that we are only interested in the integral problem and the norm only  appears in the budget constraint. 
 This allows us to discretize the norm and 
 approximate the feasibility constraint using a weighted $\ell_1$ norms after slightly expanding the budget constraint.

\medskip \textbf{Norm composition.}  
In fact, our techniques for \Cref{thm:genBudgetNorm}  extend  to a much larger class of norms, via norm compositions. Namely, suppose that we have a norm $\|\cdot\| : \R^m \to \Rp$ that can be written as a \emph{disjoint composition} of a norm $\|\cdot\|'$ with norms $\|\cdot\|'_1, \dots, \|\cdot\|'_L$. That is, for some disjoint subsets $S_1, \dots, S_L \subseteq [m]$, we have 
\[ \|y\| := \Big\|\big(\|y_{S_1}\|_1', \dots, \|y_{S_L}\|_L' \big)\Big\|'.\]
Then as long as $\|\cdot\|'$ and all $\|\cdot\|'_1, \dots, \|\cdot\|'_L$ norms are ``good'' outer norms for \ref{eq:packGenSched} (i.e. we have a statement like \Cref{thm:genBudgetNorm} for them), we can prove that the composite norm $\|\cdot\|$ is also a ``good'' outer norm for \ref{eq:packGenSched}. From the base cases set forth in \Cref{thm:genBudgetNorm}, this allows us to generate many other outer norms for which we can solve \ref{eq:packGenSched}, and by extension, \ref{eq:covGenSched}. For Set Cover settings (i.e. $\|\cdot\|_i = \|\cdot\|_\infty$), this result becomes even more powerful when combined with \cite{NS-ICALP17}, as we can then extend to the settings where the $S_\ell$ are not disjoint.

The idea behind this result is that we can run two layers of our \ref{eq:packGenSched} algorithm. Call the $S_\ell$ sets \emph{clusters}. First, we have an outer layer, which allocates jobs among the clusters $S_1, \dots, S_L$, minimizing the outer norm $\|\cdot\|'$ over the vector of loads on each cluster. Additionally, within a cluster $S_\ell$, we have an inner algorithm which seeks to schedule the jobs it receives among the machines $i \in S_\ell$, in a way that minimizes the $\|\cdot\|'_\ell$-load of the cluster. Intuitively, as long as we have a good algorithm for both levels of scheduling, we can obtain a good algorithm for the composite problem.

To make this notion formal, we first define what it means for a norm to be a good outer norm in \ref{eq:packGenSched}.
\begin{definition}\label{def:good-norm}
    We say norm $\|\cdot\| : \R^m \to \Rp$ is a \emph{$(\beta, \gamma)$-good outer norm} if \ref{eq:packGenSched}$_{\|\cdot\|, \{\|\cdot\|_i\}_i}$ is $(\beta\alpha, \gamma c)$-solvable whenever the inner-norm problems \ref{eq:normPack}$_{\|\cdot\|_i}$ are all $(\alpha, c)$-solvable.
\end{definition}

For example, by \Cref{thm:genBudgetNorm} we have that symmetric norms are $\big(\Omega({1}/{\log^2 m}),~O(1)\big)$-good outer norms. Using this theorem as a base case, we can construct a larger class of good outer norms using the following composition result.

\begin{theorem}\label{thm:norm-compose-sched-pack}
    Suppose $\|\cdot\| : \R^m \to \Rp$ is a disjoint composition of a norm $\|\cdot\|'$ 
     with norms $\|\cdot\|'_1, \dots, \|\cdot\|'_L$.
    Then, if $\|\cdot\|'$ is a $(\beta_1, \gamma_1)$-good outer norm, and all $\|\cdot\|'_\ell$ are $(\beta_2, \gamma_2)$-good outer norms for each $\ell \in [L]$, then $\|\cdot\|$ is a $(\beta_1\beta_2,~\gamma_1\gamma_2)$-good outer norm.
\end{theorem}

 Combining this composition theorem with \Cref{thm:genBudgetNorm}, we obtain \Cref{thm:genSchedComposition} and \Cref{cor:setCoverComp} stated in the introduction, although in the latter we also need an argument from~\cite{NS-ICALP17} to handle non-disjoint compositions; details are given in \Cref{app:proofsIntro}.
 
\medskip \textbf{Roadmap for the remaining section.} 
In \Cref{sec:BSPl1}, we design a bi-criteria approximation algorithm for \ref{eq:budgetGenSched} and prove  \Cref{thm:genBudgetNorm} in the special case of weighted $\ell_1$-outer norms, which will form our basic building block. In Sections \ref{sec:sym-reduction} and \ref{sec:lp-topk-reduction},
we design bi-criteria approximation algorithms for \ref{eq:budgetGenSched} with the norms given in \Cref{thm:genBudgetNorm} by reducing them to weighted $\ell_1$-outer norm. 
Finally, in \Cref{sec:norm-compose}, we will prove our norm composition result in \Cref{thm:norm-compose-sched-pack}.

\subsection{\ref{eq:budgetGenSched} with Weighted $\ell_1$ Outer Norm} \label{sec:BSPl1}

Consider the \ref{eq:budgetGenSched} problem with a weighted $\ell_1$ outer norm $\|z\| = \sum_i w_i \cdot |z_i|$. Observe that this function is $1$-bounded, so \Cref{thm:pboundBudgetSched} already implies an approximation for our problem. However, we can actually save a $\log m$ factor in this case and obtain the following result.

\begin{theorem}\label{thm:l1-budget}
    Consider the online problem \ref{eq:budgetGenSched}$_{\|\cdot\|,\{\|\cdot\|_i\}_i}$ where the outer norm $\|\cdot\|$ is a weighted $\ell_1$ norm and for each inner norm $\|\cdot\|_i$ the problem \ref{eq:normPack}$_{\|\cdot\|_i}$ is $(\alpha, c)$-solvable. Then, there is an online algorithm for \ref{eq:budgetGenSched}$_{\|\cdot\|,\{\|\cdot\|_i\}_i}$ with the following guarantee: Given a lower bound $M \leq \Opt$ up front, the algorithm obtains expected objective value at least $\Omega\big(\frac{\alpha}{\log m}\big) \cdot M$ while violating each machine's load budget by at most a factor $c$ (i.e. so that $\|(x_{ijk} p_{ijk})_{j,k}\|_i \leq cy_i b_i$).
\end{theorem}

To lighten the notation, we define $a_i := w_i b_i$, so the budget constraint $\|(y_1 b_1,\ldots,y_m b_m)\| \le B$ can be written as $\sum_i a_i y_i \le B$; so the offline version of the problem is $\max\{\sum_{i,j,k} x_{ijk} : \sum_i a_i y_i \le B,~ \Lambda_i(x) \le y_i b_i~\forall i,~ \sum_{i,k} x_{ijk} \le 1, x \in \{0,1\}^{mnr}, y \in \{0,1\}^m\}$, where $\Lambda_i(x) = \|(x_{ijk} p_{ijk})_{j,k}\|_i$.

 To save the $\log m$ factor, we only need to modify the \Cref{alg:p-bounded} with respect to how we guess $M_i$ for each machine. At a very high level, in \Cref{lem:alg_i-bound1} in the previous analysis we were losing a $\log m$ factor since we had to condition on the event that $M_i$ was a ``good'' guess on {$|\Opt_i(O_i^{post})|$}, which only happens with probability $\frac{1}{\Theta(\log m)}$. However, in the proof of this lemma, we only truly need to have $M_i = \Theta(\frac{\ba_i(\bar \tau_{max}, \omega) \cdot \Opt}{B\log m})$, where $\ba_i$ is the discrete derivative $\Delta_i f(y)$ at the current machine activation $y$. This value is difficult to guess when $\ba_i$ may change dynamically, but in the weighted $\ell_1$ setting we have $\ba_i = a_i$ statically, so we can choose the appropriate value of $M_i$ with constant probability. 
 
 We now present the new \Cref{alg:norm} (with the changes relative to \Cref{alg:p-bounded} marked in {\color{blue}\textbf{blue}}). As before, the \ref{eq:normPack}$_{\|\cdot\|_i}$ algorithm employed in each machine $i$ is one that achieves the assumed $(\alpha, c)$-solvability. We also assume that each machine $i$ has $a_i \le B$, else due to the budget constraint this machine is not activated in the optimal solution and can be ignored.

\begin{algorithm}
    \caption{Algorithm for \ref{eq:budgetGenSched} with weighted $\ell_1$ outer norm} \label{alg:norm}

First, for each machine $i$, draw independently a threshold $\bar{\tau}_i$ distributed as $(1 - \frac{k}{3\log m})^+$ with probability $2^{-k-1}$ (for $k \ge 0$), and set its threshold to be $\tau_i := \bar{\tau}_i \cdot \blue{\frac{1}{20 B}\cdot \OptBSP \cdot a_i}$. Then for each arriving job $j$:
\begin{enumerate}
    \item Offer it to each active machine $i$ in order of activation time, i.e. add job $j$ to $O^{post}_i$. If the internal \ref{eq:normPack}$_{\|\cdot\|_i}$ algorithm of this machine schedules this job, set the corresponding $x_{ijk} = 1$ and add $j$ to $\Alg_i$. Else offer $j$ to the next active machine, etc.
    \medskip

    \item If the item is not selected by any active machine, then offer it to all inactive machines (e.g. in lexicographical order). In other words, add $j$ to the $O_i$'s of each inactive machine $i$ one at a time. Continue this until some machine activates and schedules $j$, or until all machines have been offered $j$.
    \medskip

    \item If an inactive machine $i$ is offered $j$ and crosses the threshold with its offered items, i.e., $\OPT(O_i, b_i) \ge \tau_i$, AND \blue{$\sum_i a_i y_i < B$}, then do the following: 
    \begin{enumerate}
        \item Activate this machine (i.e., add $i$ to $A$ and update $y_i = 1$).

        \item Draw a guess $M_i$ from the distribution: 
               {\color{blue} \[M_i = 
                \begin{cases}
                    {\OptBSP \cdot a_i}\cdot ({60 B \log m})^{-1} & \textrm{w.p. } \frac{1}{2},\\
                     {2^{-k}}\cdot {\OptBSP} & \textrm{w.p. } \frac{1}{2\floor{2\log m}}, ~\forall k \in \{1, \dots, \floor{2 \log m}\}.
                \end{cases}
                \]}
        
        \item Start an internal \ref{eq:normPack}$_{\|\cdot\|_i}$ algorithm for machine $i$ with guess $M_i$. This algorithm will receive only jobs offered to $i$, i.e. jobs in $O_i^{post}$.
        \item Offer $j$ to machine $i$. If $i$ schedules $j$, set $x_{ijk} = 1$, add $j$ to $\Alg_i$, and skip to next job $j+1$. Else, continue Step 2, i.e., offer to the next inactive machine.
    \end{enumerate}
    \medskip

    \item Stop after all $n$ items have arrived.
\end{enumerate}
 \end{algorithm}

It is clear that \Cref{alg:norm} satisfies the budget violation requirements for \Cref{thm:pboundBudgetSched}, so it only remains to show that it gets the required approximation factor of $\OptBSP$. Much of the analysis is the same as in \Cref{sec:p-bound}, since they do not depend on the choice of the $M_i$ and, in all other ways, \Cref{alg:norm} is identical to \Cref{alg:p-bounded} with $p = s = 1$. Thus, for brevity, we only highlight the modifications in the argument in order to prove \Cref{thm:l1-budget}. Since these modification are mostly syntactical, we only recall the following necessary elements: 1) $\Alg_i$ is the set of jobs that the algorithm schedules in machine $i$; 2) $O_i^{post}$ is the instance of the \ref{eq:normPack}$_{\|\cdot\|_i}$ algorithm employed in machine $i$; 3) When the estimate $M_i$ is lower bound on the opt of this instance, namely $M_i \le \OPT(O_i^{post})$, the guarantee of the \ref{eq:normPack}$_{\|\cdot\|_i}$ algorithm kicks in and it schedules at least $\alpha M_i$ jobs on this machine. 

The key change is the following alternative to \Cref{lem:unifiedLemma} that saves the aforementioned $\log m$ factor. 

\begin{lemma}\label{lem:normUnifiedLemma}
    In \Cref{alg:norm}, for some constant $C > 0$ we have for every machine $i \in [m]$.
    \[
    \E\, |\Alg_i| \geq \frac{\alpha}{C\log m} \cdot \E \left[ \ind(i \in A) \max\left\{\frac{ \OptBSP \cdot a_i}{B},~ |\Opt_i(T_i)|\right\}\right] - \frac{\OptBSP}{m^2}.
    \] 
\end{lemma}

\begin{proof}
    Recall that \Cref{lem:unifiedLemma} was proven by combining \eqref{eq:alg_i-bound1} and \eqref{eq:alg_i-bound2}. Moreover, we still have \eqref{eq:alg_i-bound2} with an additionally factor $\frac{1}{2}$ loss, as \Cref{alg:norm} still has probability $\frac{1}{2}$ of sampling $M_i$ the same way as in \Cref{alg:p-bounded}. This gives
    \begin{align}
    \E|\Alg_i| \ge \E\bigg[\ind(i \in A) \cdot \frac{\alpha}{8 \log m} \cdot \left(|\Opt_i(T_i)| - \frac{a_i \cdot \OptBSP}{20B}\right)^+\bigg] - \frac{\OptBSP}{m^2}.\label{eq:norm-alg_i-bound2}
    \end{align}

    Next, we obtain an alternative to \eqref{eq:alg_i-bound1}. We begin with \eqref{eq:opt-opost-lowerb}, which holds independent of the method of choosing $M_i$:
    \[
    |\Opt(O_i^{post})| \geq (\bar{\tau}_{max} - \bar{\tau}_i) \cdot \frac{a_i \cdot \OptBSP}{20 B}.
    \]
    Hence, if $\bar \tau_i$ is strictly less than $\bar \tau_{\max}$, and $\bar \tau_{\max} \geq \frac{1}{3\log m}$, then we have $|\Opt(O_i^{post})| \geq \frac{a_i \cdot \OptBSP}{60 B \log m}$. Thus, if in addition we have $M_i = \frac{a_i \cdot \OptBSP}{60B \log m}$ (which happens independently with probability $\frac{1}{2}$), $M_i$ is a lower bound on the optimum $|\Opt(O_i^{post})|$ for the $i$th machine and the guarantee for the internal algorithm for \ref{eq:normPack}$_{\|\cdot\|_i}$ used in this machine kicks and it schedules at least $\alpha M_i = \alpha \frac{a_i \cdot \OptBSP}{60B \log m}$ jobs. Therefore:    
    \begin{align}
        \E |\Alg_i| &\geq \frac{1}{2}\,\E\left[\ind_{\bar \tau_{\max} \geq \frac{1}{3\log m}} \cdot \ind_{\bar \tau_i < \bar \tau_{\max}} \cdot \alpha \frac{a_i \OptBSP}{60 B \log m}\right] \notag\\
        &\geq \frac{1}{2}\,\E\left[\ind_{\bar \tau_i < \bar \tau_{\max}} \cdot \alpha \frac{a_i \OptBSP}{60 B \log m}\right] - \Pr\bigg[\bar \tau_i < \bar \tau_{\max} < \frac{1}{3\log m}\bigg] \cdot \frac{\OptBSP}{2} \notag\\
        &\geq \frac{1}{2}\Pr[\bar \tau_i \leq \bar \tau_{\max}]\cdot \alpha\frac{a_i \OptBSP}{60 B \log m} - \frac{\OptBSP}{m^3} \notag\\
        &= \Pr[i \in A] \cdot \alpha \frac{a_i \OptBSP}{120 B \log m} - \frac{\OptBSP}{m^3}, \label{eq:norm-alg_i-bound1}
    \end{align}
    where the second inequality uses the fact that $a_i \le B$, and the following inequality uses the fact that $\Pr(\bar{\tau}_i < \frac{1}{3 \log m}) \le \frac{2}{m^3}$. 
    Finally, averaging bounds \eqref{eq:norm-alg_i-bound2} and \eqref{eq:norm-alg_i-bound1}  gives
        \begin{align*}
        \E|\Alg_i| &\ge \frac{1}{2}\,\E\left[\ind(i \in A) \cdot \left( \frac{\alpha a_i \cdot \OptBSP}{120 B \log m} + \frac{\alpha}{8 \log m} \cdot \left(|\Opt_i(T_i)| - \frac{a_i \cdot \OptBSP}{20B}\right)^+ \right)\right]- \frac{\OptBSP}{m^2}\\
        &\ge \frac{1}{2}\,\E\left[\ind(i \in A) \cdot \left(\frac{\alpha}{6\log m} \frac{a_i \cdot \OptBSP}{20 B} + \frac{\alpha}{8 \log m} \cdot \left(|\Opt_i(T_i)| - \frac{a_i \cdot \OptBSP}{20B}\right)^+ \right)\right]- \frac{\OptBSP}{m^2}\\
        &\ge \frac{\alpha}{16\log m} \cdot\E\left[\ind(i \in A) \cdot  \max\left\{\frac{a_i \OptBSP}{20B}, ~|\Opt_i(T_i)|\right\} \right]- \frac{\OptBSP}{m^2}.
    \end{align*}
    The lemma then holds with constant $C = 16 \cdot 20$, thus concluding the proof. 
\end{proof}

In addition to the single-machine control provided by the previous lemma, the upper-bound on $\OptBSP$ that \Cref{lem:opt-upper-bound} provides for \Cref{alg:p-bounded} still holds for the current algorithm  \Cref{alg:norm}, namely in every scenario we have
    \begin{align}
    \frac{9}{10}\OptBSP \leq |\Alg| + \sum_{i \in A} \max \left\{\frac{\OptBSP \cdot a_i}{B},~ |\Opt_i(T_i)| \right\}. \label{eq:oldLemma}
    \end{align}

To complete the proof of \Cref{thm:l1-budget}, we again sum \Cref{lem:normUnifiedLemma} over all machines $i \in [m]$ and apply \eqref{eq:oldLemma}:
  \begin{align*}
        \E |\Alg| 
        &\geq \frac{\alpha}{C \log m}\E\left[  \sum_{i \in A} \max\left\{\frac{\OptBSP \cdot a_i}{B},~ |\Opt_i(T_i)|\right\}\right] - \frac{\OptBSP}{m}\\
        &\geq \frac{\alpha}{C \log m} \cdot \bigg(\frac{9}{10}\OptBSP - \E|\Alg|\bigg) - \frac{\OptBSP}{m}, 
    \end{align*}
    and collecting the terms $\E |\Alg|$ gives $\E|\Alg| \geq \Omega(\frac{\alpha}{\log m}) \cdot \OptBSP$ as desired. This concludes the proof of \Cref{thm:l1-budget}.



\subsection{\ref{eq:packGenSched} with Symmetric Outer Norm}\label{sec:sym-reduction}

    We now prove \Cref{thm:norm-compose-sched-pack}, namely we design a competitive algorithm for \ref{eq:packGenSched} with symmetric outer norm. So consider an instance $I_{SP}$ of the \ref{eq:packGenSched} problem where the outer norm is given by a symmetric norm $\|\cdot\|$, i.e., we want to find online a schedule of jobs to machines $x$ so that $\max\{\sum_{ijk} x_{ijk} : \|\Lambda(x)\| \le B, \sum_i x_{ijk} \le 1 ~~\forall j, x \in \{0,1\}^{m n r}\}$.
    
    We solve this problem by reducing it to \ref{eq:budgetGenSched} with (weighted) $\ell_1$ outer norm. The key element for this is the claim that there is a budget $\bar{b}$ that we can set for the load of each of the machines and still at least a $\frac{1}{\log m}$ fraction of $\OPT(I_{SP})$. To make this precise, for a scalar $b \ge 0$, let $I_{BSP-symm}(b)$ denote the instance of \ref{eq:budgetGenSched} with: 1) the same jobs outer and inner norms as $I_{SP}$; 2) with outer load budget $2B$ instead of $B$; 3) where all machines have a load budget $b$, namely, whose offline version is $\max\{\sum_{ijk} x_{ijk} : \|(y_1 b, y_2 b, \ldots, y_m b)\| \le 2B, \Lambda_i(x) \le b y_i ~~\forall i, \sum_{ik} x_{ijk} \le 1 ~~\forall j, x \in \{0,1\}^{m n r}, y \in \{0,1\}^m.\}$.
    
\begin{claim} \label{claim:symm}
    Let $\tilde{B} = \frac{B}{\|e_i\|}$ (which is the same for every canonical vector $e_i$, since $\|\cdot\|$ is symmetric), and consider the $\log m$ values $\{\tilde{B}, \frac{\tilde{B}}{2}, \frac{\tilde{B}}{4}, \ldots, \frac{\tilde{B}}{2^{\log m - 1}}\}$. Then $\sum_{\bar{b} \in \{\tilde{B}, \frac{\tilde{B}}{2}, \frac{\tilde{B}}{4}, \ldots, \frac{\tilde{B}}{2^{\log m - 1}}\}} \OPT(I_{BSP-symm}(\bar{b})) \ge \OPT(I_{SP})$.
\end{claim}

\begin{proof}
    Consider the optimal offline solution $x^*$ of $I_{SP}$. For $\ell = 0,1,\ldots, \log m - 1$, define $G_\ell$ as the group of machines $i$ where $x^*$ has load $\Lambda_i(x^*)$ between $\frac{\tilde{B}}{2^{\ell+1}}$ and $\frac{\tilde{B}}{2^\ell}$, and define the group $G_{\log m - 1}$ of machines with $\Lambda_i(x^*) \le {\tilde{B}}/{2^{\log m - 1}}$.

    Fix any group $G_\ell$ and let $\bar{b} = \frac{\tilde{B}}{2^{\ell}}$. We claim that $\OPT(I_{BSP-symm}(\bar{b}))$ is at least the number of jobs that the optimal solution $x^*$ schedules on the machines in group $G_\ell$, i.e., at least $\sum_{i \in G_{\ell}, j,k} x^*_{ijk}$; to proves this, let $\bar{x}$ denote the assignment that $x^*$ makes of jobs into the machines in $G_{\ell}$ (and let $\bar{y}_i = 1$ iff $i \in G_{\ell}$ be the indicator of activating all machines in this group), so it suffices to show that $(\bar{x}, \bar{y})$ is a feasible solution for $I_{BSP-symm}(\bar{b})$.

    For that, first, the budget $\bar{b}$ of the machines in the instance $I_{BSP-symm}(\bar{b})$ are respected by the schedule $\bar{x}$, since it just mimics the assignment of $x^*$ onto the machines $i$ in group $G_{\ell}$, which by definition have load $\Lambda_i(x^*) \le \frac{\tilde{B}}{2^{\bar{\ell}}} = \bar{b}$. For feasibility of the outer budget, namely $\|(\bar{y}_1 \bar{b}, \ldots, \bar{y}_m \bar{b})\| \le 2B$: if $\bar{\ell} \neq \log m - 1$, then for every machine we have the load bound $\bar{y}_i \bar{b} \le 2 \Lambda_i(\bar{x})$ (if $i \notin G_{\bar{\ell}}$ then both sides are 0, and if $i \in G_{\bar{\ell}}$ the RHS is $\bar{b}$ and the LHS is at least $2 \cdot \frac{\tilde{B}}{2^{\bar{\ell} - 1}} = \bar{b}$); thus, $\|(\bar{y}_1 \bar{b}, \ldots, \bar{y}_m \bar{b})\| \le \|2 \Lambda(\bar{x})\| \le \|2 \Lambda(x^*)\| \le 2B$, the last inequality by feasibility of $x^*$ in $I_{SP}$, and the claim holds. If $\bar{\ell} = \log m - 1$, then $\bar{b} = \frac{2\tilde{B}}{m}$ and so using convexity of norms $\|(\bar{y}_1 \bar{b}, \ldots, \bar{y}_m \bar{b})\| \le \|(\frac{2\tilde{B}}{m}, \ldots, \frac{2\tilde{B}}{m})\| \le \frac{1}{m} \sum_{i = 1}^m \|2\tilde{B} \cdot e_i\| = 2B$, the last equation by the definition of $\tilde{B}$. Thus, in both cases of $\bar{\ell}$ we have that $(\bar{x}, \bar{y})$ is a feasible solution for  $I_{BSP-symm}(\bar{b})$, proving the desired result. 

    Adding this result over all groups $G_\ell$, we get  
    \[ \sum_{\bar{b} \in \{\tilde{B}, {\tilde{B}}/{2}, {\tilde{B}}/{4}, \ldots, {\tilde{B}}/{2^{\log m - 1}}\}} \OPT(I_{BSP-symm}(\bar{b})) ~~\ge~~ \sum_\ell \sum_{i \in G_\ell, j, k} x^*_{ijk} ~~= ~~\OPT(I_{SP}),\] 
    which concludes the proof of the claim. 
\end{proof}

    The second key observation is that because the outer norm $\|\cdot\|$ is symmetric, and because all machine budgets are the same, the instance $I_{BSP-symm}(\bar{b})$ is equivalent to an instance $I_{BSP-\ell_1}(\bar{b})$ of \ref{eq:budgetGenSched} with $\ell_1$ outer norm (over which we can apply the algorithm from \Cref{sec:BSPl1}). More precisely, let $\kappa := \max\{ k : \|(\underbrace{\bar{b},\ldots,\bar{b}}_{k}, 0, \ldots,0)\| \le 2B \}$ be the maximum number of machines that can be activated and remain within budget $2B$. Then let $I_{BSP-\ell_1}(\bar{b})$ be the instance that is identical to $I_{BSP-symm}(\bar{b})$ but with the $\ell_1$ outer budget constraint $\sum_i y_i \le \kappa$ instead of $\|(y_1 \bar{b}, \ldots, y_m \bar{b})\| \le 2B$. The symmetry of $\|\cdot\|$ immediately gives the equivalence of these constraints, and thus, of these instances: for any $y \in \{0,1\}^m$,
    \begin{align}
        \sum_i y_i \le \kappa  \quad \Longleftrightarrow\quad \|(y_1 \bar{b}, \ldots, y_m \bar{b})\| \le 2B \enspace .\label{eq:symmL1}
    \end{align}

    To solve $I_{SP}$, we then to the following: 1) Pick uniformly at random a value $\bar{b}$ in the set $\{\tilde{B}, \frac{\tilde{B}}{2}, \frac{\tilde{B}}{4}, \ldots, \frac{\tilde{B}}{2^{\log m - 1}}\}$; 2) Run the algorithm from \Cref{thm:l1-budget} over the instance $I_{BSP-\ell_1}(\bar{b})$, obtaining a solution $(\bar{x},\bar{y})$.

    First, we check that the schedule $\bar{x}$ is a solution for the original \ref{eq:packGenSched} instance $I_{SP}$ that: (1) Violates the outer budget by a factor of at most $2c$; (2) schedules at least $\frac{\alpha}{\log^2 m} \OPT(I_{SP})$ jobs (the parameters $\alpha,c$ come from the assumed  $(\alpha, c)$-solvability of \ref{eq:normPack}$_{\|\cdot\|_i}$ for the inner norms $\|\cdot\|_i$). To verify (1), 
    \Cref{thm:l1-budget} and the equivalence from \eqref{eq:symmL1} guarantee that $\|(\bar{y}_1 \bar{b}, \ldots, \bar{y}_m \bar{b})\| \le 2B$ and $\Lambda_i(\bar{x}) \le c\bar{y}_i \bar{b}$, which combined give $\|\Lambda(\bar{x})\| \le 2c B$, giving item (1). For Item (2), \Cref{thm:l1-budget} guarantees that $\bar{x}$ schedules at least $\Omega(\frac{\alpha}{\log m}) \OPT(I_{BSP-symm}(\bar{b}))$; taking expectation over the $\log m$ possibilities of $\bar{b}$, the expected number of jobs scheduled is $\frac{1}{\log m} \sum_{\bar{b} \in \{\tilde{B}, {\tilde{B}}/{2}, {\tilde{B}}/{4}, \ldots, {\tilde{B}}/{2^{\log m - 1}}\}} \Omega(\frac{\alpha}{\log m}) \OPT(I_{BSP-symm}(\bar{b})) \ge \Omega(\frac{\alpha}{\log^2 m}) \OPT(I_{SP})$, where the last inequality follows from \Cref{claim:symm}. 

    Thus, $\bar{x}$ is the desired solution for the original \ref{eq:packGenSched} instance $I_{SP}$, concluding the proof of Item 2 of \Cref{thm:genBudgetNorm}.
    


\subsection{\ref{eq:packGenSched} with $\ell_p$ or Top-$k$ Outer Norm}\label{sec:lp-topk-reduction}

    We now prove Item 2 of \Cref{thm:genBudgetNorm}, namely competitive algorithms for \ref{eq:packGenSched} with $\ell_p$ and Top-$k$ outer norms; since the analyses for both norms share most elements, we do them together. So consider an instance $I^f_{SP}$ of the \ref{eq:packGenSched} problem where the outer function $f$ is either the $\ell_p$ or the Top-$k$ norm, and the inner norms $\|\cdot\|_i$ are $(\alpha, c)$-packable (with job processing times $p$ and load budget $B$); i.e., we want to find online a schedule of jobs to machines $x$ so that $\max\{\sum_{ijk} x_{ijk} : f(\Lambda(x)) \le B, \sum_{i,k} x_{ijk} \le 1 ~~\forall j, x \in \{0,1\}^{mnr}\}$, and $\Lambda_i(x) = \|(x_{ijk} p_{ijk})_{jk}\|_i$.

    We reduce this problem to \ref{eq:budgetGenSched} with weighted $\ell_1$ outer norm in two steps. First, we employ the reduction to \ref{eq:budgetGenSched} used in  \Cref{thm:pboundPackSched}, which we recall here since we need to exploit finer grained information from that construction. This reduction created the following instance $I^f_{BSP}$ of \ref{eq:budgetGenSched}: for each machine $i \in [m]$ in $I_{SP}$, we make $\log m$ copies $(i,1), (i,2), \dots, (i, \log m)$ in $I_{BSP}$ with the same internal norm $\|\cdot\|_{i,\ell} := \|\cdot\|_i$ and job processing times $p_{i\ell j k} := p_{i j k}$, but the $\ell$th copy has load budget    
    $b_{i \ell} = \frac{b_i^{\max}}{2^\ell}$, where
    \[b_i^{\max} = \sup\{b \geq 0 : f(b \cdot e_i) \leq B\};
    \]
    in other words, $b_i^{\max}$ is the largest cost that machine $i$ can incur without exceeding the {outer load budget} on its own. The outer function $f'$ of the new instance $I_{BSG}$ is defined by employing the original outer function $f$ over the sum of the copies of the machines:
    \begin{align*}
    f'((z_{i, \ell})_{i,\ell}) = f\Big(\sum_\ell z_{1,\ell} \,,\, \sum_\ell z_{2,\ell} \,,\, \ldots \,,\, \sum_\ell z_{m,\ell}\Big).
    \end{align*}
    Finally, the new instance has load budget $3 B$. So the offline version of $I_{BSP}$ is given by $\max\{\sum_{i\ell j k} x_{i \ell j k} : f((\sum_\ell y_{i \ell} b_{i \ell})_i) \le 3B, \Lambda'_{i \ell}(x) \le y_{i \ell} b_{i \ell}  ~~\forall i,\ell,~ \sum_{i,\ell, k} x_{i\ell j k} \le 1 ~~\forall j, x \in \{0,1\}^{m (\log m) n r}, y \in \{0,1\}^{m \cdot \log m}.\}$, where the machine load is given by $\Lambda'_{i \ell}(x) = \|(x_{i \ell j k} p_{i j k})_{j,k} \|_i$. 

    The main property is that this instance is a ``relaxation'' of the original instance $I_{SP}$; this is argued in the proof in \Cref{sec:reductionP} (we remark that while the theorem statement assumes $p$-bounded outer functions, this part of the argument only uses their $p$-subadditivity; since our outer function $f$ is a norm, it is 1-subadditive and the guarantees still hold). 
    
    \begin{claim} \label{claim:relSPtoBSP}
        For $f$ being any monotone norm, we have:
        \begin{enumerate}
            \item $\OPT(I^f_{BSP}) \ge \OPT(I_{SP})$.
            \item Any online solution $(x^f, y^f)$ for $I^f_{BSP}$ that violates the outer load budget by a factor of at most $\alpha$ (i.e., $f((\sum_\ell y^f_{i \ell} b_{i \ell})_i) \le \alpha \cdot 3B$) violates each machine's load budget by at most a factor of $c$ (i.e., $\Lambda'_{i \ell}(x^f) \le c y^f_{i \ell} b_{i \ell}$ for all $i,\ell$) can be converted online into a solution $\bar{x}$ for $I_{SP}$ that schedules the same number of jobs (i.e., $\sum_{ijk} \bar{x}_{ijk} = \sum_{i\ell j k} x^f_{i \ell j k}$) and violates the outer load budget by a factor of at most $3 \alpha c$ (i.e., $f(\Lambda(\bar{x})) \le 3 \alpha c B$). 
        \end{enumerate}
    \end{claim}

    Now we reduce the instance $I^f_{BSP}$ to an instance $I^{\ell_1}_{BSP}$ of the same problem but with weighted $\ell_1$; now this will depend whether $f$ is the $\ell_p$ or the Top-$k$ norm. The instance $I^{\ell_1}_{BSP}$ is the same as $I^f_{BSP}$ but: 
    
    \begin{enumerate}
        \item If $f = \ell_p$, then $I^{\ell_1}_{BSP}$ uses the weighted $\ell_1$ norm where the weight of machine $(i,\ell)$ is $w_{i \ell} := b_{i \ell}^{p-1}$, and the outer load budget is $(3 B)^p$. That is, the offline version of the problem is $\max\{\sum_{i\ell j k} x_{i \ell j k} : \sum_{i, \ell} y_{i \ell} b_{i \ell}^p \le (3 B)^p, \Lambda'_{i \ell}(x) \le y_{i \ell} b_{i \ell}  ~~\forall i,\ell,~ \sum_{i,\ell,k} x_{i\ell j k} \le 1 ~~\forall j, x \in \{0,1\}^{m (\log m) n r}, y \in \{0,1\}^{m \cdot \log m}.\}$.

        \item If $f = \ell_p$, then $I^{\ell_1}_{BSP}$ uses the weighted $\ell_1$ norm where the weight of machine $(i,\ell)$ is $w_{i \ell} := 1$ if $b_{i \ell} > \frac{3B}{k}$ and $w_{i \ell} := 0$ otherwise, and the outer load budget is $3 B$. That is, the offline version of the problem is $\max\{\sum_{i\ell j k} x_{i \ell j k} : \sum_{i, \ell : b_{i \ell} > \frac{3B}{k}} y_{i \ell} b_{i \ell} \le 3 B, \Lambda'_{i \ell}(x) \le y_{i \ell} b_{i \ell}  ~~\forall i,\ell,~ \sum_{i,\ell,k} x_{i\ell j k} \le 1 ~~\forall j, x \in \{0,1\}^{m (\log m) n r}, y \in \{0,1\}^{m \cdot \log m}.\}$.
    \end{enumerate}

    The main claim is that in both cases, $I^{\ell_1}_{BSP}$ is a relaxation of the respective instance $I^f_{BSP}$.

    \begin{claim} \label{claim:relBSPl1}
        For $f$ being either $\ell_p$ or Top-$k$, the instance $I^{\ell_1}_{BSP}$ defined above satisfies the following:
        \begin{enumerate}
            \item $\OPT(I^{\ell_1}_{BSP}) \ge \OPT(I^f_{BSP})$
            \item Any online solution $(x^{\ell_1}, y^{\ell_1})$ for $I^{\ell_1}_{BSP}$ that violates each machine's load budget by at most a factor of $c$ (i.e., $\Lambda'_{i \ell}(x^{\ell_1}) \le c y^{\ell_1}_{i \ell} b_{i \ell}$ for all $i,\ell$)  is also a solution for $\OPT(I^f_{BSP})$ that violates the outer load budget by a factor of at most $3$ and violates each machine's load budget by at most a factor of $c$.
        \end{enumerate}
    \end{claim}

    \begin{proof}
    \textbf{Case 1 ($f = \|\cdot\|_p$)}: Observe that for every possible machine activation $y \in \{0,1\}^{m \cdot \log m}$ and every $i$ we have the following connection between the $\ell_p$-norm (to the power of $p$) outer constraint of $I^f_{BSP}$ and the weighted $\ell_1$-norm outer constraint of $I^{\ell_1}_{BSP}$: 
    \begin{align}
        \sum_\ell y_{i \ell} b_{i \ell}^p \le \bigg(\sum_\ell y_{i \ell} b_{i \ell}  \bigg)^p \le 2^p \cdot \sum_\ell y_{i \ell} b_{i \ell}^p; \label{eq:lpTol1}
    \end{align}
    the first inequality is straightforward and just uses the fact $y_{i \ell}^p = y_{i \ell}$, and the second uses that because the $(b_{i,\ell})_\ell$ are defined using powers of 2 we have $\sum_\ell y_{i \ell} b_{i \ell} \le 2 \max_\ell y_{i \ell} b_{i \ell}$ and so $(\sum_\ell y_{i \ell} b_{i \ell})^p \le 2^p \max_\ell (y_{i \ell} b_{i \ell})^p \le 2^p \cdot \sum_\ell y_{i \ell} b_{i \ell}^p$. 

    To prove $\OPT(I^{\ell_1}_{BSP}) \ge \OPT(I^f_{BSP})$, consider an optimal solution $(x^*,y^*)$ for $\OPT(I^f_{BSP})$; it suffices to show that this solution is feasible for $\OPT(I^{\ell_1}_{BSP})$. For that, it suffices to check validity of the outer load budget $\sum_{i, \ell} y^*_{i \ell} b_{i \ell}^p \le (3 B)^p$. But this follows from the first inequality of \eqref{eq:lpTol1} and $\|(\sum_\ell y^*_{i \ell} b_{i \ell})\|_p \le 3B \equiv \sum_i (\sum_\ell y^*_{i \ell} b_{i \ell})^p \le (3B)^p$ (by feasibility of $y^*$ in $I^f_{BSP}$). 

    To prove the second item of the claim, consider any solution $(x^{\ell_1}, y^{\ell_1})$ for $I^{\ell_1}_{BSP}$. Again it suffices to verify that, seeing it as a solution for $I^f_{BSP}$, it violates the outer budget by a factor of at most 3, namely $\|(\sum_\ell y^{\ell_1}_{i \ell} b_{i \ell})\|_p \le 9 B \equiv  \sum_i (\sum_\ell y^{\ell_1}_{i \ell} b_{i \ell})^p \le (6 B)^p$. But this follows from the second inequality of \eqref{eq:lpTol1} and the fact $\sum_{i,\ell} y_{i \ell} b_{i \ell}^p \le (3 B)^p$ (by feasibility of $y^{\ell_1}$ in $I^{\ell_1}_{BSP}$). 

\medskip
    \textbf{Case 2 ($f = \|\cdot\|_{\topk}$)}:
    The analogous of \eqref{eq:lpTol1} relating the $\topk$-norm with the weighted $\ell_1$-norm in his case is:
    \begin{align}
        \min\bigg\{\sum_{i, \ell : b_{i \ell} > \frac{3B}{k}} y_{i \ell} b_{i \ell}\,,\, 3B\bigg\}
 <  \bigg\|\bigg(\sum_\ell y_{i \ell} b_{i \ell}  \bigg)_i \bigg\|_{\topk} \le \sum_{i, \ell : b_{i \ell} > \frac{3B}{k}} y_{i \ell} b_{i \ell} + 6B; \label{eq:topkTol1}
    \end{align}    
    To see that the first inequality holds: assume $\|(\sum_\ell y_{i \ell} b_{i \ell})_i\|_{\topk} \le 3B$, else the inequality is trivial. Notice that in this case, every coordinate of the vector $(\sum_\ell y_{i \ell} b_{i \ell})_i$ that is $\ge \frac{3B}{k}$ is among one od its top $k$ coordinates (else its top $k$ coordinates would all be strictly bigger than $\frac{3B}{k}$ and we could not have $\|(\sum_\ell y_{i \ell} b_{i \ell})_i\|_{\topk} \le 3B$). In particular, for any coordinate $i$, its contribution to the top-$k$ norm is at least $\sum_{\ell : b_{i \ell} > \frac{3B}{k}} y_{i \ell} b_{i \ell}$ (i.e., if the RHS is greater than 0, then it is greater than $\frac{3B}{k}$ and thus these terms (plus others) are picked up by the top-$k$ norm). Considering all coordinates, this proves $\|(\sum_\ell y_{i \ell} b_{i \ell} )_i\|_{\topk} \ge \sum_i \sum_{\ell : b_{i \ell} > \frac{3B}{k}} y_{i \ell} b_{i \ell}$, giving the first inequality in \eqref{eq:topkTol1}. To prove the second inequality in \eqref{eq:topkTol1}, we use $u = (\sum_{\ell : b_{i \ell} > \frac{3B}{k}} y_{i \ell} b_{i \ell})$ and $v = (\sum_{\ell : b_{i \ell} \le \frac{3B}{k}} y_{i \ell} b_{i \ell})$ and the triangle inequality $\|u+v\|_{\topk} \le \|u\|_{\topk} + \|v\|_{\topk}$ to obtain
    \begin{align*}
    \bigg\|\bigg( \sum_\ell y_{i \ell} b_{i \ell}   \bigg)_i \bigg\|_{\topk} &\le \bigg\|\bigg( \sum_{\ell : b_{i \ell} > \frac{3B}{k}} y_{i \ell} b_{i \ell}   \bigg)_i \bigg\|_{\topk} +     \bigg\|\bigg( \sum_{\ell : b_{i \ell} \le \frac{3B}{k}} y_{i \ell} b_{i \ell}   \bigg)_i \bigg\|_{\topk} \\
    &\le \sum_{i,\ell : b_{i \ell} > \frac{3B}{k}} y_{i \ell} b_{i \ell} \,+\, k \cdot \max_i  \sum_{\ell : b_{i \ell} \le \frac{3B}{k}} y_{i \ell} b_{i \ell},
    \end{align*}
    where for the last inequality we upper bound the top-$k$ coordinates in the $\topk$-norm by all coordinates and by $k$ times the top coordinate, respectively. Finally, since the $b_{i\ell}$'s grow as powers of 2, we can upper bound the sum in the last term as $\sum_{\ell : b_{i \ell} \le \frac{3B}{k}} y_{i \ell} b_{i \ell} \le \sum_{\ell : b_{i \ell} \le \frac{3B}{k}} b_{i \ell} \le 2 \cdot \frac{3B}{k}$. We thus conclude that $\|( \sum_\ell y_{i \ell} b_{i \ell}   )_i \|_{\topk} \le  \sum_{i,\ell : b_{i \ell} > \frac{3B}{k}} y_{i \ell} b_{i \ell} + 6 B$, concluding the proof of \eqref{eq:topkTol1}.

    As before, \eqref{eq:topkTol1} immediately proves \Cref{claim:relBSPl1} when $f = \topk$: To prove $\OPT(I^{\ell_1}_{BSP}) \ge \OPT(I^f_{BSP})$, consider an optimal solution $(x^*,y^*)$ for $\OPT(I^f_{BSP})$; it suffices to show that this solution is feasible for $\OPT(I^{\ell_1}_{BSP})$, i.e., that it satisfies the outer load budget $\sum_{i, \ell : b_{i \ell} > \frac{3B}{k}} y^*_{i \ell} b_{i \ell} \le  3B$. But this follows from the first inequality of \eqref{eq:topkTol1} and $\|(\sum_\ell y^*_{i \ell} b_{i \ell})\|_{\topk} \le 3B$ (by feasibility of $y^*$ in $I^f_{BSP}$). 

    To prove the second item of the claim, consider any solution $(x^{\ell_1}, y^{\ell_1})$ for $I^{\ell_1}_{BSP}$. Again it suffices to verify that, seeing it as a solution for $I^f_{BSP}$, it violates the outer budget by a factor of at most 3, namely $\|(\sum_\ell y^{\ell_1}_{i \ell} b_{i \ell})\|_{\topk} \le 9 B$. But this follows from the second inequality of \eqref{eq:topkTol1} and the fact $\sum_{i,\ell : b_{i \ell} > \frac{3B}{k}} y_{i \ell} b_{i \ell} \le  3B$ (by feasibility of $y^{\ell_1}$ in $I^{\ell_1}_{BSP}$). This concludes the proof of \Cref{claim:relBSPl1}.
    \end{proof}

    Now we can finally approximate our original instance $I^f_{SP}$ of \ref{eq:packGenSched} using these reductions and the algorithm for the weighted $\ell_1$-norm case from \Cref{sec:BSPl1} (\Cref{thm:l1-budget}): We run this algorithm on the instance $I^{\ell_1}_{BSP}$ to obtain a solution $(\tilde{x}, \tilde{y})$ that violates each machine's load budget by at most a factor of $c$ (i.e., $\Lambda'_{i \ell}(x^{\ell_1}) \le c y^{\ell_1}_{i \ell} b_{i \ell}$ for all $i,\ell$), use \Cref{claim:relBSPl1} to convert this to a solution $(\check{x}, \check{y})$ for $I^f_{BSP}$, then use \Cref{claim:relSPtoBSP} to convert it to a solution $\bar{x}$ for $I^f_{SP}$. By the guarantees of these theorems/claims, we have that $\bar{x}$ violates the original outer load budget by a factor of at most $9 c$ and schedules at least $\Omega(\frac{\alpha}{\log m}) \OPT(I^{\ell_1}_{BSP})$ jobs (\Cref{thm:l1-budget}; recall that $\alpha$ is the assumed solvability guarantee of \ref{eq:normPack}$_{\|\cdot\|_i}$ for the inner norms $\|\cdot\|_i$ of the instance $I_{SP}$). This is at least $\Omega(\frac{\alpha}{\log m}) \OPT(I^f_{SP})$ by the claims, which then proves Items 1 of \Cref{thm:genBudgetNorm}.



\subsection{\ref{eq:packGenSched} with Composite Outer Norms}\label{sec:norm-compose}

    Finally, we now prove \Cref{thm:norm-compose-sched-pack}. Namely we design a competitive algorithm for \ref{eq:packGenSched}$_{\|\cdot\|,\{\|\cdot\|_i\}_i}$ when the outer norm is a disjoint composition of norms. 
\begin{proof}[Proof of \Cref{thm:norm-compose-sched-pack}]
    Let $S_1, \dots, S_L \subseteq [m]$ be disjoint, and let $\|\cdot\|'_1, \dots, \|\cdot\|'_L$ be norms such that $\|\cdot\|'_\ell : \R^{|S_\ell|} \to \R$, and let $\|\cdot\|' : \R^L \to \Rp$ be a norm. Additionally, we assume that $\|\cdot\|'$ is a $(\beta_1, \gamma_1)$-good outer norm, and each $\|\cdot\|'_\ell$ is a $(\beta_2, \gamma_2)$-good outer norm, as defined by \Cref{def:good-norm}.
    Now, let $\|\cdot\| : \R^m \to \R$ be the composite norm given by
\[ 
\|y\| := \Big\|\big(\|y_{S_1}\|_1', \dots, \|y_{S_L}\|_L' \big)\Big\|'.
\]
We seek to show that $\|\cdot\|$ is a $(\beta_1\beta_2,~\gamma_1 \gamma_2)$-good outer norm. Thus, we consider the problem \ref{eq:packGenSched}$_{\|\cdot\|, \{\|\cdot\|_i\}_i}$ with outer norm $\|\cdot\|$ and inner norms $\|\cdot\|_i : \R^{nr} \to \R$ such that for $i \in [m]$ \ref{eq:normPack}$\|\cdot\|_i$ is $(\alpha, c)$-solvable. That is, the offline version of this problem is given by
\begin{align*}
   \max~ & \sum_{i,j,k} x_{ijk} \\
    \textrm{s.t.}~ & \Big\|\big(\|\Lambda_{S_1}\|_1', \dots, \|\Lambda_{S_L}\|_L' \big)\Big\|' \le B \\
    & \sum_{i,k} x_{ijk} \leq 1, \qquad \forall j \in [n] \notag\\
    & \Lambda_i(x) = \|(x_{ijk} p_{ijk})_{j,k}\|_i, \qquad \forall i\in [m]\\
    & x \in \{0,1\}^{m n r} .
\end{align*}
Here, we use $\Lambda_{S_\ell}$ to denote the vector $(\Lambda_i)_{i \in S_{\ell}}$.
Now consider the related instance of \ref{eq:packGenSched}$_{\|\cdot\|', \{\|\cdot\|''_\ell\}_\ell}$ where: 1) we have $L$ machines (one corresponding to each original machine group $S_\ell$), and machine $\ell \in [L]$ has inner norm  
$\|\cdot\|''_\ell : \R^{|S_\ell|nr} \to \R$ given by
\[
\|(z_{i j k})_{i \in S_\ell, j,k}\|''_\ell := \Big\|\big(\|(z_{i j k})_{j,k}\|_i\big)_{i \in S_\ell}\Big\|';
\]
2) the outer norm is just $\ell_1$, and the outer load budget is still $B$; 3) There are $|S_\ell| \cdot r$ ways\footnote{This differs slightly from our definition of \ref{eq:packGenSched}, where the number of ways of assignment is the same for each machine. However, we can easily get around this by adding dummy scheduling ways which incur infinite load, and hence are not usable.} of assigning job $j$ to machine $\ell$, and the way indexed by the pair $(i,k)$ (with $i \in S_\ell$) incurs load $p_{\ell,j,(i,k)} := p_{i j k}$ (and the corresponding variable of the assignment is $x_{\ell,j,(i,k)}$. That is, the offline version of this new problem is given by
\begin{align*}
   \max~ & \sum_{\ell,j, i \in S_\ell, k} x_{\ell,j,(i,k)} \\
    \textrm{s.t.}~ & \|(\Lambda_1, \dots, \Lambda_L)\|' \le B \\
    & \sum_{\ell, i \in S_{\ell},k } x_{\ell,j, (i,k)} \leq 1, \qquad \forall j \in [n] \notag\\
    & \Lambda_\ell = \big\|(x_{\ell, j, (i,k)} p_{\ell,j,(i,k)})_{j, i\in S_\ell, k}\big\|''_\ell, \qquad \forall \ell \in [L]\\
    & x_{\ell,j,(i,k)} \in \{0,1\}, \qquad \forall \ell \in [L], j \in [n], i \in S_{\ell}, k \in [r].
\end{align*}

By identifying the variable $x_{\ell,k,(i,k)}$ ($i \in S_{\ell}$) in the new instance with the variable $x_{i,j,k}$ of the old instance, we see that they are actually equivalent. Thus, the new instance expresses the original problem of \ref{eq:packGenSched}$_{\|\cdot\|,\{\|\cdot\|_i\}_i}$ as the problem of \ref{eq:packGenSched}$_{\|\cdot\|',\{\|\cdot\|''_\ell\}_\ell}$. Using this mapping, we see that the two problems are equivalent.

Moreover, by unraveling the new inner-norm problem \ref{eq:normPack}$_{\|\cdot\|''_\ell}$, we see that it can be represented as
\begin{align*}
   \max~ & \sum_{j, i \in S_\ell, k} x_{j,(i,k)} \\
    \textrm{s.t.}~ & \|(\Lambda_1, \dots, \Lambda_{|S_\ell|})\|_\ell' \le B \\
    & \sum_{i \in S_{\ell},k } x_{j, (i,k)} \leq 1, \qquad \forall j \in [n] \notag\\
    & \Lambda_i = \|(x_{j(i,k)} p_{\ell, j, (i,k)})_{j, i\in S_\ell, k}\big\|'_i, \qquad \forall i \in S_\ell\\
    & x_{j,(i,k)} \in \{0,1\}, \qquad \forall \ell \in [L], j \in [n], i \in S_{\ell}, k \in [r].
\end{align*}
Identifying $x_{j, (i, k)}$ with $x_{ijk}$, we also see that this \ref{eq:normPack}$_{\|\cdot\|''_\ell}$ is exactly the problem \ref{eq:packGenSched}$_{\|\cdot\|_\ell', \{\|\cdot\|_i\}_{i \in S_\ell}}$. 

Given these equivalences, we now apply our bicriteria guarantees. First, for each $\ell \in [L]$ and $i \in S_\ell$, since we assume that each $\|\cdot\|_i$ is $(\alpha, c)$-solvable, and $\|\cdot\|'_\ell$ is $(\beta_2, \gamma_2)$-good, we have that \ref{eq:packGenSched}$_{\|\cdot\|_\ell', \{\|\cdot\|_i\}_{i \in S_\ell}}$ is $(\beta_2 \alpha,~\gamma_2 c)$-solvable. Hence, we also have that \ref{eq:normPack}$_{\|\cdot\|''_\ell}$ is $(\beta_2 \alpha, \gamma_2 c)$-solvable, as the two problems are equivalent.

Now, since \ref{eq:normPack}$_{\|\cdot\|''_\ell}$ is $(\beta_2 \alpha, \gamma_2 c)$-solvable and $\|\cdot\|'$ is $(\beta_1, \gamma_1)$-good, we obtain that \ref{eq:packGenSched}$_{\|\cdot\|',\{\|\cdot\|''_\ell\}_\ell}$ is $(\beta_1\beta_2\alpha,~\gamma_1\gamma_2 c)$-solvable. By the first equivalence we showed, this means that \ref{eq:packGenSched}$_{\|\cdot\|, \{\|\cdot\|_i\}_i}$ is $(\beta_1\beta_2\alpha,~\gamma_1\gamma_2 c)$-solvable.

To recap, we have now shown that for any norms $\|\cdot\|_i : \R^{nr} \to \Rp$ which are $(\alpha, c)$-solvable for all $i \in [n]$, we have that  \ref{eq:packGenSched}$_{\|\cdot\|, \{\|\cdot\|_i\}_i}$ is $(\beta_1\beta_2\alpha,~\gamma_1\gamma_2 c)$-solvable. Therefore, we conclude that $\|\cdot\|$ is a $(\beta_1\beta_2,~\gamma_1\gamma_2)$-good outer norm.
\end{proof}

%% file: appendix.tex
\section{Missing Proofs}

\subsection{Proof of the Theorems in the Introduction} \label{app:proofsIntro}

\begin{proof}[Proof of \Cref{thm:OGSpboundedSym}]
    This theorem follows directly from \Cref{cor:genSched} and the $(\frac{1}{3},1)$-solvability of \ref{eq:normPack} with a monotone symmetric norm of \Cref{lem:single-machine}.
\end{proof}

\begin{proof}[Proof of \Cref{thm:setCoverpBoundedNew}]
    First we formally define Online Set Cover with convex cost (OSC-cxv): $n$ elements come one-by-one, and once an element comes it reveals which of the $m$ available sets $\{S_1,\ldots,S_m\}$ contain it. At all times, the algorithm has to maintain a cover, namely a collection of sets that cover all elements seen thus far. Sets can only be added to this collection, but not deleted. Each set $S_i$ has an associated cost $c_i$, and we have a convex cost function $f : \R^m \rightarrow \R$ that aggregated the costs of the selected sets. The goal is then to minimize the cost of the final collection of sets after all elements arrive, namely if $x_i \in \{0,1\}$ indicates whether the set $S_i$ is in our final collection or not, we want to minimize $f(c_1x_1, c_2x_2, \ldots, c_m x_m)$.

    As mentioned in \Cref{sec:warmup}, Online Set Cover with convex cost $f$ (OSC-cxv) is a special case of \ref{eq:covGenSched} (where $r=1$, so there is only 1 option of scheduling a job on machine; thus we omit the subscript $k$ from items like $p_{ijk}$ etc): The $j$th element to be covered corresponds to job $j$, and each set $S_i$ correspond to a machine $i$, which incurs a load $c_i$ if any element in $S_i$ is assigned to it, which corresponds to buying the set $i$; the load $p_{ij}$ of assigning job $j$ to machine $i$ is then given by:
\[
p_{ij} =\begin{cases}
        c_i & j \in S_i\\
        +\infty & j \not \in S_i
    \end{cases}, \qquad \forall i \in [m],~ \forall j \in [n] \enspace .
  \]
  Each machine is endowed with the $\ell_\infty$ as its inner norm, and the outer norm/aggregation function is $f$. That is, OCS-cvx corresponds to \ref{eq:covGenSched} with the following parameters:
\begin{align*}
    \min f(\Lambda(x)), \quad \text{ where }
    \Lambda_i(x) := \|(x_{ij} p_{ij})_j\|_\infty \quad \forall i \in [m].
\end{align*}
    Notice that any solution (of cost $< \infty$) for this \ref{eq:covGenSched} instance can be converted online to a solution of the original OCS-cvx problem with the same cost (if at any point $x_{ij} = 1$ for any element $j$, pick the set $S_i$ in OCS-cvx), and vice versa (assign job $j$ to a machine corresponding to any set that covers it in the OCS-cvx solution).  

    Then, if the cost function $f$ is $p$-bounded, we can apply the result from \Cref{cor:genSched} on the \ref{eq:covGenSched} instance to obtain a $O(p^2 \log n \cdot \log\!\log n \cdot \log^2 m)^p$-competitive solution for OSC-Cvx (since by \Cref{lem:single-machine} \ref{eq:normPack} is $(1,1)$-solvable for the $\ell_{\infty}$ norm). 
    
    However, to save a factor $\log m$, we note that for the problem \ref{eq:covGenSched}$_{f, \{\|\cdot\|_i\}_i}$ in the special case that norms when all norms $\|\cdot\|_i$ are $\ell_{\infty}$, then we get the following improvement to \Cref{thm:pboundPackSched}:
    \begin{claim}\label{claim:pboundPackSched-setCover}
        The problem \ref{eq:packGenSched}$_{f, \{\|\cdot\|_i\}_i}$, where $f$ is $p$-bounded and $\|\cdot\|_i = \|\cdot\|_\infty$ for all $i$, is $(\Omega(\frac{1}{\log m}), (3p)^p)$-solvable.
    \end{claim}
    \begin{proof}
    To see this, we note that the online algorithm for \ref{eq:normPack}$_{\|\cdot\|_\infty}$ given in \Cref{app:single-machine}, which solves the problem exactly, does not need a guess on $\Opt$. Hence, in proof of \Cref{lem:alg_i-bound1} and \Cref{lem:alg_i-bound2}, we can actually save a $\log m$ factor in the denominator. Note that for general norms $\|\cdot\|_i$, we were losing this factor because we needed to condition on the event that the guess $M_i$ was within a factor $2$ of the optimum on the machine, but for $\|\cdot\|_i = \|\cdot\|_\infty$ we need no such guarantee.

    With this $\log m$ improvement to \Cref{lem:alg_i-bound1} and \Cref{lem:alg_i-bound2} in this special case, we also save a $\log m$ in the denominator from \Cref{lem:unifiedLemma}, \Cref{thm:pboundBudgetSched}, and ultimately \Cref{thm:pboundPackSched} by retracing the steps of each proof. This gives our claim.
    \end{proof}

    Hence, combining \Cref{claim:pboundPackSched-setCover} with \Cref{thm:packCovReduction} gives the first part of \Cref{thm:setCoverpBoundedNew}. For the second part, notice that if $f$ is $p$-bounded and $\|\cdot\|_i = \|\cdot\|_\infty$ for all $i$, then \Cref{claim:pboundPackSched-setCover} gives us that \ref{eq:packGenSched}$_{f^{1/p},\{\|\cdot\|_i\}_i}$ is $(\Omega(\frac{1}{\log m}, 3p))$-solvable. Moreover, since $f^{1/p}$ is subadditive by \Cref{lem:p-subadditive}, we can apply \Cref{thm:packCovReduction} to get that \ref{eq:covGenSched}$_{f^{1/p},\{\|\cdot\|_i\}_i}$ has a $O(p \log m \cdot \log n)$-competitive algorithm, which gives our desired result.
\end{proof}

\begin{proof}[Proof of \Cref{thm:nonmetricpBounded}]
    This theorem follows directly from \Cref{cor:genSched} and the $(\frac{1}{3},2)$-solvability of \ref{eq:normPack} with a norm given by Item 3 of \Cref{lem:single-machine}. Crucially, we use the fact that this \ref{eq:normPack} problem has a bi-criteria guarantee even when the weights $w_{jk}$ are revealed online, as is the case in Online Non-metric Facility Location.
\end{proof}

\begin{proof}[Proof of \Cref{thm:schedNorm}]    
    This theorem follows directly from \Cref{thm:packCovReduction} (with $p=1$), \Cref{thm:genBudgetNorm} and the $(1,\frac{1}{3})$-solvability of \ref{eq:normPack} with a monotone symmetric norm of \Cref{lem:single-machine}.    
\end{proof}

\begin{proof}[Proof of \Cref{thm:genSchedComposition}]
    Recall the definition of goodness of a norm (\Cref{def:good-norm}): a norm $\|\cdot\| : \R^m \to \Rp$ is a \emph{$(\beta, \gamma)$-good outer norm} if \ref{eq:packGenSched}$_{\|\cdot\|, \{\|\cdot\|_i\}_i}$ is $(\beta\alpha, \gamma c)$-solvable whenever the inner-norm problems \ref{eq:normPack}$_{\|\cdot\|_i}$ are all $(\alpha, c)$-solvable.

    By \Cref{thm:genBudgetNorm}, we know that symmetric norms (over $k$ coordinates) are $\big(\Omega({1}/{\log^2 k}),~O(1)\big)$-good outer norms, and the $\ell_1$ norm is a $\big(\Omega({1}/{\log k}),~O(1)\big)$-good outer norm. \Cref{thm:genSchedComposition} then follows from the composition result of \Cref{thm:norm-compose-sched-pack}.
\end{proof}

\begin{proof}[Proof of \Cref{cor:setCoverComp}]
    Recall that here we consider Online Set Cover with cost function given by the nested norm $\|x\| = \|(\|x_{S_1}\|'_{1},\ldots,\|x_{S_L}\|'_{L})\|'$, where $\|\cdot\|'$ is the $\ell_1$-norm and where $S_1,\ldots,S_L$ are, not necessarily disjoint, subsets of the $m$ machines.

    The key point is that, for any such nested norm, we can assume without loss of generality that the sets $S_1,\ldots,S_L$ are actually disjoint, because we can perform an online reduction to an equivalent instance that has this property; this is done in Appendix C.3 of \cite{KMS24arxiv} (which is the same reduction introduced in  \cite{NS-ICALP17}).

    Once this is done, we can apply the composition result of \Cref{thm:norm-compose-sched-pack} (since Online Set Cover is a special case of Generalized Online Scheduling): 
    
    \begin{enumerate}
        \item When each of the norms $\|\cdot\|'$ is an $\ell_p$-norm, the composed norm is $\|\cdot\|$ is a $\big(\Omega({1}/{\log m \cdot \log L}),~O(1)\big)$-good outer norm (since  by \Cref{thm:genBudgetNorm} we know that $\ell_p$ norms over $k$ coordinates are $\big(\Omega({1}/{\log k}),~O(1)\big)$-good outer norms).

        \item When each of the norms $\|\cdot\|'$ is a symmetric norm, the composed norm $\|\cdot\|$ is a $\big(\Omega({1}/{\log^2 m \cdot \log L}),~O(1)\big)$-good outer norm (since  by \Cref{thm:genBudgetNorm} symmetric norms over $k$ coordinates are $\big(\Omega({1}/{\log^2 k}),~O(1)\big)$-good outer norms).
    \end{enumerate}

    Combining with the \ref{eq:covGenSched}-to-\ref{eq:packGenSched} reduction from \Cref{thm:packCovReduction}, we obtain competitive ratios $O(\log m \cdot \log L \cdot \log n)$ and $O(\log^2 m \cdot \log L \cdot \log n)$ for the two respective options of composed norm. This proves \Cref{cor:setCoverComp}.
\end{proof}


\subsection{Proof of \Cref{lem:single-machine}} \label{app:single-machine}
We prove the lemma for each norm.
\begin{enumerate}
    \item When $\|\cdot\| = \|\cdot\|_\infty$, we can solve \ref{eq:normPack} online exactly by scheduling every job $j$ in any way $k$ for which $p_{jk} \leq B$, skipping $j$ if no such $k$ exists. Thus, the problem is $(1, 1)$-solvable.

    \item When $\|\cdot\|$ is symmetric, we can perform the following algorithm. First, if our lower bound $M$ on $\Opt$ has $M \leq 2$, we can schedule any job to get the required bound, so suppose $M > 2$. We may also assume $M$ is an integer, otherwise take $\ceil{M}$.
    
    Let $e_{\leq M/2} := (1, \dots, 1, 0, \dots, 0)$ denote the vector of $\floor{M/2}$ ones followed by $nr - \floor{M/2}$ zeros, and define $p^* = \frac{B}{\|e_{\leq M/2}\|}$. This ensures that $\|p^* \cdot e_{\leq M/2}\| = B$.

    Then we greedily schedule any job $j$ in a way $k$ which does not violate the budget and has $p_{jk} \leq p^*$, skipping $j$ if no such $k$ exists.

    Notice that since $\Opt \geq M$, there must be at least $\frac{M}{2}$ jobs $j$ that have $p_{jk} \leq p^*$ for some $k$. Otherwise, the optimal assignment would schedule at least $\Opt - \frac{M}{2} \geq \frac{M}{2}$ jobs in a way such that $p_{jk} > p^*$, which incurs load larger than $\|p^* \cdot e_{\leq M/2}\| = B$ by symmetry of $\|\cdot\|$. Thus, our algorithm schedules at least $\floor{\frac{M}{2}} \geq \frac{M}{3}$ jobs, as the total load of these scheduled jobs does not exceed $\|p^* \cdot e_{\leq M/2}\| = B$, again by symmetry of $\|\cdot\|$.

    \item When $\|x\| = c\|x\|_\infty + \sum_{j,k} w_{jk} |x_{jk}|$, we simply can perform a combination of the above algorithms. We schedule any job $j$ in a way $k$ such that $cp_{jk} \leq B$ and $w_{jk} p_{jk} \leq \frac{2B}{M}$, ensuring that such an assignment respects $\sum_{j, k} x_{jk} w_{jk} p_{jk} \leq B$. As before, we can assume $M > 2$, and we note that there must be at least $\frac{M}{2}$ elements in the optimal assignment that satisfy this constraint. Thus, we schedule at least $\floor{\frac{M}{2}} \geq \frac{M}{3}$ jobs an incur total load at most $2B$.
\end{enumerate}



\subsection{Proof of \Cref{lemma:packCovReduction1}}\label{app:cov-to-pack-proof}

To make this formal, let $R_k$ be the (random) number of elements received by the first agent of group $k$; hence, the number of elements assigned by group $k$ is precisely $R_k - R_{k+1}$. The follow is the main claim which gives the desired halving of the unscheduled items.

\begin{claim} \label{claim:groupAss}
    For every $k$, $$\Pr\bigg(R_{k+1} > \frac{n}{2^k}  ~\bigg|~ R_k \le \frac{n}{2^{k-1}} \bigg) \,\le\, \frac{1}{2 (\log n + 1)}.$$ 
\end{claim}

\begin{proof}
    Let $\cG_{k-1}$ be the $\sigma$-algebra generated by all agents in groups $1,2,\ldots,k-1$. Let $\omega \in \cG_{k-1}$ be a scenario for these agents such that $R_k \le \frac{n}{2^{k-1}}$ (notice that $R_k$ is fully determined by $\cG_{k-1}$); further assume that in $\omega$ we have $R_k \ge \frac{n}{2^k}$, else we automatically have $R_{k+1} \le \frac{n}{2^k}$. 

    (From now on, by ``agent or $i$th agent'' we mean the agent or $i$th agent of group $k$.) For the $i$th agent, let $Y_i$ denote the fraction of jobs that this agent schedules out of the $R_k$ jobs received by the group. Let $S := \frac{R_k - \frac{n}{2^k}}{R_k}$ be the fraction of jobs that this group needs to schedule in order to leave $R_{k+1} \le \frac{n}{2^k}$, namely the claim is equivalent to showing that $Y_1 + \ldots + Y_N < S$ happens with probability at most $\frac{1}{2 (\log n + 1)}$, conditioned on $\omega$.   
    
    To prove the latter, let $\tau$ be the first agent where $Y_1 + \ldots + Y_\tau > S$, or $\tau = \infty$ if no such agent exists. We claim that for an agent $i \le \tau$, $M_k$ is a lower bound on the optimum of its packing instance, i.e., it is possible to schedule at at least $M_k$ of the jobs in its instance even within a {load budget of} $\widehat{\OPT}_{OGS}$. To see this, since $\widehat{\OPT}_{OGS} \ge \OPT_{OGS}$, all items of $I_{OGS}$ (and in particular all items in $I'$, and in particular all items in the instance of an agent) fit within a {load budget} of $\widehat{\OPT}_{OGS}$; moreover, for any $i \le \tau$, the $i$th agent receives at least $M_k$ jobs in his instance, namely $R_k \cdot (1 - (Y_1 + \ldots + Y_{i-1})) \ge R_k \cdot (1-S) = \frac{n}{2^k} = M_k$, thus giving claim. 
    
    Thus, before the stopping time $\tau$, we can apply the guarantee of the algorithm $\cA_{SP}$ used by the agents, i.e., that it assigns at least $\alpha M_k$ requests in expectation (while blowing up the {load budget} $\widehat{\OPT}_{OGS}$ by at most a factor of $c$). More specifically, letting $\cF_{i-1}$ denote the $\sigma$-algebra generated by the agents $1,\ldots,i-1$, the fraction $Y_i$ of the group's items that is scheduled by agent $i$ has conditional expected value lower bounded 
    \begin{align}
    \E[Y_i \mid \cF_{i-1}, \omega] \cdot \ones(i \le \tau) \,\ge\, \frac{\alpha M_k}{R_k} \cdot \ones(i \le \tau) \,\ge\, \frac{\alpha}{2} \cdot \ones(i \le \tau), \label{eq:lbExp}
    \end{align}
    where the last inequality uses that under $\omega$ we have $R_k \le \frac{n}{2^{k-1}}$. Thus, $\sum_{t \le \tau} \E[Y_i \mid \cF_{i-1}, \omega] \ge \frac{\alpha \tau}{2}$.

    We now employ the martingale concentration inequality from \Cref{thm:mart} in \Cref{app:conc}, which gives an exponential tail bound based on the gap of deterministic thresholds set for the sum of the random variables and the sum of their conditional expectations:
    we have (setting $\lambda := \frac{\alpha N}{2} - S$)
    \begin{align*}
        \Pr\bigg(\sum_{i \le \tau} Y_i < S \textrm{ and } \sum_{i \le \tau} \E[Y_i \mid \cF_{i-1}, \omega] \ge \frac{\alpha N}{2} ~\bigg|~ \omega \bigg) \le e^{-\frac{3}{14} \lambda} \le \frac{1}{2\log n},
    \end{align*}
    the last inequality following because $S = \frac{R_k - \frac{n}{2^k}}{R_k} \le 1$ and $N$ was set large enough, so $\lambda \ge \frac{14}{3} \ln (2 \log n)$.

    To conclude the proof, using \eqref{eq:lbExp} we claim that the original event $Y_1 + \ldots + Y_N < S$ we want to bound is contained in the event ``$\sum_{i \le \tau} Y_i < S \textrm{ and } \sum_{i \le \tau} \E[Y_i \mid \cF_{i-1}, \omega] \ge \frac{\alpha N}{2}$'' from the previous inequality: when the former holds, we are in a scenario where $\tau = \infty$ and so the latter becomes $[Y_1 + \ldots + Y_N < S \textrm{ and } \sum_{i \le N} \E[Y_i \mid \cF_{i-1}, \omega] \ge \frac{\alpha N}{2}] \equiv [Y_1 + \ldots + Y_N < S]$, the last equivalence because of \eqref{eq:lbExp}. Therefore, we conclude that $\Pr(Y_1 + \ldots + Y_N < S) \le \frac{1}{2\log n}$, finishing the proof of the claim.
\end{proof}

    Back to the proof of \Cref{lemma:packCovReduction1}. By taking a union bound of the previous claim over the agent groups $k = 1,2,\ldots, \log n + 1$, we obtain that $\Pr(R_{\log n + 1} < \frac{1}{2}) \ge \frac{1}{2}$, that is, with probability at least $\frac{1}{2}$ no jobs are left unscheduled after the last group; this proves Item 2 of \Cref{lemma:packCovReduction1}.

    For Item 1, namely that the cost of the assignment $X$ produced is bounded as $f(\Lambda(X)) \le  O(\frac{\log n \log\!\log n}{\alpha})^p \cdot c^p\,\Opt_{OGS}$: Let $X^i$ be the (partial) assignment of jobs done by the $i$th agent (from all groups), so $X = \sum_i X^i$. By the guarantee of the algorithm $\cA_{SP}$, load of the assignment of each agent is at most $c^p\, \widehat{\OPT}_{OGS}$, namely $f(\Lambda(x^i)) \le c^p\, \widehat{\OPT}_{OGS}$. Then, since the outer function $f$ is $p$-subadditive (\Cref{lem:p-subadditive}) and the function $\Lambda$ is subadditive (by subadditivity of norms), $f(\Lambda(X))$ can be upper bounded as
    \begin{align*}
      f\big(\Lambda\big(x^1 + \ldots + x^{N (\log n +1)}\big)\big) \le \Big( f(\Lambda(x^1))^{1/p} + \ldots + f(\Lambda(x^{N (\log n +1)}))^{1/p} \Big)^p \le (N (\log n +1))^p \cdot c^p\, \widehat{\OPT}_{OGS}.
    \end{align*}
    Since $N = O(\frac{\log\!\log n}{\alpha})$ and $\widehat{\OPT}_{OGS} \le 2^p \cdot \OPT_{OGS}$, we obtain $f(\Lambda(X)) \le  O(\frac{c \log n \log\!\log n}{\alpha})^p \cdot \Opt_{OGS}$ as desired. This concludes the proof of the lemma.

\subsubsection{Concentration Inequalities} \label{app:conc}

	The following concentration inequality can be considered as an extension of Bernstein's inequality to the martingale case that is suitable for our use. This is a special case of Lemma 10 of \cite{cicalese20a}, which is stated only for binary random variables; the proof for variables in $[0,1]$ is identical and we reproduce here for completeness. 
    
	\begin{theorem} \label{thm:mart}
		Consider a sequence of, possibly correlated, random variables $X_1,X_2,\ldots,X_T \in [0,1]$. Let $\cF_t$ denote the $\sigma$-algebra generated by the variables $X_1,\ldots,X_t$. Then for any stopping time $\tau$ adapted to $\{\cF_t\}_t$,  $\beta \ge 0$, and $\lambda \ge \beta + 1$
		\begin{align*}
			\Pr\bigg(\sum_{t \le \tau} X_t \le \beta \textrm{ and } \sum_{t \le \tau} \E[X_t \mid \cF_{t-1}] \ge \beta + \lambda \bigg) \le e^{-\frac{3}{14}\lambda}.
		\end{align*}
        In particular, if $\hat{\mu}$ is a deterministic quantity such that $\sum_{t \le \tau} \E[X_t \mid \cF_{t-1}] \ge \hat{\mu}$ with probability 1, then for every $\lambda \ge \frac{\hat{\mu} + 1}{2}$ 
		\begin{align*}
			\Pr\bigg(\sum_{t \le \tau} X_t \le \hat{\mu} - \lambda \bigg) \,\le\, e^{-\frac{3}{14}\lambda}.
		\end{align*}
	\end{theorem}

    To prove this lemma, we will need Freedman's Inequality.

	\begin{lemma}[Freedman's Inequality~\cite{Freedman}] \label{lemma:Freedman}
		Consider a filtration $\cF_1 \subseteq \cF_2 \subseteq \ldots$ and an adapted martingale difference sequence $M_1,\ldots,M_T$ taking values in $[-1,1]$. Let $V = \sum_t \E[M_t^2 \mid \cF_{t-1}]$ denote the predictable variance. Then for all $\lambda,v > 0$
		\begin{align*}
			\Pr\bigg(\sum_t M_t \le - \lambda \textrm{ and } V \le v\bigg) \le \exp\bigg(- \frac{\lambda^2}{2(v + \lambda/3)}\bigg).
		\end{align*}
	\end{lemma}	

	\begin{proof}[Proof of Theorem \ref{thm:mart}]
	To simplify the notation let $\E_{t-1}[\cdot]$ denote the conditional expectation $\E[\cdot \mid \cF_{t-1}]$ ($\E_0[\cdot] := \E[\cdot]$). The idea is to apply Freedman's Inequality to the sequence $Y_t - \E_{t-1} Y_t$, but to control the predictable variation we need to further stop the process.

	Let $\tau'$ be the smallest $t$ such that $\E_0 Y_1 + \ldots + \E_{t-1} Y_t \ge \beta + \lambda$ (and $\tau' = \infty$ if there is no such $t$). Define the truncated process $\overbar{Y}_t := Y_t \cdot \,\ones(t \le \tau) \cdot \,\ones(t \le \tau') $. We claim that this truncated process upper bounds the original one, namely
	\begin{align}
		\Pr\bigg( \sum_{t \le \tau} Y_t \le \beta \textrm{ and } \sum_{t \le \tau} \E_{t-1} Y_t \ge \beta + \lambda \bigg) \le \Pr\bigg( \sum_t \overbar{Y}_t \le \beta \textrm{ and } \sum_t \E_{t-1} \overbar{Y}_t \ge \beta + \lambda \bigg). \label{eq:freed2}
	\end{align}
	To justify this inequality, first notice that since the increments are non-negative $\sum_t \overbar{Y}_t \le \sum_{t \le \tau} Y_t$, and so for every scenario satisfying the LHS we have $\sum_t \overline{Y}_t \le \beta$, i.e., the first part on the RHS. Moreover, since $\ones(t \le \tau)$ and $\ones(t \le \tau')$ are $\cF_{t-1}$-measurable, $\E_{t-1} \overbar{Y}_t = \ones(t \le \tau) \cdot \,\ones(t \le \tau') \cdot \E_{t-1} Y_t$, and hence $\sum_t \E_{t-1} \overbar{Y}_t  = \sum_{t \le \min\{\tau,\tau'\}} \E_{t-1} Y_t$; and for every scenario satisfying the LHS we have $\sum_{t \le \tau} \E_{t-1} Y_t \ge \beta + \lambda$, and thus $\tau' \le \tau$ by definition of $\tau'$, so using the previous equation we see that the second part of the RHS is also satisfied, thus proving \eqref{eq:freed2}.	
	
	Now define the martingale difference sequence $M_t := \overbar{Y}_t - \E_{t-1} \overbar{Y}_t$ and notice that
	\begin{align}
	\Pr\bigg( \sum_t \overbar{Y}_t \le \beta \textrm{ and } \sum_t \E_{t-1} \overbar{Y}_t \ge \beta + \lambda \bigg) \le \Pr\bigg(\sum_t M_t \le -\lambda \bigg). \label{eq:freed4} 
	\end{align}
 Moreover, because of the truncation we can bound the predictable variance of this martingale:  we have $|\overbar{Y}_t| \le 1$ and hence
	\begin{align*}
		\E_{t-1} M^2_t = \Var(\overbar{Y}_t \mid \cF_{t-1}) \le \E_{t-1} (\overbar{Y}_t)^2 \le \ones(t \le \tau') \cdot \E_{t-1} Y_t.
	\end{align*}
	So using the definition of the stopping time $\tau'$ and the fact $\E_{t-1} Y_t \le 1$, the predictable variance is $V := \sum_t \E_{t-1} M^2_t \le \sum_{t \le \tau'} \E_{t-1} Y_t \le \beta + \lambda + 1$ always.
	
	Finally, applying Lemma \ref{lemma:Freedman} to $(M_t)_t$ with $v = \beta + \lambda + 1$ we get 
	\begin{align*}
	\Pr\bigg(\sum_t M_t \le - \lambda \bigg) = 
    \Pr\bigg(\sum_t M_t \le - \lambda \textrm{ and } V \le \beta + \lambda + 1\bigg) \le \exp\bigg(- \frac{\lambda^2}{2(\beta + \lambda + 1 + \lambda/3)}\bigg) \le e^{- \frac{3}{14} \lambda},
	\end{align*} 
	where last inequality uses the assumption $\lambda \ge \beta + 1$. Chaining this bound with inequalities \eqref{eq:freed2} and \eqref{eq:freed4} concludes the proof of the theorem.
	\end{proof}		


\subsection{Proof of \Cref{lem:p-bound-doubling}} \label{app:doubling}

    Let $\cA_{base}$ denote the $\beta$-competitive algorithm assumed in the statement of the lemma. We run $\cA_{base}$ in phases: In the first phase we run the $\beta$-competitive algorithm with the starting estimate $\widehat{\OPT}$ being the minimum cost required to schedule the first job. Each time the the hindsight optimum of the problem for the jobs seen thus far exceeds $\widehat{\OPT}$, we multiply $\widehat{\OPT}$ by $2^p$, restart the algorithm (feeding it all jobs up until the current time) and use its decisions until the next phase. 

    Suppose there are $k$ phases, and let $\cI_i$ denote the jobs in phase $i$, and $x^i$ be the allocation of the jobs in this phase by the algorithm. The total cost of the algorithm is then $f(\Lambda(x^1 + \ldots + x^k))$. Since $f$ is $p$-subadditive and $\Lambda$ is subadditive (by subadditivity of norms), we have
    \begin{align}
     \textrm{cost of the full algorithm} = f(\Lambda(x^1 + \ldots + x^k)) \le \Big( f(\Lambda(x^1))^{1/p} + \ldots + f(\Lambda(x^k))^{1/p} \Big)^p.  \label{eq:redEstimate}
    \end{align}
    
    To bound the per-phase costs $f(\Lambda(x^i))$ on the RHS, note that by construction during phase $i$ the estimate $\widehat{\OPT}$ satisfies $\OPT(\cI_1 \cup \ldots \cup \cI_i) \le \widehat{\OPT} \le 2^p \cdot \OPT(\cI_1 \cup \ldots \cup \cI_i)$, so the $\alpha$-competitiveness of $\cA_{base}$ holds; thus, the schedule $x^i$ (which is actually only for the jobs in $\cI_i$) then has cost $f(\Lambda(x^i)) \le \beta \cdot \OPT(\cI_1 \cup \ldots \cup \cI_i)$; since these optima at grow by at least a factor of $2^p$ in each phase, we have that $\OPT(\cI_1 \cup \ldots \cup \cI_i)$ is at most $2^{-p(k-i)}$ times the optimum of the last phase, namely $\OPT_{GS}$. Employing these observations on \eqref{eq:redEstimate} we have that the cost of the full algorithm is at most 
    \begin{align*}
    \Big(\sum_{i=1}^k (\beta 2^{-p(k-i)} \OPT_{GS})^{1/p} \Big)^p \le \beta \OPT_{GS} \cdot 2^p,
    \end{align*}
    and the algorithm is $\beta 2^p$-competitive as claimed. 


\subsection{Reduction from Packing to Covering for Norms}\label{app:normReduction}

We first use the following technical lemma.
\begin{lemma}\label{lemma:det-to-rand-pack}    
    Suppose for some $\alpha \in (0,1)$, the problem \ref{eq:packGenSched}$_{f, \{\|\cdot\|\}}$ is $(\alpha, c)$-solvable, where $f$ is subadditive.
    
    Suppose also that $D$ is a distribution over instances such that $\E_{I \sim D} \Opt(I) \leq L$ for some $L \geq 2$. Then there is an online algorithm which takes $I \sim D$ and the value $L$ and outputs a solution $x$ such that 
    \begin{enumerate}
        \item $\E \left[\Opt(I) - \sum_{i,j,i} x_{ijk}\right] \leq (1 - \frac{\alpha}{4}) L$, i.e. our bound on unpacked elements is cut down by a constant.
        \item $\E f(\Lambda(x)) \leq 5cB$, i.e. we expand our budget by only a constant in expectation.
    \end{enumerate}
\end{lemma}
\begin{proof}
    First, we let $M = \floor{\frac{L}{2}}$. Our algorithm on $I$ will be as follows. 
    
    As each $j$ arrives online, we construct a series of subinstances $I_1, I_2, \dots$ of \ref{eq:packGenSched} as follows: First, initialize $\ell = 1$, and add incoming requests $j$ to subinstance $I_\ell$. When each request $j$ arrives, we calculate $\Opt(I^{\leq j})$, i.e. the optimum for the global instance $I$ restricted to observed arrivals. If we observe $\Opt(I^{\leq j}) = \ell M$, then after assigning $j$, increment $\ell$ by 1 (i.e. send future requests to the next instance in the sequence). Notice that $\Opt(I^{\leq j})$ cannot ``skip'' a checkpoint $\ell M$ since $\ell M$ is an integer, and each request increases the optimum by at most 1. 
    
    On each subinstance $I_\ell $, we run our black-box $\alpha$-competitive algorithm using estimate $M$ to get a partial assignment $x^\ell $. Our total assignment is then $x = \sum_\ell  x^\ell $. 
    
    Let $R = \floor{\frac{\Opt(I)}{M}}$. For each $\ell  \leq R$, we have $\Opt(I_\ell ) \geq M$ by construction, so our black-box algorithm is guaranteed to obtain $\E[\sum_{i,j,k} x_{ijk}^\ell ] \geq \alpha M$ on these instances. Therefore, conditional on the value of $R$, our algorithm obtains objective $\E[\sum_{i,j,k} x_{ijk} \mid R] \geq R \alpha M$.
    
    Using this fact, we will verify the two claims of \Cref{lemma:det-to-rand-pack}.

    \begin{enumerate}
        \item Notice that if $R \geq 1$, then $\Opt(I) \leq (R+1)M \leq 2RM$. Thus, 
        $$\E_I[\sum_{i,j,k} x_{ijk}] \geq \E_I[R\alpha M] \geq \frac{\alpha}{2}\E_I[\ind_{R \geq 1} \cdot \Opt(I)] = \frac{\alpha}{2}\E_I[\ind_{\Opt(I) \geq M} \cdot \Opt(I)].$$

        With this, we can calculate our desired bound as follows.
        \begin{align*}
            \E_I[\Opt(I) - \sum_{i,j,k} x_{ijk}] 
            &\leq \E_I \left[\Opt(I) - \frac{\alpha}{2}\ind_{\Opt(I) \geq M} \cdot \Opt(I)\right]\\
            &= \E_I \left[\left(1 - \frac{\alpha}{2}\right) \cdot \Opt(I) + \frac{\alpha}{2}\ind_{\Opt(I) < M} \cdot \Opt(I)\right]\\
            &\leq (1- \frac{\alpha}{2})L + \frac{\alpha}{2}M ~\leq~ (1 - \frac{\alpha}{4})L.
        \end{align*}
        
        \item Notice that $f(\Lambda(x^\ell )) \leq cB$ for each $\ell  \in \{1, \dots, R+1\}$, so $f(\Lambda(x)) \leq (R+1)cB$. Additionally, we have
        $$
        \E[R] \leq \E\left[\frac{\Opt(I)}{M}\right] \leq \frac{L}{M} \leq 4.
        $$
        The last inequality is because $M = \floor{\frac{L}{2}} \geq \frac{L}{4}$ for $L \geq 2$. Thus, by subadditivity, $\E f(\Lambda(x)) \leq \E[(R+1)cB] \leq 5cB$.
    \end{enumerate}
    
\end{proof}

\begin{proof}[Proof of Theorem \ref{thm:packCovReduction} for subadditive $f$]
    Suppose we have an instance of \ref{eq:covGenSched}$_{f, \{\|\cdot\|\}}$ with optimum $\Opt_{OGS}$, where $f$ is subadditive. Using \Cref{lem:p-bound-doubling}, we may lose a factor $2$ in our competitive ratio to assume that we have knowledge of a value $B$ such that $\Opt_{OGS} \leq B \leq 2 \Opt_{OGS}$ (i.e., each time the hindsight optimum increases above our guess $B$, double $B$ and restart the algorithm). With this value of $B$, we can construct and instance $I$ of \ref{eq:packGenSched}$_{f, \{\|\cdot\|_i\}_i}$ corresponding to the \ref{eq:covGenSched} instance, so that $\Opt(I) = n$.

    Now, with \Cref{lemma:det-to-rand-pack}, we construct a sequence of algorithmic agents $A_0, A_1, \dots, A_N$, where $N$ will be chosen later. Each $A_k$ uses the online algorithm specified by \Cref{lemma:det-to-rand-pack} with value $L = L_k := n(1-\frac{\alpha}{4})^k$. 
    
    Overall algorithm for OGS will be as follows. Let $A_0$ operate on input $I_0 := I$. Each job $j$ which is not assigned by $A_0$ is then considered part of input $I_1$, and is given to $A_1$. This continues, i.e. for each $k \in [N]$, the input $I_k$ to algorithm $A_k$ consists of arrivals which are left unassigned by $A_{k-1}$. Finally, any arrivals which are unassigned by $A_N$ are assigned greedily to the machine which results in the smallest marginal increase in $f(\Lambda)$.

    Note that by induction using \Cref{lemma:det-to-rand-pack}, we have that $\E \Opt(I_k) \leq n(1-\frac{\alpha}{4})^k$. Hence, we can choose $N = \ceil{\log n \cdot \log(1-\alpha/4)} - 1 = \Theta(\log n / \alpha)$, so that the expected number of unassigned jobs is $n(1-\frac{\alpha}{4})^{N+1} \leq 2$. Finally, this gives 
    \[ \E f(\Lambda(x)) \leq 5cBN + 2B \leq (10N + 4) \cdot \Opt_{OGS} = O(c\log n / \alpha) \cdot \Opt_{OGS}. \qedhere \]
\end{proof}

